\documentclass[11pt]{article}
\usepackage{fullpage}
\usepackage[hyphens]{url}
\usepackage{hyperref}
\hypersetup{
breaklinks=true,
    colorlinks=true, 
    linkcolor=black, 
    citecolor=black, 
    filecolor=black,
    urlcolor=black 
}
\usepackage[utf8]{inputenc} 
\usepackage[T1]{fontenc}    
\usepackage{booktabs}       
\usepackage{amsfonts}       
\usepackage{nicefrac}       
\usepackage{microtype}      
\usepackage{diagbox}
\usepackage{enumerate}
\usepackage[shortlabels]{enumitem}
\usepackage{subfigure}
\usepackage{tabularx}
\usepackage{verbatim}
\usepackage{afterpage}
\usepackage{float}
\usepackage{wrapfig}

\usepackage{cancel}
\usepackage[normalem]{ulem}
\usepackage{tcolorbox}
\tcbuselibrary{skins, breakable}

\usepackage{graphicx} 
\usepackage{caption}

\usepackage{mathrsfs}
\usepackage{amsmath}
\usepackage{amsthm}
\usepackage{amssymb}
\usepackage{tablefootnote}
\usepackage{multirow}
\usepackage{enumerate}
\usepackage{color}
\usepackage{xcolor}

\usepackage[numbers]{natbib}

\NewDocumentCommand{\new}{+m}{%
  {\begingroup\color{black}#1\endgroup}
}

\usepackage{algorithm}
\usepackage{algpseudocode}
\allowdisplaybreaks[4]
\usepackage{bm,todonotes}
\allowdisplaybreaks

\newtheorem{theorem}{Theorem}[section]

\newtheorem{thm}{Theorem}[section]
\newtheorem{lem}{Lemma}[section]
\newtheorem{cor}{Corollary}[section]
\newtheorem{prop}{Proposition}[section]
\newtheorem{asmp}{Assumption}[section]
\newtheorem{defn}{Definition}[section]

\newcounter{subassumption}[asu]

\makeatletter
\renewcommand{\p@subassumption}{\theasu}
\makeatother

\newtheoremstyle{remarkstyle}
  {}                    
  {}                    
  {\normalfont}         
  {}                    
  {\itshape}            
  {.}                   
  { }                   
  {}                    
\theoremstyle{remarkstyle}
\newtheorem{rem}{Remark}[section]

\definecolor{wjs}{RGB}{200,0,50}

\definecolor{hyw}{RGB}{153,000,000}

\hypersetup{
colorlinks=true,
filecolor=black,
citecolor=blue,
urlcolor=black,
}

\usepackage{amsmath,amsfonts,bm}

\def\ceil#1{\lceil #1 \rceil}
\def\floor#1{\lfloor #1 \rfloor}
\def\1{\bm{1}}

\def\eps{{\varepsilon}}

\def\epso{{\varepsilon_0}}
\def\deltao{{\delta_0}}

\def\rd{{\textnormal{d}}}
\def\re{{\textnormal{e}}}

\def\rp{{\textnormal{p}}}

\def\HC{\mathrm{HC}}

\DeclareMathAlphabet{\mathsfit}{\encodingdefault}{\sfdefault}{m}{sl}
\SetMathAlphabet{\mathsfit}{bold}{\encodingdefault}{\sfdefault}{bx}{n}

\renewcommand{\xi}{\zeta}

\def\0{{\bf 0}}
\def\1{{\bf 1}}

\def\AM{{\mathcal A}}

\def\GM{{\mathcal G}}
\def\FM{{\mathcal F}}

\def\NM{{\mathcal N}}

\def\PM{{\mathcal P}}
\def\RM{{\mathcal R}}
\def\SM{{\mathcal S}}

\def\UM{{U}}

\def\WM{{\mathcal W}}

\def\RB{{\mathbb R}}
\def\FB{{\mathbb F}}

\DeclareMathOperator{\EB}{\mathbb{E}}

\def\PB{{\mathbb P}}

\def\FPM{\PM_{\Delta}}

\newcommand{\KL}{D_{\mathrm{KL}}}
\newcommand{\Var}{\mathrm{Var}}

\def\rd{{\mathrm{d}}}
\def\re{{\mathrm{e}}}

\def\Algo{{{Tr-GoF}}}

\def\Ent{{\mathrm{Ent}}}

\def\Key{{\mathtt{Key}}}

\def\Simplex{{\mathrm{Simp}}}

\def\token{{w}}
\def\ttoken{\widetilde{w}}

\def\tn{{n_0}}

\def\Voca{{\WM}}
\newcommand\bP{\bm{P}}
\newcommand\bQ{\bm{Q}}

\def\SMmax{\SM^{\mathrm{gum}}}

\def\ars{{\mathrm{ars}}}

\def\Yars{Y}

\def\Sars{T^{\mathrm{ars}}}

\def\hars{h_{\mathrm{ars}}}
\def\hlog{h_{\mathrm{log}}}
\def\hind{h_{\mathrm{ind},\delta}}

\def\hlog{h_{\mathrm{log}}}

\def\hoptars{{h_{\mathrm{opt}, \Delta}}}
\def\hoptarso{{h_{\mathrm{opt}, \Delta_0}}}
\def\bfPt{{\overline{f}_{1,t}}}
\def\bfPn{{\overline{f}_{1,n}}}
\def\bmun{{\overline{\mu}_{1,n}}}
\def\bmut{{\overline{\mu}_{1,t}}}

\def\BFB{\overline{\FB}}

\title{Robust Detection of Watermarks for Large Language Models Under Human Edits}
  
\author{
{Xiang Li\thanks{University of Pennsylvania; Email: \texttt{lx10077@upenn.edu}. } } 
\and
{Feng Ruan\thanks{Northwestern University; Email: \texttt{fengruan@northwestern.edu}. }} 
\and
{Huiyuan Wang\thanks{University of Pennsylvania; Email: \texttt{huiyuanw@upenn.edu}. }  }  
\and
{Qi Long\thanks{University of Pennsylvania; Email: \texttt{qlong@upenn.edu}. Joint corresponding author.} }
\and
{Weijie J.\ Su\thanks{University of Pennsylvania; Email: \texttt{suw@wharton.upenn.edu}. Joint corresponding author.}}
}

\date{November 21, 2024}

\begin{document}

\maketitle

\begin{abstract}

Watermarking has offered an effective approach to distinguishing text generated by large language models (LLMs) from human-written text. However, the pervasive presence of human edits on LLM-generated text dilutes watermark signals, thereby significantly degrading detection performance of existing methods. In this paper, by modeling human edits through mixture model detection, we introduce a new method in the form of a truncated goodness-of-fit test for detecting watermarked text under human edits, which we refer to as \Algo. We prove that the \Algo~test achieves optimality in robust detection of the Gumbel-max watermark in a certain asymptotic regime of substantial text modifications and vanishing watermark signals. Importantly, \Algo~achieves this optimality \textit{adaptively} as it does not require precise knowledge of human edit levels or probabilistic specifications of the LLMs, in contrast to the optimal but impractical (Neyman--Pearson) likelihood ratio test. Moreover, we establish that the \Algo~test attains the highest detection efficiency rate in a certain regime of moderate text modifications. In stark contrast, we show that sum-based detection rules, as employed by existing methods, fail to achieve optimal robustness in both regimes because the additive nature of their statistics is less resilient to edit-induced noise. Finally, we demonstrate the competitive and sometimes superior empirical performance of the \Algo~test on both synthetic data and open-source LLMs in the OPT and LLaMA families.

\end{abstract}

\newtcolorbox{gptbox}{
    colback=blue!5!white, colframe=blue!75!black,
    rounded corners, boxrule=0.4mm, width=0.9\textwidth,
    before skip=2pt, after skip=2pt, boxsep=2pt,
    left=2mm, right=2mm, 
    enhanced
}

\newtcolorbox{userbox}{
    colback=gray!10!white, colframe=black!75!black,
    rounded corners, boxrule=0.4mm, width=0.9\textwidth,
    before skip=2pt, after skip=2pt, boxsep=2pt,
    left=2mm, right=2mm, 
    enhanced
}

\section{Introduction}\label{sec:intro}
Large language models (LLMs) have recently emerged as a transformative technique for generating human-like text and other media~\citep{touvron2023llama,openai2023,achiam2023gpt}. While this advancement boosts productivity across various industries, it also introduces risks related to the ownership and creation of content. These risks include the spread of misinformation \citep{zellers2019defending,weidinger2021ethical,starbird2019disinformation}, challenges to education and academic integrity \citep{stokel2022ai,milano2023large}, and issues concerning data authenticity \citep{radford2023robust,shumailov2023curse}. These problems highlight the urgent need for methodologies to authenticate and verify the origin of text, specifically, determining whether it is generated by LLMs or humans.

Watermarking text during generation by LLMs has offered a principled and viable approach to resolving these issues \citep{kirchenbauer2023watermark,scott2023watermarking,christ2023undetectable}. Since 2023, a variety of watermarking schemes have been introduced \citep{fernandez2023three,kuditipudi2023robust,hu2023unbiased,wu2023dipmark,zhao2024permute,zhao2024provable,liu2024adaptive,giboulot2024watermax,fu2024gumbelsoft,Dathathri2024,xie2024debiasing}. Loosely speaking, a watermarking scheme embeds signals into the process of generating text by using controlled pseudorandomness. This is made possible by the probabilistic nature of LLMs through next-token prediction (NTP) for sequentially generating text. The watermark signals are nearly unnoticeable, if possible, to human readers but are provably detectable once a verifier has knowledge of how the pseudorandomness is constructed \citep{christ2023undetectable}.

\begin{figure}[!htp]
\centering
\begin{flushright}
\begin{userbox}
\begin{minipage}[t]{\textwidth} 
\vspace{-0.01in} 
\textbf{Prompt:} How can statistics benefit the study of watermarks for LLMs?
\end{minipage}
\end{userbox}
\end{flushright}
\vspace{-0.3in}
\begin{flushleft}
\begin{gptbox}
\begin{minipage}[t]{0.05\textwidth} 
\centering
\includegraphics[width=6mm]{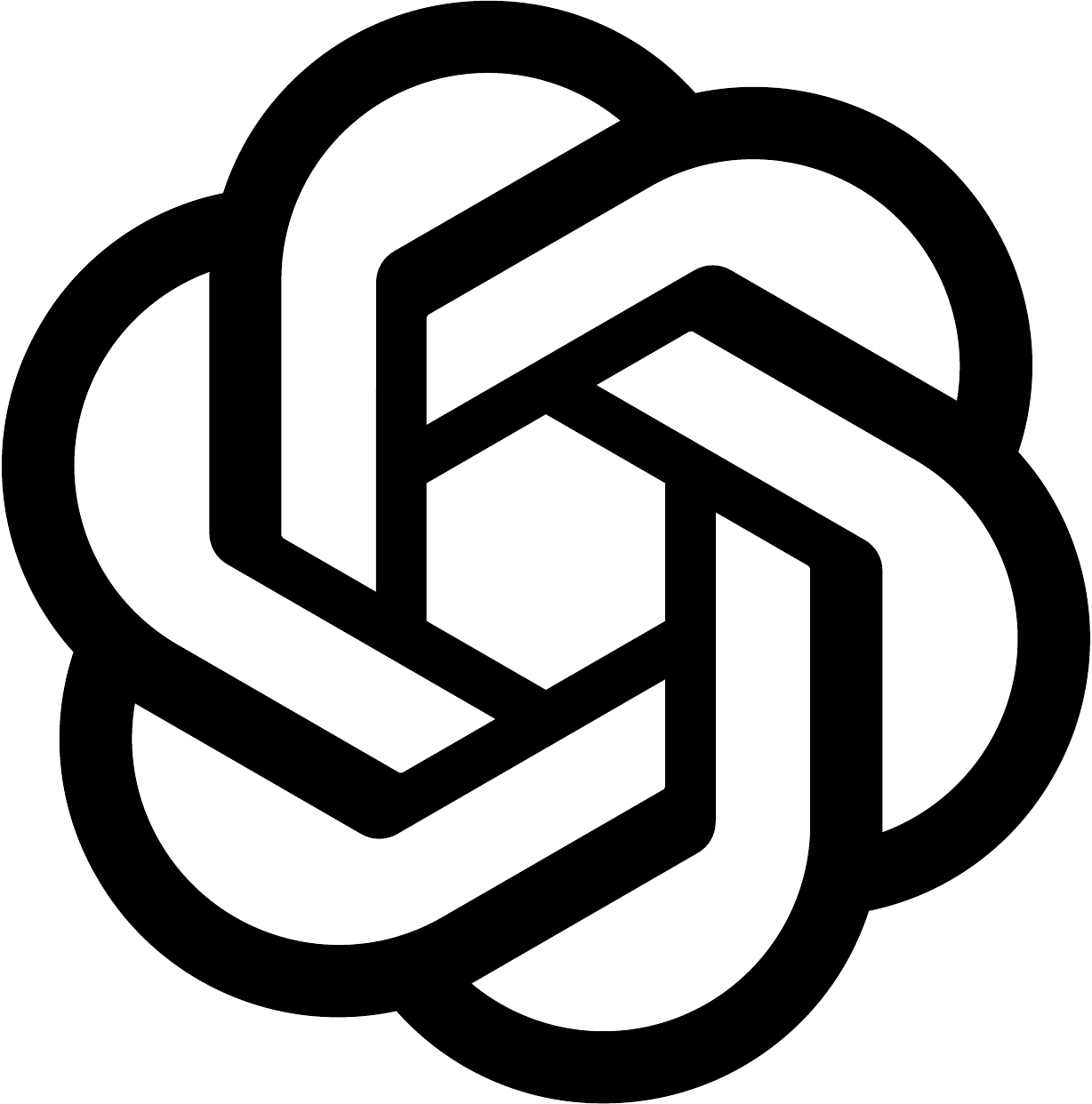} 
\end{minipage}
\begin{minipage}[t]{0.94\textwidth} 
\vspace{-0.2in} 
\textbf{ChatGPT:} Statistics serves as an invaluable conduit for discerning the intricate interplay of resilience, detectability, and adaptive robustness in watermarking frameworks tailored for large language models, thereby anchoring a profound understanding of their operational fidelity across diverse linguistic landscapes.
\end{minipage}
\end{gptbox}
\end{flushleft}
\vspace{-0.3in}
\begin{flushright}
\begin{userbox}
\begin{minipage}[t]{\textwidth} 
\vspace{-0.0in} 
\textbf{Human edits:} Statistics serves as \textcolor{orange}{\underline{a valuable tool}} for \textcolor{orange}{\underline{examining}} the \sout{intricate} interplay of resilience, detectability, and \sout{adaptive} robustness \sout{in} \textcolor{blue!60!black}{\underline{of}} watermarking frameworks \sout{tailored} for large language models, \textcolor{orange}{\underline{helping}} \sout{anchoring} \textcolor{blue!60!black}{\underline{to build}} a \textcolor{orange}{\underline{clearer}} understanding of their \textcolor{orange}{\underline{reliability}} across \textcolor{orange}{\underline{different}} linguistic \textcolor{orange}{\underline{contexts}}.
\end{minipage}
\end{userbox}
\end{flushright}
\caption{How users modify ChatGPT's response through edits such as substitution, deletion, and insertion. \textcolor{orange}{\underline{Orange underline}} marks substitutions, \sout{black strikethrough} indicates deletions, and \textcolor{blue!60!black}{\underline{blue underline}} highlights insertions.
}
\label{fig:human-edit}
\vspace{-0.1in}
\end{figure}

In real-world scenarios, however, text generated from LLMs often undergoes various forms of human edits before its use (see Figure \ref{fig:human-edit} for an illustration), which presents perhaps the most significant challenge for applying watermarks \citep{kirchenbauer2023reliability,nature_editorial}. In the simplest case, a user of LLMs might replace several words of the generated text, either to improve the text from the user's perspective or make it less LLM-sounding. This simple editing process would weaken the watermark since the signals from the modified tokens---the smallest unit in text generation, which can be a word or a sub-word---turn to noise. Worse, how the editing process proceeds is unknown to the verifier. For example, it is not clear if a token in the text is generated by the LLM or has been modified by a human. In light of this, a watermark detection rule should not only be efficient in the edit-free scenario but also, perhaps more importantly, its detection power should not be unduly affected by various forms of human edits.
Equally important is adaptivity: the rule should perform well without prior knowledge of how or to what extent the text has been edited, ensuring practical applicability in handling diverse and unpredictable editing behaviors.

To get a handle on how robust and adaptive existing watermark detection rules are against human edits, we conduct a numerical experiment and present the results in Figure~\ref{fig:intro}. The experiment is concerned with the Gumbel-max watermark \citep{scott2023watermarking}, which has been implemented internally at OpenAI and is the first unbiased\footnote{A watermarking scheme is (approximately) unbiased if the watermarked LLM generates each token with (approximately) the same probability distribution as the unwatermarked counterpart.} watermark. With 5\% of the tokens modified by humans, the detection power when text length is 400 drops from 87.8\% in the edit-free case to 64.7\% with paraphrase edits and further to 30.2\% with adversarial edits. While this significant performance degradation is concerning, it is not surprising as these detection rules \citep{li2024statistical,scott2023watermarking} are developed without any robustness consideration, in particular, assuming there is no signal corruption due to human edits. More critically, the editing process would render the computation of pseudorandomness, which takes as input the before-edited tokens \citep{kirchenbauer2023watermark,scott2023watermarking}, impossible for detecting watermarks.

\begin{figure}[!t]
\vspace{-0.1in}
\centering
\includegraphics[width=0.58\textwidth]{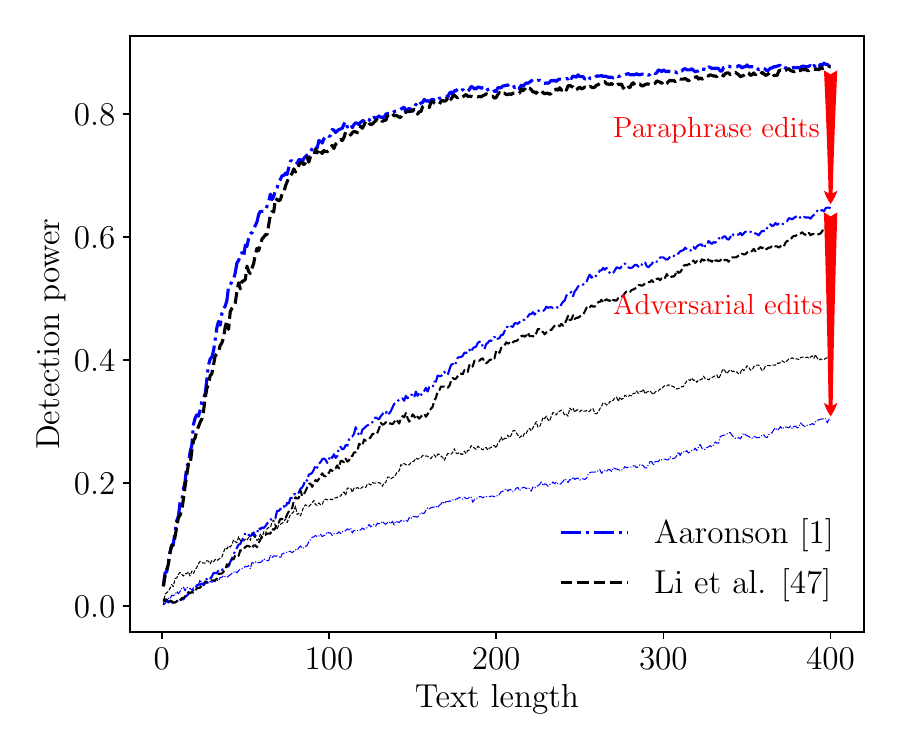}
\vspace{-0.2in}
\caption{
Empirical study of detection methods for the Gumbel-max watermark. Statistical power is evaluated at a 0.01 significance level using 1,000 prompts from the C4 dataset~\citep{raffel2020exploring} and 400 generated tokens per prompt with the OPT-1.3B model~\citep{zhang2022opt} (temperature 0.3). Paraphrase edits randomly replace 5\% of tokens with WordNet synonyms~\citep{miller1995wordnet}, while adversarial edits replace the 5\% with the strongest watermark signals. Thinner curves indicate higher levels of human edits. See Supplementary~\ref{exp:detials} for details.
}
\label{fig:intro}    
\vspace{-0.1in}
\end{figure}

Thus, there is a pressing need to develop robust and adaptive methods for detecting watermarks in LLM-generated text that may undergo human edits. Existing work on watermark detection typically assumes a scenario where each token contributes a signal to the watermark. For example, Li et al. \citep{li2024statistical} introduced a statistical framework for evaluating detection efficiency based on sum-based statistics that measure the cumulative watermark strength across all tokens. However, this framework operates under the ideal assumption that tokens are either fully human-written or entirely generated by an LLM. When applied to human-edited text---a combination of human- and LLM-generated tokens---sum-based statistics may be easily corrupted by noise, leading to substantial degradation in efficiency, as illustrated in Figure~\ref{fig:intro}.

\subsection{Our Contributions} In this paper, we address the robustness of detecting LLM watermarks from a statistical viewpoint, providing both a statistical framework and methodologies with theoretical guarantees alongside empirical validation.

\new{
A central challenge in robust watermark detection is understanding how human edits affect detection performance. Although the exact editing process cannot be recovered from the final modified text, a key insight is that—due to the indistinguishability between pseudorandomness and true randomness—each token either contributes a watermark signal or acts as noise, depending on whether it has been altered by human edits.
In the framework of Li et al.~\cite{li2024statistical}, the noise follows a fixed and known distribution, while the signal depends on the token distribution predicted by the LLM. We refer to this predictive multinomial distribution as the NTP distribution, denoted by $\bP_t = (P_{t,\token})_{\token \in \Voca}$, where $\Voca$ is the vocabulary used for predicting the $t$-th token. For instance, GPT-2/3.5 models use a vocabulary of size 50,257 \citep{radford2019language,brown2020language}, while the LLaMA-2 models use 32,000 tokens \citep{touvron2023llama}.
Motivated by this binary behavior, we model the watermark signal across all tokens as a mixture: with probability $\eps_n$, a token carries signal; otherwise, it contributes noise. This leads to a natural statistical question: under what conditions can we reliably distinguish this mixture (representing watermarked text potentially modified by humans) from a pure noise distribution (representing entirely human-written text)? We address this question in the following result, illustrated, for example, in the case of the Gumbel-max watermark. (We write $a_n \asymp b_n$ to mean there exist constants $C_1$ and $C_2$ such that $C_1 a_n \le b_n \le C_2 a_n$ for all $n$):
}


\vspace{-1.1em} 
\paragraph{Phase transition for watermark detection.} We consider an asymptotic regime where the text length $n$ tends to infinity, but individual tokens contain less watermark signal. Specifically, we assume $\eps_n \asymp n^{-p}$ and $1 - \max_{\token \in \Voca} P_{t,\token} \asymp n^{-q}$ for constants $0 < p, q < 1$, inspired by the sparse mixture detection problem \cite{donoho2004higher}. Detection becomes more challenging as either $p$ or $q$ increases. A larger value of $p$ means fewer tokens contribute to the watermark after human editing, resulting in a diluted watermark; similarly, a larger value of $q$ indicates that the NTP distributions are less regular due to being approximately degenerate, which weakens watermark signals (see elaboration in Section~\ref{sec:detectability}). Our finding is that the optimal detection region in the $(p, q)$ plane has boundary $q + 2p = 1$. Explicitly, when $q + 2p > 1$, no test can asymptotically achieve zero Type I and Type II errors for the Gumbel-max watermark; when $q + 2p < 1$, in stark contrast, both Type I and Type II errors approach zero using the likelihood-ratio test.

While the likelihood-ratio test achieves theoretical optimality, it requires knowledge of the watermark fraction $\eps_n$ and all NTP distributions $\bP_t$ for $t = 1, \ldots, n$. To make this boundary practically relevant, it is essential to develop a practical method that can reliably separate the null and alternative distributions as $n \rightarrow \infty$. For this purpose, it is instructive to examine the statistical limits of sum-based test statistics as employed in \cite{li2024statistical,kuditipudi2023robust,fernandez2023three}. We prove that these sum-based detection rules fail to achieve the optimal boundary $q + 2p = 1$ in general, and more precisely, the best possible detection boundary these methods can achieve is $q + p = 1/2$.

This analysis offers insights into the causes of this suboptimality, which we leverage to develop a new, practically adaptive method that achieves the optimal boundary:

\vspace{-1.1em} 
\paragraph{An adaptively optimal method.} Our approach utilizes the empirical cumulative distribution function of p-values from the text under detection, comparing it with the null counterpart through a form of divergence \citep{jager2007goodness}. Unlike sum-based approaches \cite{li2024statistical,kuditipudi2023robust,fernandez2023three}, this goodness-of-fit test adaptively identifies areas of significant departure between null and alternative distributions without requiring prior knowledge of the watermark fraction or NTP regularity. Moreover, a crucial innovation is to truncate the search for the location to filter out extreme scores for stability. This gives a family of truncated goodness-of-fit tests, which we refer to as \Algo~for short (details in Algorithm \ref{alg:detection}). We show that \Algo~achieves the robust detection boundary $q + 2p = 1$ without relying on any tuning parameters, thereby possessing adaptive optimality \cite{donoho2004higher,donoho2015higher}. Figure \ref{fig:intro-better-result} shows the robustness of \Algo~in detecting the Gumbel-max watermark in human-edited text.

\begin{figure}[t!]
\centering
\includegraphics[width=\textwidth]{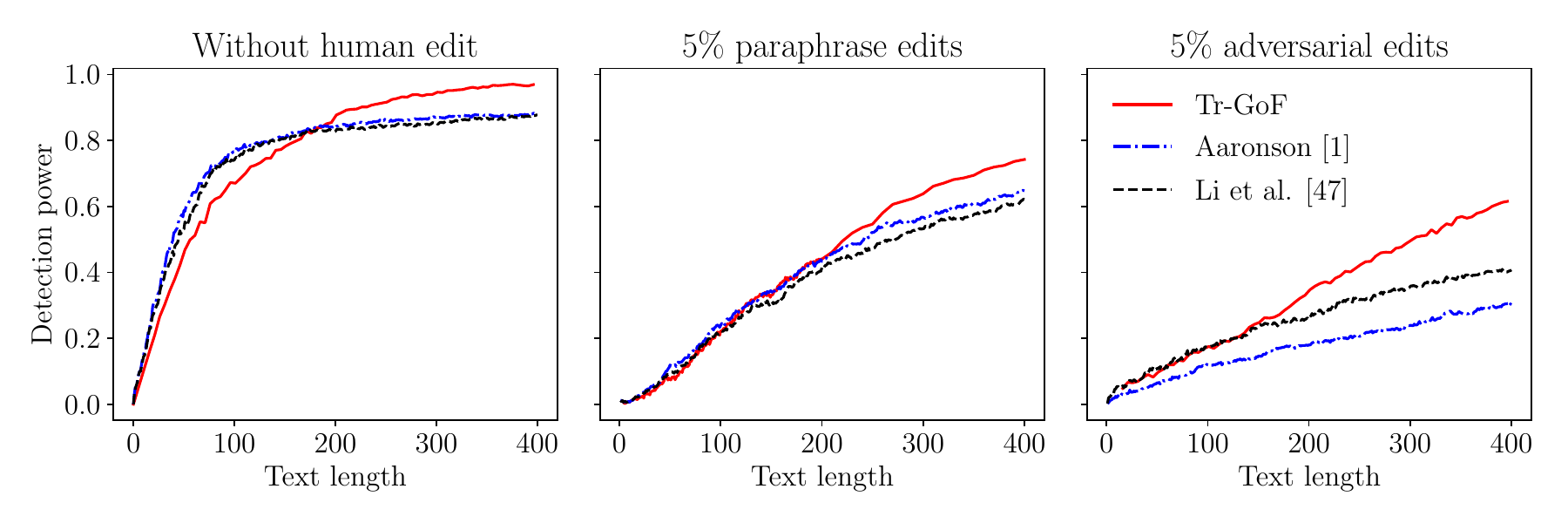}
\vspace{-0.2in}
\caption{
\Algo~shows improved detection efficiency under the same setup as Figure \ref{fig:intro}, with a more pronounced performance advantage in adversarial edits compared to paraphrase edits.
}
\label{fig:intro-better-result} 
\end{figure}

Having established the optimality of \Algo~for sparse watermark detection, we further examine its performance in a regime of ``dense'' watermark signals. More precisely, we consider the setting where the watermark fraction $\eps_n \equiv \eps$ remains constant, and the regularity level of the NTP distribution satisfies $1 - \max_{\token \in \Voca} P_{t,\token} \ge \Delta$ for some constant $0 < \Delta < 1$. Below is our third main result in this paper:

\vspace{-1.5em}
\paragraph{\Algo~attains optimal efficiency rate.} We show that the \Algo~test achieves the highest class-dependent efficiency under the framework of watermarks for LLMs proposed by Li et al.~\cite{li2024statistical}. This optimality is adaptive in that it does not rely on any prior knowledge of $\Delta$ or $\eps$, whereas existing sum-based detection rules fail to achieve the same level of efficiency due to their lack of adaptivity for dense watermark detection (see Figure~\ref{fig:efficiency} for illustration). This finding is unexpected as the \Algo~test was originally motivated by sparse watermark detection. Taken together, our results show that the \Algo~test achieves optimality in both sparse and dense regimes of watermark detection.

\subsection{Related Work}

Most robust watermarking methods focus on algorithmic or cryptographic approaches. Algorithmic methods emphasize resilient watermarking designs \citep{yoo2023robust, kuditipudi2023robust, zhu2024duwak, hou2024semstamp, ren2023robust, giboulot2024watermax}, while cryptographic approaches protect watermark signals through robust pseudorandom codes \citep{golowich2024edit, christ2024pseudorandom}. However, rigorous statistical analysis of robustness in watermark detection remains limited, largely due to the absence of a cohesive statistical framework. Definitions of robustness also vary widely.

Our approach introduces a statistical perspective by modeling human edits as a mixture distribution, enabling the establishment of detection boundaries and class-dependent efficiency rates. Another statistical approach to robustness is the edit tolerance limit, which quantifies the maximum proportion of edits a detection method can withstand while remaining effective. Studies often apply this perspective to the green-red list watermark, initially proposed in \citep{kirchenbauer2023watermark, kirchenbauer2023reliability} and valued for its simplicity \citep{fernandez2023three, hu2023unbiased, wu2023dipmark, zhu2024duwak, liu2024a}. Zhao et al.~\cite{zhao2024provable} further improved its tolerance by introducing a fixed pseudorandom code, enhancing resilience against text modifications. Additional refinements of this approach are discussed in \citep{kuditipudi2023robust, cai2024towards, huang2023towards}.

Other forms of robustness include decoder-focused and cryptographic approaches. Zhao et al.~\citep{zhao2024permute} examines decoder functions that map logit vectors to token probabilities, proposing a provably robust decoder for enhanced resilience. In cryptographic studies, robustness is designed to make pseudorandom codes resistant to a fixed proportion of human edits \citep{golowich2024edit, christ2024pseudorandom}. Our statistical viewpoint complements these approaches, providing an orthogonal framework for assessing robustness.

Robust statistics, which traditionally addresses resilience to outliers and perturbations, provides a relevant foundation for our study \citep{huber1992robust, huber2011robust, diakonikolas2023algorithmic}. Our mixture formulation connects with Huber’s contamination model \citep{huber1992robust}; however, unlike in Huber’s model where the contaminated distribution is typically unknown, in watermark detection, this distribution is known. Additionally, our work is related to sparse mixture detection, first introduced in \citep{dobrusin1958statistical} to identify sparse signals, with further developments in mixtures model \citep{ingster1996some,donoho2004higher,tony2011optimal,cai2014optimal,arias2017distribution} and rare-signal models \citep{jin2016rare,ke2014covariance,jin2021optimal}. 
Robust watermark detection also diverges from traditional sparse detection due to the autoregressive nature of text generation, which results in time-varying signal (or NTP) distributions.
 While some methods address inhomogeneity using stationary Gaussian processes \citep{hall2008properties} or temporal correlations \citep{hall2010innovated}, they are inapplicable to text generation. Our theoretical analysis directly leverages the autoregressive nature to analyze inhomogeneity in watermark detection, as detailed in Section \ref{sec:theory}.

\subsection{Organization of the Paper}
The remainder of the paper is organized as follows. In Section~\ref{sec:pre}, we introduce basics of LLMs and the statistical framework of LLM watermarks \citep{li2024statistical}. In Section~\ref{sec:robust}, we develop the \Algo~test and establish its theoretical guarantees for robust detection in Section~\ref{sec:theory}. We evaluate the empirical performance of the \Algo~test in Sections~\ref{sec:simulation} and~\ref{sec:LLM-experiments}. We conclude in Section~\ref{sec:discuss} with a discussion of future research directions. Most technical proofs are provided in the appendix. \new{Source code is publicly available at \url{https://github.com/lx10077/TrGoF}.}

\section{Preliminaries}\label{sec:pre}

\new{
\paragraph{Watermarking protocol.} 
We consider a standard watermarking protocol involving three parties \citep{kuditipudi2023robust,xie2024debiasing}: the model provider, the user, and the verifier. For example, a student (the user) may use an LLM to assist with a homework assignment. The model provider embeds a statistical watermark into the generated text using a secret key $\Key$, which is shared with the verifier (e.g., a teacher), but not with the user. The user may modify the generated text before submitting it. The verifier, who observes only the final (possibly edited) text, attempts to determine whether it contains a watermark. Importantly, the verifier does not have access to the original prompt, the unedited generated text, or the internal NTP distributions used by the model.
This protocol unfolds over three stages: first, watermark embedding, where the model encodes a hidden signal into the generated text; second, human editing, where the user may modify the output; and finally, watermark detection, where the verifier analyzes the submitted text to detect any remaining watermark signal. 

\paragraph{Watermark embedding.}  
We begin by describing the embedding mechanism.  
LLMs, such as GPT models \citep{radford2019language,brown2020language}, generate text by sampling tokens autoregressively. At each position $t$, given the previously generated tokens $\token_{1:(t-1)} := \token_1 \token_2 \cdots \token_{t-1}$, the model computes a next-token prediction (NTP) distribution $\bP_t = (P_{t,w})_{w \in \Voca}$ and samples the next token $\token_t$ accordingly.
To embed a watermark, this sampling process is modified in a deterministic yet secret-key-dependent way. Specifically, at each step $t$, the model computes a pseudorandom variable $\xi_t = \AM(\token_{(t-m):(t-1)}, \Key)$, where $\AM$ is a cryptographic hash function, $\Key$ is the watermarking key, and $m$ is the context window size. The token is then deterministically selected as $\token_t = \SM(\bP_t, \xi_t)$ using a decoding function $\SM$. The variable $\xi_t$ is termed pseudorandom because it behaves like a random variable, yet can be deterministically reconstructed by the verifier from the observed tokens $\token_{1:n}$, the key $\Key$, and the hash function $\AM$.
For theoretical analysis, we model $\xi_t$ as i.i.d. samples from a reference distribution $\pi$, reflecting the cryptographic properties of $\AM$ \citep{barak2021book, schneier1996applied, wu2023dipmark, zhao2024permute}. A decoder $\SM$ is said to be unbiased if it preserves the original distribution in expectation, i.e., $\PB_{\xi \sim \pi}(\SM(\bP, \xi) = \token) = P_{\token}$ for any $\bP$ and token $\token$. By this definition, an unbiased decoder corresponds to standard multinomial sampling methods.

\paragraph{Watermark detection.}
Since our primary interest lies in designing and analyzing detection methods, we adopt the perspective of the verifier. With a slight abuse of notation, we refer to the observed (potentially edited) text still as $\token_{1:n}$.
During watermark detection, the verifier receives the text $\token_{1:n}$ and reconstructs the corresponding pseudorandom sequence $\zeta_{1:n} := \zeta_1 \zeta_2 \cdots \zeta_n$ using the shared key and hash function. If the text is entirely human-written, there should be no statistical dependence between $\token_{1:n}$ and $\zeta_{1:n}$, as human users do not know the key or hash function. In contrast, if the text is fully generated by a watermarked LLM, each token $\token_t$ is deterministically chosen based on $\zeta_t$ through the decoder function $\SM$.
This observation reduces watermark detection to a problem of testing for statistical dependence between the observed tokens and the pseudorandom variables. Li et al. \citep{li2024statistical} formalized this idea as a hypothesis testing framework and analyzed the effectiveness of various detection methods. The central tool in this framework is a pivotal statistic $\Yars_t = Y(\token_t, \zeta_t)$, where the function $Y(\cdot, \cdot)$ is designed so that the distribution of $\Yars_t$ is known whenever $\token_t$ and $\zeta_t$ are independent—regardless of the marginal distribution of $\token_t$.
In this way, under the null hypothesis $H_0$, where the text is entirely human-written, the statistics $\Yars_t$ are i.i.d. samples from a known reference distribution $\mu_0$. Under the alternative hypothesis $H_1$, where the text is fully generated by a watermarked LLM, each $\Yars_t = Y(\SM(\bP_t, \zeta_t), \zeta_t) \mid \bP_t$ follows a distribution that depends on the NTP distribution $\bP_t$, denoted by $\mu_{1, \bP_t}$. Watermark detection is thus formulated as the following hypothesis testing problem: \begin{align} \label{eq:original-detectopn-problem} H_0: Y_t \sim \mu_0~~\text{i.i.d.}~\text{for}~1 \le t \le n \quad\text{v.s.}\quad H_1: Y_t \mid \bP_t \sim \mu_{1, \bP_t}~\text{for}~1 \le t \le n. \end{align}

\paragraph{The Gumbel-max watermark: Decoder and detection.}
Our paper focuses on the Gumbel-max watermark, one of the most influential and the first unbiased watermark \citep{scott2023watermarking}, which has been implemented internally at OpenAI \citep{openai2024understanding} and serves as a baseline in many studies. This watermark utilizes the Gumbel-max trick, a sampling technique for multinomial distributions \citep{gumbel1948statistical, papandreou2011perturb, maddison2014sampling, jang2016categorical}. The trick ensures that $\arg\max_{\token \in \Voca} P_\token^{-1} \log U_\token$ follows the distribution $\bP = (P_\token)_{\token \in \Voca}$, where $\xi = (U_\token)_{\token \in \Voca}$ consists of $|\Voca|$ i.i.d. $U(0,1)$ random variables.
Scott Aaronson proposed the following unbiased decoder \citep{scott2023watermarking}:
\begin{equation}\label{eq:it}
\SMmax(\bP, \xi) := \arg\max_{\token \in \Voca} \frac{\log U_w}{P_w}.
\end{equation}
The pivotal statistic is $\Yars_t = Y(\token_t, \xi_t) = U_{t, \token_t}$, where $\xi_t = (U_{t, \token})_{\token \in \Voca}$ collects all pseudorandom numbers. The watermark is detected when the sum-based statistic $\Sars_n = \sum_{t=1}^n \hars(\Yars_t)$, with $\hars(y) = -\log(1-y)$, exceeds a set threshold.
The method works because, without a watermark, the $U_{t, \token_t}$'s are i.i.d. $U(0,1)$, so $\Yars_t \sim \mu_0 = U(0,1)$. In contrast, if the watermark is present, \eqref{eq:it} makes tokens with a larger pseudorandom number more likely to be selected.
Indeed, we have $Y_t \mid \bP_t \sim \mu_{1, \bP_t}$, where $\mu_{1, \bP}(Y \leq r) = \sum_{\token \in \Voca} P_\token r^{1/P_\token}$ for $r \in [0,1]$ \citep{li2024statistical}. This alternative distribution is distinct from $\mu_0$ unless $\bP$ is degenerate (i.e., concentrated on a single token).
{
Detection thus hinges on distinguishing the statistical gap between $\mu_0$ and $\mu_{1, \bP}$.} Beyond $\hars$, other scoring functions have been proposed. Examples include the log function $\hlog(y) = \log y$ \citep{kuditipudi2023robust, fernandez2023three}, and the optimal least-favorable score $\hoptars$ from \citep{li2024statistical} where $\Delta$ is a user-specified regularity parameter. The function $\hoptars$ is optimal under the assumption that each $\bP_t$ belongs to a $\Delta$-regular class, denoted $\FPM$, defined as: 
\begin{equation}
\label{eq:regular-set}
\FPM = \{\bP: \max_{\token \in \Voca} P_{\token} \le 1-\Delta \}.
\end{equation}
Li et al. \citep{li2024statistical} shows that $\hoptars$ achieves the fastest exponential decay of Type II error at a fixed significance level $\alpha$ when testing against the least favorable distribution in $\FPM$ (see Def.~\ref{def:efficiency}).}

\section{Method}\label{sec:robust}

\subsection{Problem Formulation}
\label{sec:setting}
\new{
In this section, we formally define the robust detection problem and then introduce our proposed test, \Algo, in Section~\ref{sec:GOF-test}. As outlined in Section~\ref{sec:pre}, the watermarking process proceeds in three stages: watermark embedding, human editing, and detection. The prior framework of Li et al.~\cite{li2024statistical} assumes no human edits and formulates detection as a binary hypothesis test \eqref{eq:original-detectopn-problem}, so it assumes that the entire text $\token_{1:n}$ is either fully human-written or generated by a watermarked LLM. However, this formulation does not capture the reality where users edit watermarked content.

To understand the effect of human edits, consider the pivotal statistic $\Yars_t := Y(\token_t, \xi_t)$, where both the token $\token_t$ and the pseudorandom variable $\xi_t$ are derived from the observed (possibly edited) text. Suppose the subtext $\token_{(t-m):t}$ is copied exactly from the original fully watermarked text. In this case, both $\token_t$ and $\xi_t$ remain unchanged, as $\xi_t = \AM(\token_{(t-m):(t-1)}, \Key)$ is purely determined by its preceding segment $\token_{(t-m):(t-1)}$. Therefore, $\Yars_t$ follows the alternative distribution $\mu_{1, \bP_t}$, which reflects the presence of a watermark signal.
However, if this $(m+1)$-length subtext has been edited—meaning it does not match any corresponding window in the original watermarked text—then either the recomputed $\xi_t$ or the observed $\token_t$ differs from the one originally used. Because any cryptographic hash function $\AM$ is highly sensitive to input changes, even a minor modification results in a statistically independent output. In this case, $\xi_t$ is no longer associated with the generation of $\token_t$, and the pair $(\token_t, \xi_t)$ is independent. By the definition of a pivotal statistic, this implies that $\Yars_t \sim \mu_0$, the null distribution.
These observations motivate a more flexible detection framework. 
Depending on the fraction of watermark signal—that is, the dependence between $\token_t$ and $\xi_t$—the distribution of $\Yars_t$ should be a mixture of $\mu_0$ and $\mu_{1, \bP_t}$. We therefore model the observed $\Yars_t$ as arising from a mixture distribution:
\begin{align}
\begin{split}
\label{eq:robust}
H_0:\Yars_t \sim \mu_0,~\text{i.i.d.}~\forall 1 \le t \le n
~\text{v.s.}~
H_1^{\mathrm{mix}}:\Yars_t\mid \bP_t \sim (1-\eps_n) \mu_0 + \eps_n  \mu_{1, \bP_t},~\forall 1 \le t \le n.
\end{split}
\end{align}
This $H_1^{(\mathrm{mix})}$ implies that, on average, an $\eps_n$-fraction of watermark signals remain intact and thus detectable. When $\eps_n \equiv 1$, the mixture model in \eqref{eq:robust} reduces to the original \eqref{eq:original-detectopn-problem}.
}

\subsection{The \Algo~Detection Method}
\label{sec:GOF-test}
In this section, we introduce our detection method, \Algo~(Algorithm \ref{alg:detection}), which we apply to the robust watermark detection problem \eqref{eq:robust}. 
To begin, we explain the rationale behind this approach. The null hypothesis $H_0$ in \eqref{eq:robust} assumes that all pivotal statistics $\Yars_{1:n}$ are drawn from the same distribution $\mu_0$. This indicates that the testing problem in \eqref{eq:robust} essentially assesses how well the entire data set $\Yars_{1:n}$ conforms to the null distribution $\mu_0$.
Roughly speaking, \Algo~rejects the null hypothesis $H_0$ when the deviation between the empirical distribution of $\Yars_{1:n}$ and the null distribution $\mu_0$ is sufficiently large. 

 We now introduce the deviation measure which is based on a specific $\phi$-divergence \citep{jager2007goodness}. 
 Without loss of generality, we assume the data used in \Algo~are all $\rp$-values, which is always feasible by transforming observations using the CDF of $\mu_0$.
 In the Gumbel-max watermark, the $\rp$-values are linear transformations of $\Yars_{1:n}$.
 Given that $\mu_0 = \UM(0, 1)$, the $\rp$-value for each observation is:
\[
\rp_t := \PB_{0}( Y \ge \Yars_t|\Yars_t) = 1- \Yars_t.
\]
Let $\FB_n(r)$ denote the empirical distribution of the $\rp$-values,
\[
\FB_n(r) = \frac{1}{n} \sum_{t=1}^n \1_{\rp_t \le r}~~\text{for}~~r \in [0, 1],
\]
and its expected distribution under $H_0$ is always the uniform $\UM(0, 1)$, with CDF $\FB_0(r) = r$ for $r \in [0, 1]$.
\Algo~utilizes the following test statistic to measure the deviation between the empirical CDF $\FB_n$ and the expected CDF $\FB_0$:
\begin{equation}
\label{eq:S+}
S_n^+(s) = \sup\limits_{r \in [\rp^+, 1) }  K_s^+(\FB_n(r), \FB_0(r)).
\end{equation}
In the above, there are two truncations: the first is on the $r$-domain, where $\rp^+ = \sup \{\rp_{(t)}: \rp_{(t)} \le c_n^+\}$ with $c_n^+ \in [0, 1]$ introduced for stability \new{and $\rp_{(t)}$ is the $t$-th smallest p-values among $Y_{1:n}$}. The second is to truncate $K_s$ to $K_s^+$, which is defined as a truncated version of $K_s$ defined by
\begin{equation}
\label{eq:Ks+}
K_s^+(u, v) = \begin{cases}
K_s(u, v), &~~\text{if}~~0 < v < u < 1,\\
0, &\new{~~\text{otherwise}},
\end{cases}
\end{equation}
where $K_s(u, v)$ represents the $\phi_s$-divergence between $\mathrm{Ber}(u)$ and $\mathrm{Ber}(v)$:\footnote{$\mathrm{Ber}(u)$ denotes a Bernoulli distribution with parameter (or head probability) $u$.}
\begin{align*}
 \label{eq:Ks}
 \begin{split}
 K_s(u, v) &= D_{\phi_s}(\mathrm{Ber}(u) \| \mathrm{Ber}(v))
 = v \phi_s\left( \frac{u}{v} \right) + (1-v) \phi_s\left( \frac{1-u}{1-v} \right).
 \end{split}
\end{align*}
Here, the scalar function $\phi_s(x)$, indexed by $s \in \RB$, is convex in $x$ and is defined by \citep{jager2007goodness}:
\begin{equation}
\label{eq:phi}
\phi_s(x) = \begin{cases}
x \log x -x +1, & ~~\text{if}~~s = 1, \\
\frac{1-s+sx-x^s}{s(1-s)}, & ~~\text{if}~~s \neq 0, 1, \\
-\log x +x -1,& ~~\text{if}~~s = 0.
\end{cases}
\end{equation}
In this family, $\phi_s$ provides a range of examples depending on the value of $s$. For $s = 1$, $\phi_1(x) = x \log x - x + 1$, which reduces $K_1$ to the KL divergence: $K_1(u, v) = v \log \frac{u}{v} + (1 - u) \log \frac{1 - u}{1 - v}$. When $s = 2$, $\phi_2(x) = \frac{1}{2}(x^2 - x - 1)$, yielding $K_2(u, v) = \frac{(u - v)^2}{2v(1 - v)}$, a form associated with Higher Criticism \citep{donoho2004higher, donoho2015higher}, as discussed in Remark \ref{rem:relation-HC}. For $s \neq 0, 1$, we have $K_s(u, v) = \frac{1}{s(1 - s)} \left[ 1 - u^s v^{1 - s} - (1 - u)^s (1 - v)^{1 - s} \right]$. This definition of $\phi_s$ ensures continuity in $s$ for all $x \in (0, 1)$.

\begin{algorithm}[t]
\caption{Truncated GoF detection method (\Algo)}
\label{alg:detection}
\begin{algorithmic}[1]
\State \textbf{Input:} Edited text $\token_{1:n}$, hash function $\AM$, secret key $\Key$, pivot statistic function $\Yars$, stability parameter $c_n^+$, and critical value $\delta$.
\State For $t = 1, 2, \ldots, n$, compute the pseudorandomness number $\xi_t = \AM(\token_{(t-m):(t-1)}, \Key)$.
\State For $t = 1, 2, \ldots, n$, compute the pivot statistic $\Yars_t = \Yars(\token_t, \xi_t)$.
\State For $t=1,2,\ldots, n$, calculate the p-value as $\rp_t = 1- \Yars_t$. 
\State Sort the p-values in the ascending order $\rp_{(1)} < \rp_{(2)}<\ldots < \rp_{(n)}$ and set $\rp_{(n+1)}=1$.
\State Define the function $K_s^+$ as in \eqref{eq:Ks+} and compute the test statistic by
\begin{equation}
\label{eq:equivalent-Sn}
S_n^+(s) = \sup_{t: \rp_{(t+1)} \ge c_n^+} K_s^+ (t/n, \rp_{(t)}). 
\end{equation}
\State \textbf{Claim:} Text $\token_{1:n}$ is partially LLM-generated if \eqref{eq:goodness-of-fit-test} holds; otherwise, it is human-written.
\end{algorithmic}
\end{algorithm}

\Algo~is based on the test statistic $S_n^+(s)$, which is introduced in \eqref{eq:S+} and can be straightforwardly computed using \eqref{eq:equivalent-Sn}.
This statistic is expected to be small under the null hypothesis $H_0$ and large under the alternative hypothesis $H_1^{(\mathrm{mix})}$. 
We define the detection method as follows: for a small $\delta > 0$, we reject $H_0$ in favor of $H_1^{(\mathrm{mix})}$ if  
\begin{equation}
\label{eq:goodness-of-fit-test}
n \cdot S_n^+(s)  \ge (1+\delta) \log \log n.
\end{equation}
In other words, if our measure $S_n^+(s)$ of deviance from $\FB_n$ to $\FB_0$ is sufficiently large—satisfying the condition in~\eqref{eq:goodness-of-fit-test}—then the detection method will conclude that the observed text was (partially) generated by a watermarked LLM rather than being human-written.

\begin{rem}[Selection of the critical value]
We discuss the selection of the critical value $(1+\delta) \log \log n$. Given the setup in \eqref{eq:robust}, where $Y_{1:n}$ are i.i.d. samples from $\mu_0$, the Type I error can be controlled by carefully choosing this critical value. Notably, this control is independent of any human edits. In our experiments, we determine the critical value using Monte Carlo simulations, a method that is both computationally efficient and effective in controlling the Type I error even with finite sample sizes.
\end{rem}

\begin{rem}[Differences from previous work]
\label{rem:HC-K}
Our \Algo~test departs from the GoF test in \citep{jager2007goodness} by using two truncations. First, we truncate $K_s$ to $K_s^+$ to facilitate theoretical analysis, leveraging the property that $\lim_{n \to \infty} \PB_0(\rp_{(t)} \le t/n + \delta,~\forall t \in [n]) = 1$ for any $\delta > 0$. Second, we restrict the domain of $r$ to $[\rp^+, 1)$ instead of $(0, 1)$ to exclude small $\rp$-values, such as $\rp_{(1)}$ and $\rp_{(2)}$, in defining $S_n^+(s)$. 
In summary, Jager and Wellner~\citep{jager2007goodness} focus on the untruncated statistic $S_n(s) := \sup_{r \in (0, 1)} K_s(\FB_n(r), \FB_0(r))$ and its theoretical properties under i.i.d. data. In contrast, we propose a truncated $S_n^+(s)$ in settings (watermark detection) where the data is not i.i.d.
\end{rem}

\begin{rem}[Drawbacks of an untruncated $r$-domain]
\label{rem:drawbacks-of-untruncation}
As introduced in Remark \ref{rem:HC-K}, the untruncated $S_n(s)$ has two main drawbacks. First, its weak convergence is slow: Jager and Wellner~\citep{jager2007goodness} show that $n \cdot S_n(s)$ converges weakly to a random variable under $H_0$ (see their Theorem 3.1), but Gontscharuk
et al.~\citep{gontscharuk2015intermediates} indicate that this convergence rate is extremely slow, suggesting $r$-domain truncation as a remedy. This observation also motivates our use of Monte Carlo simulations for determining critical values, despite the exact computation available in \citep{li2015higher,moscovich2023fast}.
Second, without removing extreme values from small $\rp$-values like $\rp_{(1)}$ or $\rp_{(2)}$, $S_n(s)$ is prone to a heavy-tail issue. In fact, it holds that $\PB_0(n \cdot S_n(s) \ge z) \ge \frac{1}{2z}$ for large $z > 0$ over certain $s$ values. Test statistics with heavy tails are generally undesirable, as they reduce power at stringent significance levels in small samples. Removing these small $\rp$-values mitigates the heavy-tail effect significantly in numerical experiments. See Supplementary \ref{appen:truncation} for more detailed discussion.
\end{rem}

\begin{rem}[Relationship to Higher Criticism]
\label{rem:relation-HC}
The celebrated Higher Criticism (HC) is a special instance of the above \Algo. 
HC rejects $H_0$ if a test statistic, denoted by $\HC_n^{+}$, exceeds $\sqrt{2(1+\delta)\log\log n}$ for a small value $\delta > 0$.
This test statistic is related to our $S_n^+(2)$ as follows:
\begin{equation}
\label{eq:relation-between-GoFT-HC}
n \cdot S_n^+(2) = \sup\limits_{r \in [ \rp^+, 1) }  \frac{n}{2}\frac{(\FB_n(r) - r)^2}{r(1-r)} \1_{\FB_n(r) \ge r} = \frac{1}{2} (\HC_n^{+})^2.
\end{equation}
Our rejection rule aligns perfectly with HC as described in the literature \citep{donoho2004higher, tony2011optimal, cai2014optimal}. See Supplementary \ref{appen:HC} for further introduction and its theoretical property. 
Since HC using $\HC_n^+$ can be viewed as a special case of \Algo~involving $nS_n^+(2)$, our primary analysis focuses on $nS_n^+(s)$ to cover a wide range of $s$. 
For completeness, we perform a simulation study on HC in Supplementary \ref{sec:exp-hist} and \ref{sec:HC-boundary}.
\end{rem}

\section{Theoretical Guarantees}
\label{sec:theory}

\new{
In this section, we present the theoretical properties of \Algo. We begin with several distributional assumptions on the LLM watermarks and human edits.
\begin{asmp}
\label{asmp:main}
Let $\bP_t$ denote the NTP distribution for $\token_t$.
Let $\FM_{t-1} = \sigma( 
\{\token_j, \xi_j, \bP_{j+1} \}_{j=1}^{t-1})$ be the $\sigma$-field generated by all text-generating randomness within the subtext $\token_{1:(t-1)}$. We assume
\begin{enumerate}
\item[\rm{(a)}] Perfect pseudorandomness: $\xi_{1:n}$ are i.i.d. and $\xi_t \perp \FM_{t-1}$.
\item[\rm{(b)}] Mixture of latent sources: $\token_t=\SM(\bP_t, \zeta_t)$ with probability $\eps_n$; otherwise, $\token_t$ is drawn independently from $\bP_t$, and is independent of $\zeta_t$ when conditional on $\FM_{t-1}$.
\end{enumerate}
\end{asmp}
We briefly comment on Assumption~\ref{asmp:main}. The $\sigma$-field $\FM_{t-1}$ captures all related information from the subtext $\token_{1:(t-1)}$, including the observed tokens, computed pseudorandom numbers, and their corresponding NTP distributions.
Condition (a) states that the pseudorandom $\xi_t$, which either generates or verifies the watermark signal in $\token_t$, behaves as an independent draw from a fixed distribution and is statistically independent of the past. 
This treatment—viewing pseudorandom as truly random—is standard in the watermarking literature, whether explicitly stated \citep{wu2023dipmark, zhao2024permute, piet2023mark, li2024statistical} or assumed implicitly \citep{kirchenbauer2023watermark, kirchenbauer2023reliability, fernandez2023three}. Notably, Condition (a) aligns with the spirit of the Working Hypothesis proposed in our companion paper \citep{li2024statistical}.
Condition (b) characterizes the relationship between the observed $\token_t$ and the computed $\zeta_t$. The rationale behind it is that, since a verifier doesn't know which tokens have been edited, modeling $\token_t$ as a mixture of a watermarked token and a human-written one offers a principled way to capture this uncertainty.
Specifically , $\token_t$ is either watermarked using $\zeta_t$ with probability $\eps_n$, or it does not depend on $\zeta_t$ due to a broken watermark signal.
This formulation enables theoretical analysis without modeling human behaviors, while still preserving the essential statistical structure for effective watermark detection.
}

\subsection{Detectability}
\label{sec:detectability}

We begin by focusing on the challenging scenario where the watermark signal diminishes asymptotically. Specifically, with Assumption \ref{asmp:main} and the following Assumption \ref{asmp:regular}, we assume that both the non-null fraction, $\eps_n \asymp n^{-p}$, and the distribution singularity, $\Delta_n \asymp n^{-q}$, decrease with the text length $n$. 
This setup draws inspiration from the classic framework of sparse detection problems \citep{donoho2004higher,donoho2015higher}, where the signal strength of non-null effects also declines at a polynomial rate in sparse Gaussian mixture detection.
Additionally, Assumption \ref{asmp:regular} can be relaxed to a weaker condition, which we discuss in Remark \ref{rem:weaker-condition}.

\begin{asmp}\label{asmp:regular}
Define $\Delta(\bP) := 1-\max_{\token \in \Voca} P_{\token}$ as the distribution singularity of $\bP$.
For any positive integer $t \le n$, we assume $\Delta(\new{\bP_t}) = \Delta_n \asymp n^{-q}$ with $q \in [0, 1]$. 
\end{asmp}
\begin{rem}
The distribution singularity $\Delta(\bP)$ is equivalent to the token entropy \new{up to constant factors}, the latter defined as $\mathrm{Ent}(\bP):= \sum_{\token \in \Voca} P_{\token} \log \frac{1}{P_{\token}}$. They are related via $\mathrm{Ent}(\bP) = \Theta(\Delta(\bP) \log \frac{1}{\Delta(\bP)})$. The proof is provided in Proposition \ref{prop:expectation-gaps} in the appendix. 
\end{rem}

The first question to address is whether there exists a statistical algorithm that can reliably solve the problem \eqref{eq:robust}.
As an extreme example, if $\eps_n=0$ or $\Delta_n=0$, we would have $H_1^{\mathrm{mix}} = H_0$ and thus no test could differentiate $H_1^{(\mathrm{mix})}$ from $H_0$.
In general, the detectability depends on the difference between the joint distributions of $Y_{1:n}$ under $H_0$ and $H_1^{(\mathrm{mix})}$. In the terminology of Donoho and Jin~\citep{donoho2004higher}, we say that $H_0$ and $H_1^{(\mathrm{mix})}$ merge asymptotically if the total variation distance between the joint distribution of $\Yars_{1:n}$ under $H_0$ and that under $H_1^{(\mathrm{mix})}$ tends to zero as the sequence length $n$ goes to infinity. Otherwise, we say they separate asymptotically.
The detection problem is detectable if and only if $H_0$ and $H_1^{(\mathrm{mix})}$ separate asymptotically.

The detectability behaves very differently in two regimes: the heavy edit regime where $\frac{1}{2} < p \le 1$ and the light edit regime where $0 < p \le \frac{1}{2}$.
We begin by considering the heavy edit regime where $p \in (1/2, 1]$. Theorem~\ref{thm:gumbel-sparse} proves that no statistical test based on the observed pivotal statistics $Y_{1:n}$ can reliably detect the embedded watermark\new{, if there are too many human edits.} Notably, this impossibility result holds irrespective of the underlying NTP distribution $\bP_t$'s. Therefore, in the heavy edit regime, watermark detection is impossible, regardless of the NTP distribution $\bP_t$'s.

\begin{thm}[Heavy edits case]
\label{thm:gumbel-sparse}
Under Assumption \ref{asmp:main}, if $\frac{1}{2} < p \le 1$, $H_0$ and $H_1^{(\mathrm{mix})}$ merge asymptotically.
For any test, the sum of Type I and Type II errors is 1 as  $n\to\infty$.
\end{thm}

The situation is much more involved when the edit is light where $p \in (0, 1/2]$. It turns out that the detectability of the embedded watermark further depends on the NTP distribution $\bP_t$'s. 
 To account for this, we introduce Assumption \ref{asmp:regular} which assumes the largest probability in all NTP distributions share the same value $1-\Delta_n$ and $\Delta_n \asymp n^{-q}$.

\begin{thm}[Light edits case]
\label{thm:gumbel-dense}
Under Assumptions \ref{asmp:main} and \ref{asmp:regular}, let $0 \le p \le \frac{1}{2}$ and $0 \le q \le 1$.
\begin{enumerate}
\item[\rm{(a)}] If $q + 2p > 1$, $H_0$ and $H_1^{(\mathrm{mix})}$ merge asymptotically.
Hence, for any test, the sum of Type I and Type II errors tends to 1 as  $n\to\infty$.
\item[\rm{(b)}] If $q + 2p < 1$, $H_0$ and $H_1^{(\mathrm{mix})}$ separate asymptotically.
For the likelihood-ratio test that rejects $H_0$ if the log-likelihood ratio is positive, the sum of Type I and Type II errors tends to 0 as $n\to\infty$.
\end{enumerate}
\end{thm}

\begin{rem}[A weaker assumption for Theorem \ref{thm:gumbel-dense}]
\label{rem:weaker-condition}
For simplicity of presentation, we adopt Assumption \ref{asmp:regular}, though the proof relies on a weaker condition: we assume that all NTP distributions $\bP_t$ are either uniformly $\bP_{1:n} \subset \PM_{\Delta_n}$ or $\bP_{1:n} \subset \PM_{\Delta_n}^c$ according to \eqref{eq:regular-set}.
In our appendix, we show that Theorem \ref{thm:gumbel-dense} holds under Assumption \ref{asmp:main} and this weaker condition.
Specifically, we prove that under Assumption \ref{asmp:main}, if $q + 2p > 1$ and $\bP_{1:n} \subset \PM_{\Delta_n}$, $H_0$ and $H_1^{(\mathrm{mix})}$ merge asymptotically; conversely, if $q + 2p < 1$ and $\bP_{1:n} \subset \PM_{\Delta_n}^c$, $H_0$ and $H_1^{(\mathrm{mix})}$ separate asymptotically.
The same reasoning applies to the subsequent Theorem \ref{thm:adaptivity-fitness-of-good} and \ref{thm:suboptimality}.
\end{rem}

Theorem~\ref{thm:gumbel-dense} shows that given $q + 2p < 1$, it is possible to distinguish between $H_0$ and $H_1^{(\mathrm{mix})}$ asymptotically, with the likelihood-ratio test an effective detection method. Conversely, if $q + 2p > 1$, $H_0$ and $H_1^{(\mathrm{mix})}$ become indistinguishable asymptotically.
Hence, $q + 2p = 1$ represents the theoretical boundary distinguishing detectable from undetectable regions.
Unfortunately, the likelihood-ratio test is impractical cause it relies on the unknown token distribution $\bP_t$ and non-null fraction $\eps_n$.

%

\subsection{Adaptive Optimality}
\label{sec:optimality}

An ideal optimal method should achieve the detection boundary $q + 2p = 1$ automatically in the previous difficult case, without requiring any problem-dependent knowledge such as $p$, $q$, or $\bP_t$'s.
This property is known as ``adaptive optimality'' in the literature \citep{donoho2004higher,donoho2015higher}, as it consistently works regardless of the underlying parameters and does not require this information.
In the following theorem, we show that \Algo~achieves adaptive optimality.

\begin{thm}[Optimal adaptivity]
\label{thm:adaptivity-fitness-of-good}
Under Assumptions \ref{asmp:main} and \ref{asmp:regular}, if $q + 2p < 1$, for \Algo~with any $0 \le c_n^+ \le 1/n$ and $s \in [-1, 2]$, the sum of Type I and Type II errors tends to 0 as $n\to\infty$.
\end{thm}

Due to the prevalence of sum-based detection rules in the literature \citep{scott2023watermarking,fernandez2023three,li2024statistical}, we next examine whether these tests can achieve the same optimal adaptivity.
We say a score function $h$ is parameter-free if, for any $y$, $h(y)$ does not depend on $\Delta$ and $\eps$.

\begin{thm}[Suboptimality of sum-based detection rules]
\label{thm:suboptimality}
Let Assumptions \ref{asmp:main} and \ref{asmp:regular} hold.
Consider the detection rule specified by $h$: $T_h(Y_{1:n}) = 1$ if $\sum_{t=1}^n h(\Yars_t) \ge n \cdot \EB_0 h(\Yars) +\Theta(1) \cdot n^{\frac{1}{2}} a_n$, otherwise it equals 0, where $a_n \to \infty$ and $\frac{a_n}{n^{\gamma}} \to 0$ for any $\gamma > 0$.\footnote{The choice of $a_n$ ensures that $T_h$ has a vanishing Type I error. Examples include any polynomial function of $\log n$ or $\log\log n$ with positive coefficients.}
For any score function that is \rm{(i)} non-decreasing, \rm{(ii)} non-constant, \rm{(iii)} parameter-free, and \rm{(iv)} does not have discontinuities at both $0$ and $1$, the following results hold for $T_h$:
\begin{enumerate}
\item If $q + p < \frac{1}{2}$, the sum of Type I and Type II errors tends to 0.
\item If $q + p > \frac{1}{2}$, the sum of Type I and Type II errors tends to 1.
\end{enumerate}
\end{thm}

As shown in Theorem \ref{thm:suboptimality}, for nearly any non-decreasing and parameter-free score function $h$, the sum-based detection rule it introduces is strictly suboptimal. Specifically, Theorem \ref{thm:suboptimality} shows that the detection boundary for these rules is $q + p = \frac{1}{2}$, rather than the optimal $q + 2p = 1$. Consequently, all existing sum-based detection rules fail to achieve the optimal detection boundary.

\begin{cor}
Under Assumptions \ref{asmp:main} and \ref{asmp:regular}, the detection boundary for the existing score function $h \in \{\hars, \hlog, \hind, \hoptarso\}$ with both $\delta, \Delta_0 \in (0, 1)$ is $q+p=\frac{1}{2}$.
\end{cor}

\subsection{Optimal Efficiency Rate}
\label{sec:optimal-rate}
We now turn to the constant edit region, where $\Delta_n = \Delta \in (0, 1)$ and $\eps_n \equiv \eps \in (0, 1]$. In this scenario, $p = q = 0$ according to the notation $\Delta_n \asymp n^{-p}$ and $\eps_n \asymp n^{-q}$. Theorem \ref{thm:gumbel-sparse} implies that detection is always possible, as $H_0$ and $H_1^{(\mathrm{mix})}$ asymptotically separate. Furthermore, Theorems \ref{thm:adaptivity-fitness-of-good} and \ref{thm:suboptimality} suggest that both \Algo~and sum-based detection rules are viable detection methods. Unlike the detection boundary, which assesses adaptivity to varying problem difficulties, we employ a different criterion to evaluate feasible detection methods.

Li et al.~\citep{li2024statistical} introduces a notion of test efficiency for watermark detection. The key idea is to define efficiency as the rate of exponential decrease in Type II errors for a fixed significance level $\alpha$, considering the least-favorable NTP distribution within a belief class $\PM$. The formal definition is provided in Definition \ref{def:efficiency}.

\begin{defn}[$\PM$-efficiency \citep{li2024statistical}]
\label{def:efficiency}
Consider the detection rule that rejects $H_0$ if $S(\Yars_{1:n})$ is larger than a critical value.
Let $\gamma_{n,\alpha}$ be the critical value that ensures a Type I error of $\alpha$ for the problem \eqref{eq:robust}, that is, $\PB_0(S(\Yars_{1:n}) \ge \gamma_{n,\alpha}) = \alpha$ for all $n \ge 1$.
For a given belief class $\PM$, we define the following limit as the $\PM$-efficiency rate of $S$ and denote it by $R_{\PM}(S)$:

\[
\lim_{n \to \infty}\sup_{\bP_t \in \PM}\frac{1}{n}\log\PB_1(S_n \le \gamma_{n,\alpha})  = -
R_{\PM}(S).
\]
\end{defn}

This efficiency notion is referred to as class-dependent efficiency, as it depends on a given class $\PM$, which characterizes the prior belief about the underlying NTP distributions $\bP_t$'s. In their study, they set the prior class as the $\Delta$-regular set $\FPM$ (defined in \eqref{eq:regular-set}) and identify the optimal score function $\hoptars$ for sum-based detection rules. It is natural to compute the $\FPM$-efficiency rate for \Algo, which provides a quantitative measure of their detection efficiency in the constant parameter region. We present the results in the following theorem.

\begin{thm}[Optimal $\FPM$-efficiency]
\label{thm:optimal-rate}
Assume Assumption \ref{asmp:main} hold.
Let $s \in (0, 1)$, $c_n^+=0$, $\eps_n \equiv \eps \in (0, 1]$ and $\Delta_n\equiv\Delta \in (0, 1)$.
Given a dataset of pivotal statistics $\Yars_{1:n}$, it follows that
\[
\sup_{\text{measurable}~S} R_{\PM_{\Delta}}(S)
=  \KL(\mu_0, (1-\eps)\mu_0 + \eps \mu_{1, \bP_{\Delta}^{\star}}) = R_{\PM_{\Delta}}(S_n^+(s))
\]
where $\bP_{\Delta}^{\star}$ is the least-favorable NTP distribution defined by
\begin{equation}
\label{eq:worst-case-P}
\bP_{\Delta}^{\star} =\biggl(\underbrace{1-\Delta, \ldots, 1-\Delta}_{\floor{\frac{1}{1-\Delta}}\ \textnormal{times}}, 1-(1-\Delta)\cdot \left\lfloor\frac{1}{1-\Delta}\right\rfloor, 0, \ldots\biggr).
\end{equation}
\end{thm}

\begin{figure}[t!]
\centering
\includegraphics[width=0.9\textwidth]{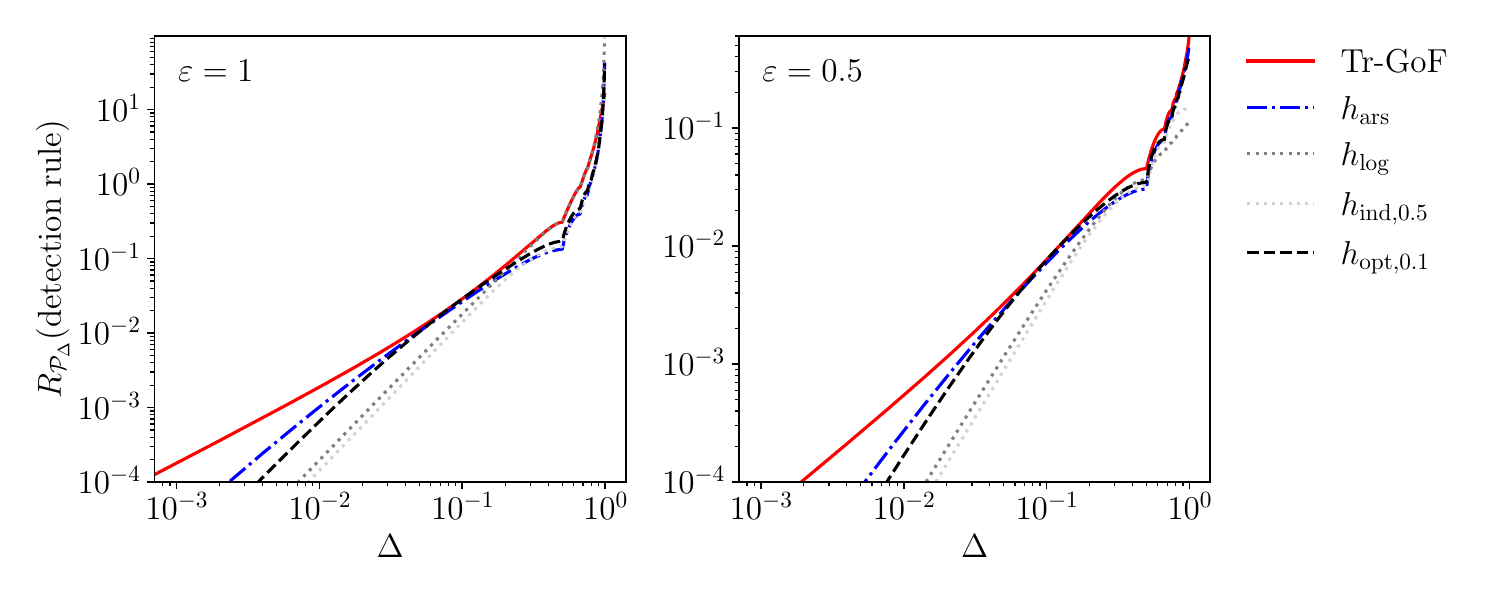}
\vspace{-0.2 in}
\caption{$\FPM$-efficiency rates of different detection methods for the Gumbel-max watermarks. $R_{\FPM}$ has non-smooth points when $\Delta = \frac{1}{2}, \frac{2}{3}, \frac{3}{4}, \ldots$ \citep{li2024statistical}.
}
\label{fig:efficiency}    
\vspace{-0.1in}
\end{figure}

Theorem \ref{thm:optimal-rate} establishes both upper and lower bounds in terms of $\FPM$-efficiency for the robust detection problem \eqref{eq:robust} under the constant parameter region. On one hand, it implies that the $\FPM$-efficiency rate of any measurable function is upper-bounded by $\KL(\mu_0, (1-\eps)\mu_0 + \eps \mu_{1, \bP_{\Delta}^{\star}})$. On the other hand, it shows that \Algo achieves this optimal $\FPM$-efficiency rate without any prior knowledge of $\eps$ and $\Delta$. When $\eps = 1$, this optimality is also achieved by the sum-based detection rule introduced by $\hoptars$ in \citep{li2024statistical}. However, as can easily be seen, computing $\hoptars$ requires knowledge or correct belief about the value of $\Delta$, which limits its practical applicability.
Furthermore, when $\eps < 1$, $\hoptars$ is no longer optimal because it does not consider the factor $\eps$.

To illustrate this optimal efficiency further, Figure \ref{fig:efficiency} presents the $\FPM$-efficiency rates of different detection methods for different values of $\Delta$. The left panel shows the case \(\eps = 1\) and the right panel shows \(\eps = 0.5\). Here, \(\hars\), \(\hlog\), \(\hind\), and \(\hoptarso\) represent commonly used sum-based detection rules, with \(\delta = 0.5\) and \(\Delta_0 = 0.1\) set for illustration purposes. 
In both cases, \Algo~consistently achieves the optimal $\FPM$-efficiency rate for all \(\Delta \in [0.001, 1]\). In contrast, sum-based detection rules generally fail to reach the optimal efficiency rate across most values of \(\Delta\) due to their lack of adaptivity. For example, when \(\eps = 1\), \(\hoptarso\) attains the optimal $\FPM$-efficiency rate if \(\Delta = \Delta_0\). However, once \(\eps < 1\), \(\hoptarso\) loses this optimality, as it no longer adapts to the reduced $\eps$.

\section{Simulations}
\label{sec:simulation}

In this section, our simulation studies first visualize the null and alternative distributions of $S_n^+(s)$ and then verify the empirical transition boundaries for different detection methods.

\subsection{Experimental Setup}
\label{sec:simulation-main}

We use a vocabulary of size $|\Voca|$ and model the pseudorandom variables $\xi_{t, \token}$ as true i.i.d. samples from $\UM(0, 1)$. We will explore the practical setting where $\xi_t$ is computed via a hash function later.
For given $p, q \in (0, 1]$, we set $\eps_n = n^{-p}$ and $\Delta_n = n^{-q}$ where $n$ is the text length.
We use the following procedure to obtain samples of $\log(n S_n^+(s))$ under different settings:
\begin{enumerate}
\item Draw $n$ i.i.d. samples from $\UM(0, 1)$ to represent $\Yars_{1:n}$ in $H_0$ and calculate $\log(n S_n^+(s))$.
\item Replace $\ceil{n \eps_n}$ of the previous samples by the same number of samples from $F_{1,\bP_{t}}$ where the NTP distributions $\bP_t$ are i.i.d. generated. Then, calculate $\log(n S_n^+(s))$.
\item Repeat Steps 1 and 2 $N$ times. Obtain $N$ independent samples of $\log(n S_n^+(s))$ under $H_0$ and $H_1^{\mathrm{mix}}$ respectively and create histograms of the simulated statistics.
\end{enumerate}
In all subsequent experiments, we fix $N=10^3$ while varying the values of $|\Voca|$, $n$, and the pair $(p, q)$, which will be specified accordingly.
We explore two methods for generating $\bP_t$, denoted by \textsf{M1} and \textsf{M2}. In particular, \textsf{M1} randomly generates $\bP_t$ by first setting its largest probability to $1-\Delta_n$, and then configuring the remaining probabilities in $\bP_t$ to adhere to Zipf's law \citep{zipf2016human}.
    In particular, we first i.i.d. sample $a_t \sim \UM(0.95, 1.5)$ and $b_t \sim \UM(0.01, 0.1)$, and then define $P_{t, \token} = \Delta_n \cdot (\token - 1 + b_t)^{-a_t} / C$, where $C = \sum_{\token=2}^{|\Voca|} (\token -1 + b_t)^{-a_t}$ serves as the normalizing constant. \textsf{M2} straightforwardly sets $\bP_t$ to be equivalent to $\left(1-\Delta_n, \frac{\Delta_n}{|\Voca|-1},\ldots, \frac{\Delta_n}{|\Voca|-1}\right)$.
     

\subsection{Histograms of \Algo~Statistics}

\begin{figure}[t!]
\centering
\includegraphics[width=1.0\textwidth]{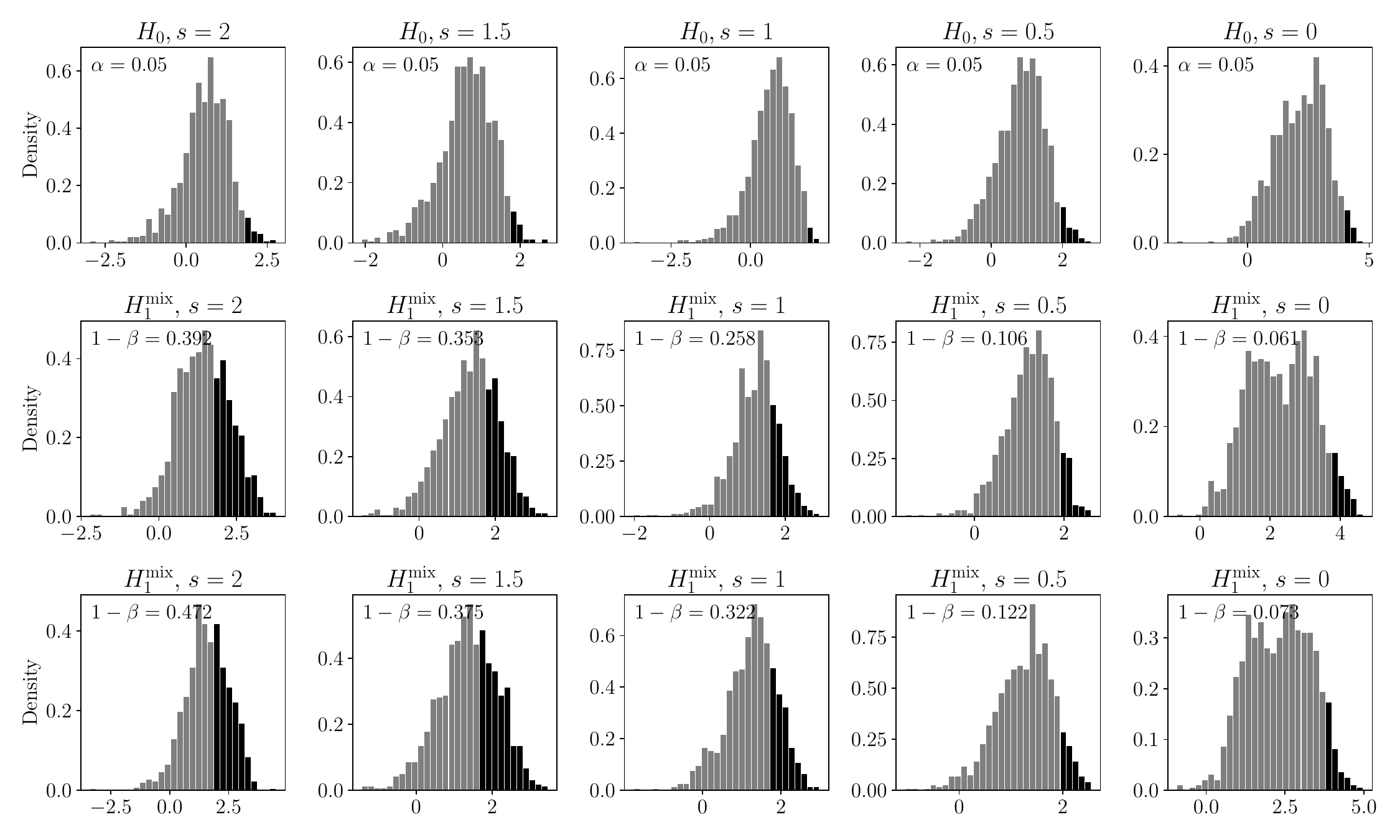}
\vspace{-0.3in}
\caption{
Density histograms and powers of $\log(nS_n^+(s))$ for different values of $s$, with $c_n^+ = \tfrac{1}{n^2}$ and $(p, q) = (0.2, 0.5)$. The first row reports results under $H_0$, while the second and third rows correspond to $H_1^{\mathrm{mix}}$ in the \textsf{M1} and \textsf{M2} settings. The shaded region indicates the rejection area, where $\alpha$ denotes the Type I error and $1-\beta$ the power.
}
\label{fig:other-s-hist}    
\vspace{-0.1in}
\end{figure}

To begin with, we investigate the empirical distribution of $\log(nS_n^+(s))$ under $H_0$ and $H_1^{\mathrm{mix}}$.
We set $|\Voca|=n=10^3$ in this investigation and collect $N=10^3$ samples of $\log(nS_n^+(s))$ for different settings.
In Figure \ref{fig:other-s-hist}, we visualize the empirical density histograms of $\log(n S_n^+(s))$ for five values of $s \in \{2, 1.5, 1,0.5, 0\}$ and three settings.
We shade the rejection regions and mark the corresponding statistical powers as $1-\beta$.
We observe that 

\begin{itemize}
    \item For $0.5 \le s \le 1.5$, the null distributions of $\log(n S_n^+(s))$ exhibit similar shapes and supports. However, for $s \in \{0, 2\}$, there are noticeable changes in shape and support. Remarkably, the maximum value that $\log(n S_n^+(s))$ could take increases dramatically from $2$ (at $s=1$) to $4$ or even $5$ (when $s=0$ or $2$).
    \item For the majority of the $s$ values, the empirical distribution of $\log(n S_n^+(s))$ under $H_1^{\mathrm{mix}}$ resembles its counterpart under $H_0$, albeit with a notable shift to the right. This shift can range from moderate (for $s \in \{0, 0.5\}$) to quite significant (for $s \in \{1,1.5, 2\}$). It is this shift that enables the statistics $n S_n^+(s)$ to detect the distributional differences caused by embedded watermarks.
\end{itemize}
Even if we set $c_n^+ = 0$, the observed patterns remain consistent, except that the support under $H_1^{\mathrm{mix}}$ would be considerably larger due to the heavy-tailed behavior of $\rp_{(1)}$. See Supplementary \ref{sec:additional-result-GOT-nomask} for additional histogram results of $c_n^+ = 0$. Negative values of $s \in [-1, 0)$ yielded results nearly identical to $s = 0$ in both histograms and statistical power, so we omit them.

\subsection{Detection Boundary and Optimal Adaptivity}
\label{sec:adaptivity}

\begin{figure}[t!]
\centering
\includegraphics[width=1.0\textwidth]{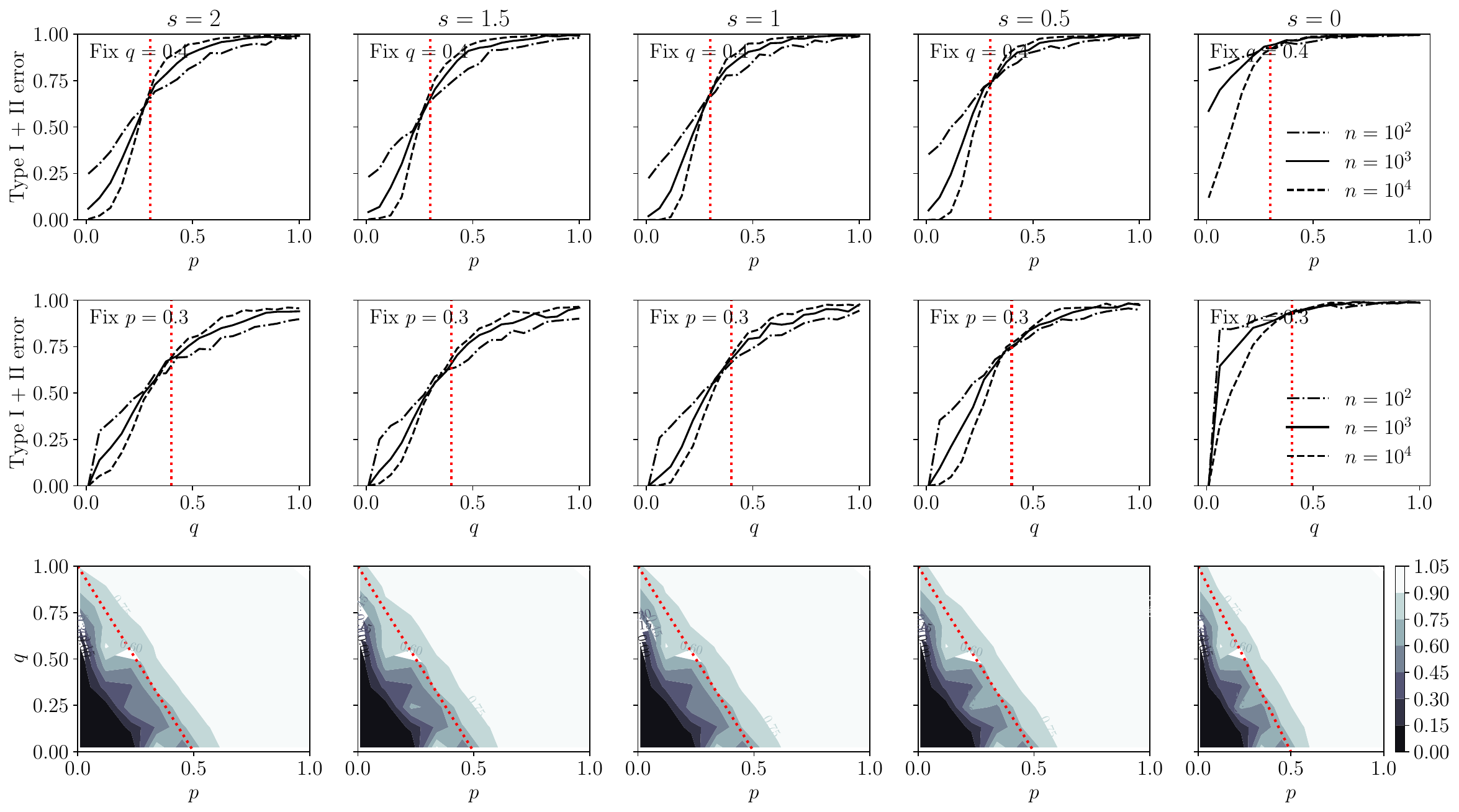}
\vspace{-0.3in}
\caption{Empirical detection boundaries of \Algo for different values of $s$.
Each column corresponds to results for a specific value of $s$.
The first two rows display the sum of Type I and II errors, with $q$ fixed at 0.4 and $p$ at 0.3, respectively.
The last row features the contour plot which illustrates the sum of errors across the domain $[0, 1]^2$, with $n = 10^4$.
Red dotted lines present the theoretical detection boundaries.
All the results are averaged over $N=10^3$ independent trials. 
}
\label{fig:empirical-error-others}
\vspace{-0.1in}
\end{figure}

We then study the detection boundary of \Algo. They reject $H_0$ when $\log(nS_n^+(s)) \ge C_{\mathrm{GoF}}$ for a predetermined critical value $C_{\mathrm{GoF}}$ and a given value $s \in \{2,1.5,1,0.5,0\}$.
We use the same setup introduced in Section \ref{sec:simulation-main} but with $|\Voca|=5$ and focus on the \textsf{M2} construction for NTP distributions. We aim to check how the smallest sum of Type I and Type II errors changes as we vary the other parameters, namely $(p, q)$, and $n$. Following the approach used in \cite{tony2011optimal}, we compute the smallest sum errors by tuning the critical value $C_{\mathrm{GoF}}$ from a predetermined set. Here, we use the set $\RM(0, 30, 10^3)$ where $\RM(a, b, K)$ consist of $K$ equally spaced points starting from $a$ to $b$ over an interval defined by $(K-1)$ divisions, that is, $\RM(a, b, K)=\{a + \frac{k}{K-1} \cdot (b-a):k=0, 1, \ldots, K-1 \}$.

\vspace{-1em}\paragraph{Phase transition for a fixed $p$ or $q$.}
We first fix either $q=0.4$ or $p=0.3$ and select a sample size $n$ from $\{10^2, 10^3, 10^4\}$.
According to Theorem~\ref{thm:gumbel-dense}, the detection boundary is given by $q+2p=1$, which suggests a transition at either $p=0.3$ or $q=0.4$. 
This prediction is validated by the first and second rows of Figure~\ref{fig:empirical-error-others}. For instance, when $q$ is fixed at 0.4, the error sum $\alpha + \beta$ initially increases from zero and stabilizes at one as we increase $p$ from zero to one. The transition point occurs around $p=0.3$ and aligns well with the red dashed line. Furthermore, larger sample sizes make the alignment with the theoretical prediction more pronounced.

\vspace{-1em}\paragraph{Adaptive optimality of \Algo.}
To accurately capture the empirical detection boundary, we use $n=10^4$ independent samples to calculate $\log(n S_n^+(s))$.
For any $p \in \RM(0.01, 1, 20)$ and $q \in \RM\left(\log_n\frac{|\Voca|}{|\Voca|-1}, 1, 20\right)$,\footnote{$\log_n\frac{|\Voca|}{|\Voca|-1} \le q$ is solved from the fact $1-n^{-q} \ge \frac{1}{|\Voca|}$ which always holds due to the pigeonhole principle.} we compute the smallest sum of Type I and Type II errors by tuning the critical value as mentioned earlier.
These results are displayed in the bottom row of Figure~\ref{fig:empirical-error-others}.
Here, darker areas indicate lower error sums, while lighter regions present higher error sums.
A red dashed line represents the theoretical detection boundary $q+2p=1$. Most darker regions are below this boundary, while lighter regions are above it. 
This empirical boundary aligns well with the theoretical prediction in Theorem \ref{thm:adaptivity-fitness-of-good}, no matter what the value of $s$ we choose.

\vspace{-1em}\paragraph{Suboptimality of existing detection methods.}

\begin{figure}[ht]
\centering
\vspace{-0.2in}
\includegraphics[width=0.9\textwidth]{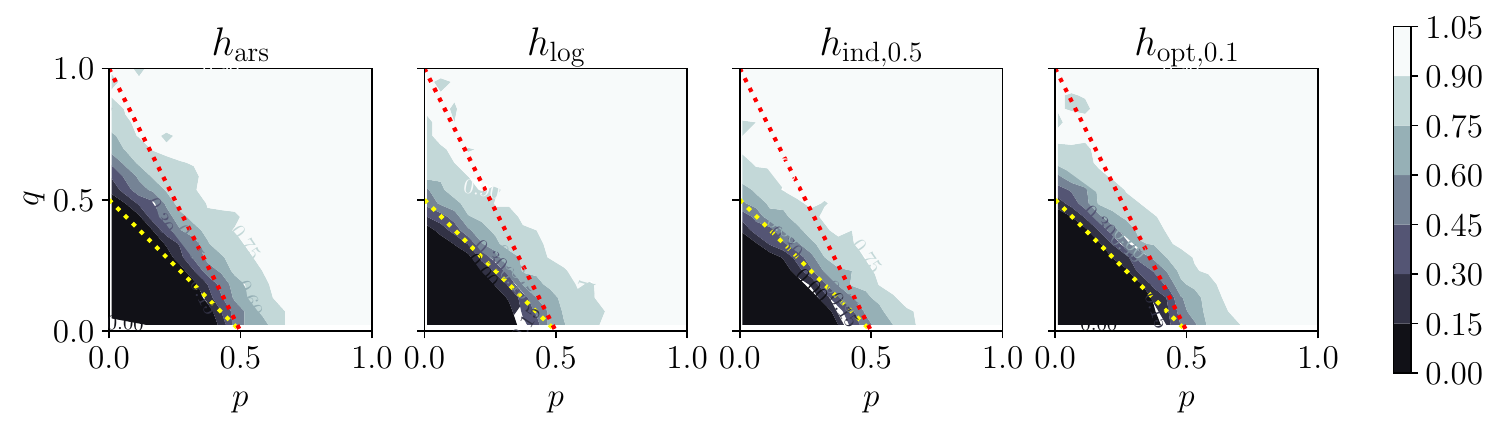}
\vspace{-0.1in}
\caption{Empirical detection boundaries for sum-based detection rules $T_h$'s.
}
\label{fig:empirical-error-existing}
\vspace{-0.1in}
\end{figure}
For any function $h \in \{\hars, \hlog, \hind, \hoptarso\}$, we consider the following sum-based detection rule specified by $h$:
\vspace{-1em}
\begin{equation*}
T_h(Y_{1:n}) = 
\begin{cases}
1 & ~\text{if}~\sum_{t=1}^n h(\Yars_t) \ge n \cdot \EB_0 h(\Yars) + C_{\mathrm{sum}} \cdot n^{\frac{1}{2}} \log n,\\
0 &~\text{if}~\sum_{t=1}^n h(\Yars_t) < n \cdot \EB_0 h(\Yars) + C_{\mathrm{sum}} \cdot n^{\frac{1}{2}} \log n.
\end{cases}
\vspace{-1em}
\end{equation*}
The expectation $\EB_0 h(\Yars)$ for different $h$'s has been listed in Table 1 in \citep{li2024statistical}.
Mirroring our approach with \Algo, for each detection rule $T_h$, we tune the parameter $C_{\mathrm{sum}}$ from a given set to obtain the smallest sum of Type I and Type II errors.
This set for $\hars$ is $\RM(8, 60, 10^3)$, for $\hlog$ is $\RM(-20, 0, 10^3)$, and for $h_{\mathrm{ind}, 0.5}, h_{\mathrm{opt}, 0.1}$ is $\RM(-10, 10, 10^3)$.
The findings are presented in Figure \ref{fig:empirical-error-existing}, which was created using the same procedure as the contour plot in Figure \ref{fig:empirical-error-others}. Observations reveal that the empirical boundaries closely match the theoretical prediction of $p+q = \frac{1}{2}$, represented by yellow dotted lines. 
This alignment corroborates our Theorem \ref{thm:suboptimality}.
In contrast, the red dotted line represents the optimal detection boundary $q+2p=1$.
The discrepancy between the yellow and red dotted lines shows the suboptimality of sum-based tests.

\section{Experiments on Open-Source LLMs}
\label{sec:LLM-experiments}

\subsection{Experiment Setup}

We follow the experimental setup from \citep{kuditipudi2023robust, li2024statistical}. We begin by sampling 1,000 documents from the news-like C4 dataset \citep{raffel2020exploring}, which serve as the initial prompts. Using these prompts, we then ask two language models, the OPT-1.3B model \citep{zhang2022opt} and Sheared-LLaMA-2.7B \citep{xia2023sheared}, to generate an additional $n=400$ tokens for each document. Our evaluation focuses on two key aspects:

\begin{enumerate}   
    \item \textbf{Statistical power}:
    To evaluate the statistical power of different detection methods, we set the significance level at $\alpha = 0.01$ by tuning the critical values, either using central limit theorem predictions or Monte Carlo simulations. We assess the actual Type I errors using unwatermarked texts and Type II errors using watermarked texts. 
    \vspace{-0.5em}
   \item \textbf{Robustness to edits}: Almost any editing method tends to weaken watermark signals. We focus on three representative types of edits \citep{kuditipudi2023robust}: random edits (including substitution, insertion, and deletion), adversarial edits, and roundtrip translation. In random edits, a random fraction of the watermarked tokens is replaced, inserted, or deleted, with substitutions and insertions involving tokens uniformly selected from the vocabulary \(\Voca\). 
   Adversarial edits, by contrast, are more targeted, selectively modifying watermarked tokens to maximize the removal of watermark signals within a given edit budget. 
   Finally, roundtrip translation involves translating the text from English to French and back to English using another language model.
Random and adversarial edits enable systematic control over the level of edits, while roundtrip translation is more likely to be encountered in practice.
    \vspace{-0.5em}
\end{enumerate}

We compare the \Algo~test across three different values of $s$ from the set $\{1, 1.5, 2\}$ alongside three sum-based detection rules, each specified by a particular score function $h$. These include: (i) Aaronson's function $\hars(y) = - \log(1-y)$ \citep{scott2023watermarking}, (ii) the logarithmic function $\hlog(y) = \log y$ \citep{kuditipudi2023robust, fernandez2023three}, and (iii) the optimal least-favorable function \citep{li2024statistical} with $\Delta$ a user-specified regularity parameter
Among $\Delta \in {0.1, 0.2, 0.3}$, $\Delta = 0.1$ performs best and is used in all subsequent experiments.
In all experiments, we use 1-sequence repeated context masking \citep{Dathathri2024}. This approach watermarks a token only when the current text window is unique within the generation history, aiming to preserve text quality.
In the following, we present our numerical results on the OPT-1.3B model for illustrative purposes, as the results for Sheared-LLaMA-2.7B are similar and are provided in the appendix. Additional experimental details are also deferred to the appendix.

\subsection{Statistical Power}
\label{sec:statistical-power}

\begin{wrapfigure}[13]{r}{0.5\textwidth}
\vspace{-0.3in}
\includegraphics[width=0.5\textwidth]{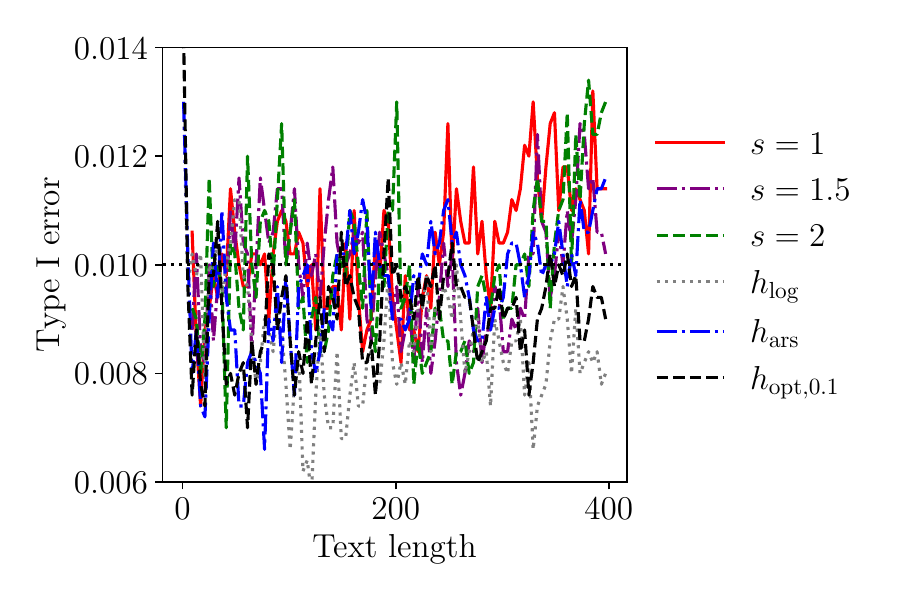} 
\vspace{-0.3in}
\caption{Empirical Type I errors.}
\label{fig:null}
\end{wrapfigure}

We first evaluate the statistical power of considered detection methods using unmodified texts.
Unlike in simulation studies, we cannot manipulate each NTP distribution $\bP_t$ in language model experiments to ensure they are $\Delta$-regular. However, $\Delta$-regularity correlates closely with the temperature parameter in LLMs \citep{ackley1985learning}. The temperature parameter modulates the raw outputs of the model's final layer (also known as the logit layer) before the softmax function is applied. Typically, a high temperature yields more uniform probabilities, encouraging diverse generations, while a low temperature sharpens the distribution, emphasizing the most probable prediction and thus favoring the greedy generation. As an approximation, we use four temperatures $\{0.1, 0.3, 0.5, 0.7\}$ to check the effect of $\Delta$-regularity on statistical power. The evaluation results are in Figures \ref{fig:null} and \ref{fig:type-I-and-II-errors}.
Results on the Sheared-LLaMA-2.7B model are similar and documented in Supplementary \ref{appen:extened-results}.

\paragraph{Type I error control.} 
We begin by examining Type I error control, using 5000 unwatermarked texts sampled from the C4 dataset as human-written data and evaluating Type I error at the significance level $\alpha=0.01$. Figure \ref{fig:null} shows how the empirical Type I error varies with increasing text lengths, revealing a consistent pattern: across all detection methods, empirical Type I errors remain well-controlled within the interval $[0.006, 0.014]$. This suggests that the pseudorandom variables effectively mimic the behavior of true random variables, enabling their empirical performance to closely align with theoretical expectations.

\begin{figure}[ht!]
\vspace{-0.5em}
\centering
\includegraphics[width=\textwidth]{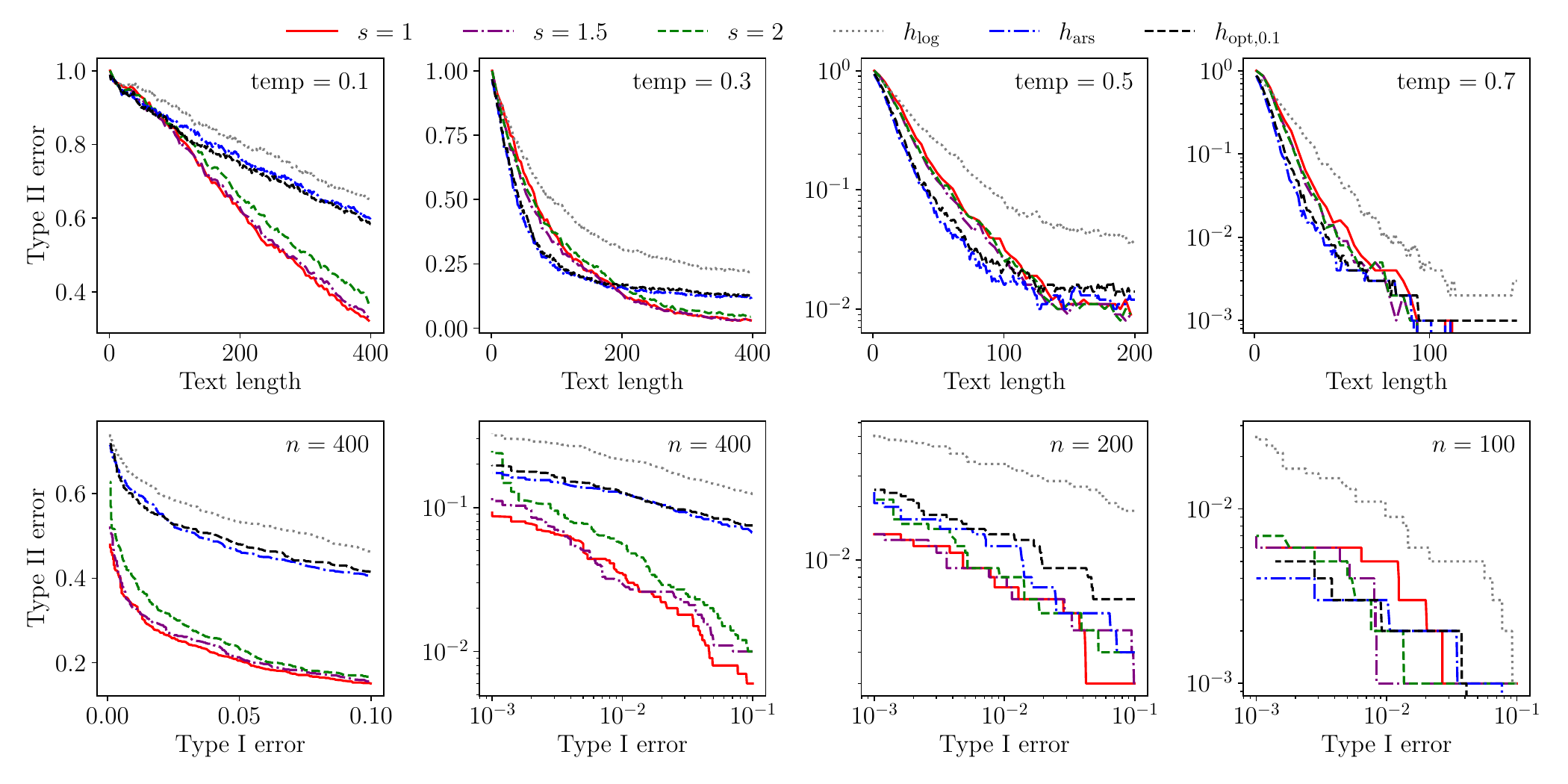}
\vspace{-0.3in}
\caption{Empirical Type II errors (top row) across different detection rules applied to the Gumbel-max watermark. The bottom row illustrates the trade-off function in the log-log scale for a specific length $n$. The temperatures used, from left to right columns, are 0.1, 0.3, 0.5, and 0.7, respectively.
}
\label{fig:type-I-and-II-errors}
\vspace{-0.2in}
\end{figure}

\vspace{-1em}\paragraph{Type II error decay.}
We then examine the decay of Type II errors at a significance level of $\alpha = 0.01$, shown in the top row of Figure \ref{fig:type-I-and-II-errors}. The bottom row displays a trade-off function illustrating detection performance across all critical values, extending beyond the fixed $\alpha = 0.01$ case. From left to right, the columns show increasing temperatures from $0.1$ to $0.7$.
We observe:
\begin{enumerate}
    \item  \textit{\Algo~excels in the low-temperature region.} 
    When the temperature is relatively low (that is, $0.1, 0.3$), the Type II errors for the \Algo~test decrease more rapidly and eventually fall below those of all baseline detection rules (when $s \in \{1, 1.5\}$). For instance, at a temperature of $0.1$ and using $400$ watermarked tokens, the baseline detection methods obtain a Type II error of approximately $0.6$, whereas the best \Algo~test exhibits Type II errors around $0.3$. Interestingly, although the optimal least-favorable function $h_{\mathrm{opt}, 0.1}$ performs (slightly) better than $\hars$, it is still inferior to the \Algo~test. The superior performance of the \Algo test is further evidenced by the trade-off functions, where the \Algo test consistently outperforms other methods for nearly all given Type I errors. However, this does not contradict the claimed optimality of $\hoptarso$ in \citep{li2024statistical}. This is because the optimality is defined for the least-favorable case, but in practice, the conditions are not adversarial enough for it to matter.

    \item \textit{\Algo~ performs comparably in the high-temperature region.}
    At relatively high temperatures (that is, $0.5$ and $0.7$), the \Algo~test achieves Type II error decay comparable to the baseline $\hars$. At a temperature of $0.5$, most detection methods reach a Type II error of $0.01$ with only $100$ watermarked tokens. The trade-off function indicates that the \Algo~test maintains a slightly lower trade-off for small Type I errors (>$0.001$). At a temperature of $0.7$, nearly all Type II errors decay to $0.01$ with approximately 50 watermarked tokens. The improved detection effectiveness at higher temperatures stems from the increased variability in token generation, which strengthens statistical signals and facilitates watermark detection.

\end{enumerate}

\subsection{Robustness Evaluation}
\label{sec:robust-evaluation}

\begin{figure}[!t]
\vspace{-0.5em}
\centering
\includegraphics[width=\textwidth]{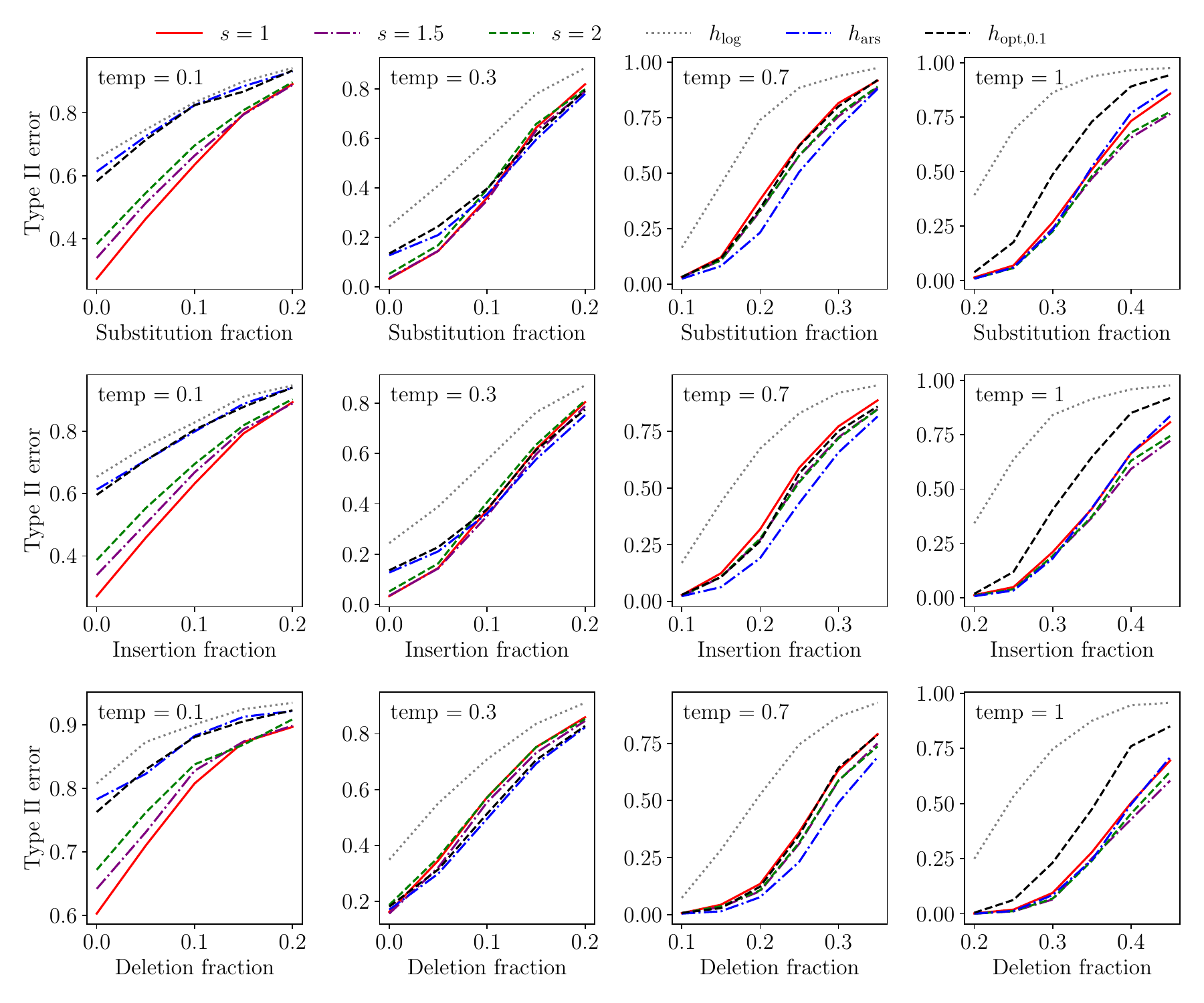}
\vspace{-0.3in}
\caption{Effect of three random edits on Type II error across different temperatures at a fixed Type I error of $\alpha=0.01$. The top, middle, and bottom plots correspond to random substitution, insertion, and deletion, respectively.
}
\label{fig:modification-nomask}
\vspace{-0.2in}
\end{figure}

\paragraph{Robustness to random edits.}
We examine three types of edits: random substitution, insertion, and deletion. For each specified edit method and a given fraction, we randomly modify the corresponding fraction of tokens in the watermarked text. We set the Type I error to $\alpha = 0.01$ and investigate how the Type II error changes with the edit fraction across four different temperatures: $\{0.1, 0.3, 0.7, 1\}$. This approach allows us to assess the robustness of watermark detection methods under various text edit conditions and temperature settings.
Generally, with a text window of size $m$, modifying a single token can affect the computation of up to $m$ pseudorandom numbers. Therefore, a window size of $m$ can always withstand an edit fraction of up to $1/m$, as not all watermarked signals are removed. In our experiments, we set $m = 5$.
The results of random edits are shown in Figure \ref{fig:modification-nomask}, with the top, middle, and bottom rows corresponding to random substitution, insertion, and deletion, respectively. Across all edit types, we observe a consistent pattern: any edit increases the Type II error for all detection methods. 
At lower temperatures, the error rate increases sharply with rising edit fractions, while at higher temperatures, the increase is more gradual, indicating reduced sensitivity to edits. 
\begin{enumerate}
\vspace{-0.5em}
\item  \textit{\Algo~demonstrates better robustness at lower temperatures.} The \Algo~test consistently achieves the lowest Type II error rates across all three edit types at lower temperatures (that is, 0.1, 0.3), outperforming other methods such as $h_{\mathrm{opt}, 0.1}$ and \(h_{\ars}\). This makes it particularly effective when the LLM outputs are more deterministic.
\vspace{-0.5em}
\item \textit{\Algo~maintains comparable robustness at higher temperatures.} While the \Algo~test slightly trails behind \(h_{\ars}\) at a temperature of $0.7$, it remains competitive, with an average difference in Type II error of only 0.02 compared to \(\hars\) across different edit fractions. At a temperature of 0.7, $s=2$ slightly outperforms the others. This suggests that, even in high-temperature scenarios, the \Algo~test maintains robust performance and adapts effectively to the impact of temperature on detection difficulty.
\end{enumerate}

\vspace{-1em}\paragraph{Edit tolerance limit.}
We define the edit tolerance limit as the largest fraction of edits that can be applied to a watermarked text with the detection method still rejecting the null hypothesis $H_0$. In general, the higher the edit tolerance limit, the more robust the detection method, as it is more sensitive to the weak watermark signal. 

We compute the edit tolerance limits for different detection methods across two tasks: (a) poem recitation, where the LLM is asked to recite an existing poem, and (b) poem generation, where the LLM generates a new poem in the style of a given one. The latter task is more open-ended, leading to more stochastic and regular generation. 
The results, averaged over 100 popular poems, are reported in Figure \ref{tab:largest-fraction-1.3B}, with the highest values (in percentage) highlighted in bold. Notably, the \Algo~test with $s=2$ consistently achieves the highest edit tolerance across all three edit types and both tasks. These findings align with the results shown in Figure \ref{fig:modification-nomask}. Similar results for the Sheared-LLaMA-2.7B model are provided in Supplementary \ref{appen:extened-robust-results}.

\begin{table}
\vspace{-0.1in}
\centering
\resizebox{\textwidth}{!}{ 
\begin{tabular}{c|c|cccccccc}
\toprule
Task & {Edit types}  & $s=1$ & $s=1.5$& $s=2$  & $\hlog$ & $\hars$ & $h_{\mathrm{opt}, 0.3}$ & $h_{\mathrm{opt}, 0.2}$& $h_{\mathrm{opt}, 0.1}$ \\
\midrule
\multirow{3}{*}{\shortstack[c]{Poem \\Recitation}} 
& Substitution & 37.16&38.8&\textbf{39.49}&25.04&37.96&27.83&30.04&33.87\\
\cmidrule{2-10}
& Insertion &42.15&45.25&\textbf{45.64}&26.46&44.12&30.38&33.32&37.78\\
\cmidrule{2-10}
& Deletion & 39.78&41.85&\textbf{42.67}&23.21&41.92&27.08&30.12&33.76\\
\midrule
\multirow{3}{*}{\shortstack[c]{Poem \\Generation}} 
& Substitution & 36.89&38.72&\textbf{38.9}&24.71&38.64&27.47&30.04&33.49\\
\cmidrule{2-10}
& Insertion & 40.08&42.52&\textbf{43.09}&25.54&41.65&29.49&32.75&36.28\\
\cmidrule{2-10}
& Deletion & 39.79&41.83&\textbf{42.42}&26.38&40.99&29.59&32.3&35.15\\
\bottomrule
\end{tabular}}
\caption{The edit tolerance limits $(\%)$ for detection methods on the OPT-1.3B model.}
\label{tab:largest-fraction-1.3B}
\end{table}

\vspace{-1em}\paragraph{Robustness to adversarial edits.}
In adversarial edits, we assume the human user knows the hash function $\AM$ and the secret key $\Key$, allowing them to selectively replace tokens with the strongest watermark signals to evade detection. To approximate this behavior, we use the following procedure: for the LLM-generated response, the user first computes all corresponding pivotal statistics, identifies a given fraction of tokens with the highest pivotal statistics, and replaces them with randomly selected tokens. This targeted replacement is more disruptive than random edits. Results for a $5\%$ replacement are shown in Figure \ref{fig:adversarial-best}, with results for other fractions in Supplementary \ref{appen:extened-robust-results}. 

In this adversarial setting, the \Algo~test demonstrates steady robustness, consistently achieving a lower Type II error across most temperature settings. Notably, $\hars$ is less resilient under adversarial edits, while $h_{\mathrm{opt}, 0.1}$ performs better. This increased robustness in $h_{\mathrm{opt}, 0.1}$ likely stems from its design through minimax optimization, which enhances its ability to withstand adversarial edits.

\begin{figure}[ht]
\vspace{-0.1in}
\centering
\includegraphics[width=\textwidth]{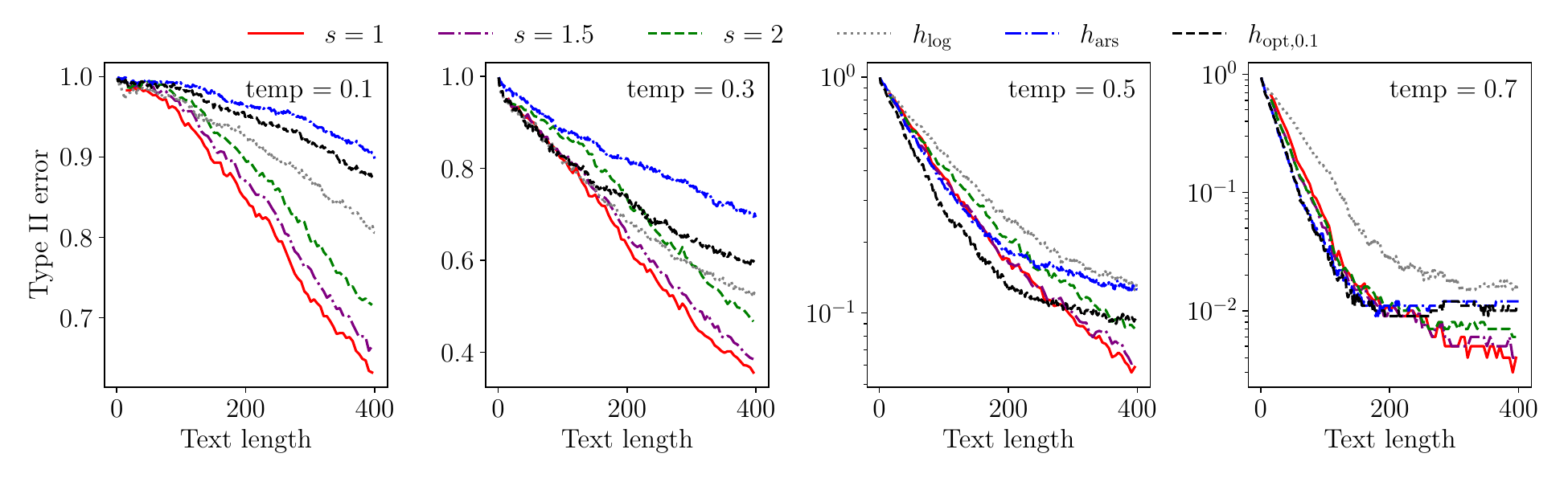}
\vspace{-0.3in}
\caption{
Effect of adversarial edits on Type II error across different temperatures.
}
\label{fig:adversarial-best}
\vspace{-0.05in}
\end{figure}

\begin{figure}[ht]
\vspace{-0.05in}
\centering
\includegraphics[width=\textwidth]{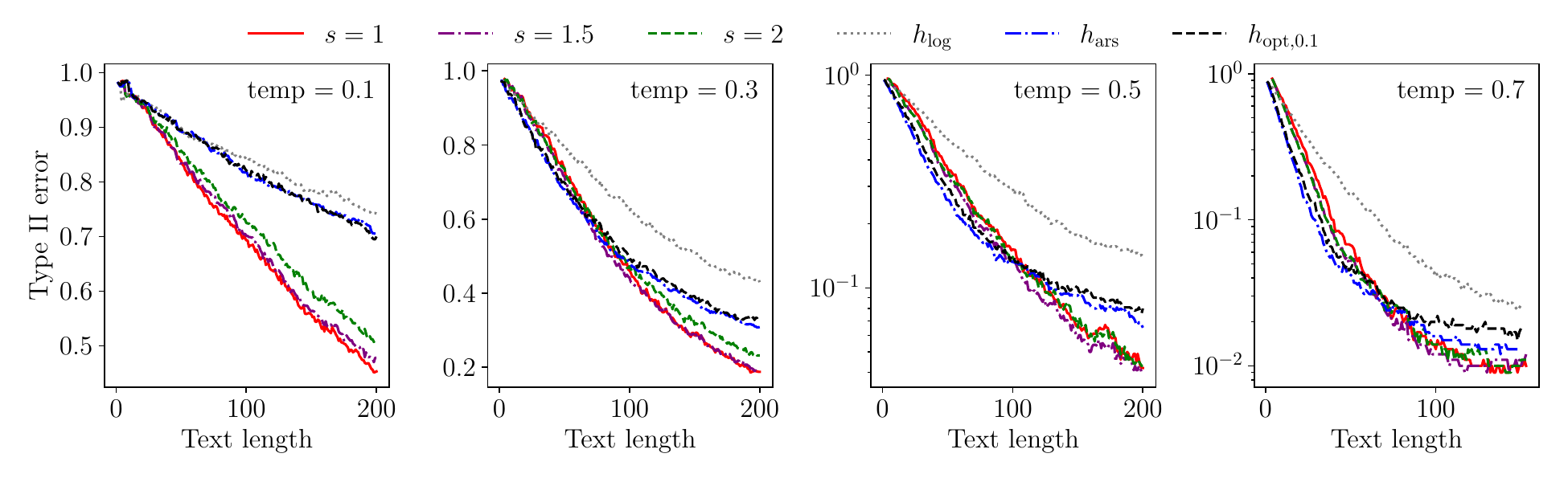}
\vspace{-0.3in}
\caption{
Effect of roundtrip translation on Type II error across different temperatures.
}
\label{fig:translation-best}
\vspace{-0.1in}
\end{figure}

\vspace{-1em}\paragraph{Robustness to roundtrip translation.}
In the roundtrip translation, we cannot control the edit level. By maintaining a fixed significance level of $\alpha=0.01$, we examine how the Type II error varies with the modified text length and different temperatures. The results are present in Figure \ref{fig:translation-best}. Using $n=200$ tokens, we observe that the final Type II errors decrease as the temperature increases, suggesting that higher temperatures facilitate easier detection. In the low-temperature range, the \Algo~test with $s = 1$ consistently outperforms all other detection methods. Conversely, in the high-temperature range, \Algo~tests with $s \in \{1, 1.5\}$ achieve comparable or occasionally superior performance to the previous $\hars$. These experiments underscore the robust performance and detection efficiency of \Algo for Gumbel-max watermarks.

\section{Discussion}
\label{sec:discuss}

In this paper, we have developed an adaptive and robust method for detecting watermarks in text generated by LLMs but subsequently edited by humans, along with theoretical guarantees corroborated by numerical experiments. We start by proposing a procedure for modeling the human editing process, which prompts us to formulate this problem as sparse mixture detection. Our method, which we call the \Algo~test, is shown to achieve the optimal detection boundary in the regime of increasing edit level and vanishing regularity of NTP distributions under the framework of \cite{li2024statistical}. In contrast, we show that sum-based detection rules are provably suboptimal in the sense that the detection region is strictly smaller than that of the \Algo~test. Additionally, we show that the \Algo~test continues to provide high detection efficiency when the edit level and regularity of NTP distributions remain constant. In contrast, sum-based detection rules fail to achieve robustness in either regime due to their inability to adapt to unknown specifics of problem instances.

Our findings open new avenues for the robust detection of LLM watermarks. First, although we focus on the Gumbel-max watermark, which is perhaps the most influential, it would be valuable to investigate robust detection for other watermarking schemes \citep{xie2024debiasing}, such as the inverse transform watermark \cite{kuditipudi2023robust} and the green-red list watermark \cite{kirchenbauer2023watermark}. In these cases, the truncation technique underlying \Algo~could offer useful insights. More broadly, future research might consider alternative nonparametric testing approaches \citep{nikitin1995asymptotic,d2017goodness}, such as the Kolmogorov--Smirnov test, Pearson's Chi-square test, and the Khmaladze--Aki statistic \citep{khmaladze1982martingale, podkorytova1994large}, which may offer adaptivity comparable to or even surpassing goodness-of-fit tests in certain parameter regimes. Empirical evaluation of these classical nonparametric methods could motivate the development of new procedures tailored to specific challenges, such as high-temperature regimes, where \Algo~may undergo a decline in detection power. Additionally, detection power could be enhanced by incorporating spatial information from watermark signals, as human edits often form clusters within text. From a theoretical standpoint, it is noteworthy that the optimal detection boundary $q+2p=1$ is achieved through pivotal statistics, suggesting that investigating an information-theoretic detection boundary could yield interesting insights. Lastly, for interpretive purposes, estimating the watermark fraction may offer insights into the degree of human contribution in content generated collaboratively by LLMs and humans.

\section*{Acknowledgments}
This work was supported in part by NIH grants RF1AG063481, U01CA274576, and R01EB036016; NSF grant DMS-2310679; a Meta Faculty Research Award; and Wharton AI for Business. The content is solely the responsibility of the authors and does not necessarily represent the official views of the NIH.


\bibliographystyle{plainnat}
\bibliography{bib/chatgpt,bib/privacy,bib/stat}

\newpage
\appendix
\begin{appendix}

\begin{center}
{\huge {Supplementary Material}}
\end{center}

This Supplementary Material includes the remaining proofs and technical details. Section~\ref{proof} presents the proofs, while Sections~\ref{sec:additional-simulation} and~\ref{sec:LLM} provide additional details for the simulation and LLM experiments, respectively.

\section{Proof for Theoretical Guarantees}
\label{proof}

\paragraph{Notations.}
We establish some conventions for the proofs in the appendix. 
For simplicity, we denote $\EB_{1, \bP_t}[\cdot] = \EB_{\Yars \sim \mu_{1, \bP_t}}[\cdot]$ and $\EB_{0}[\cdot] = \EB_{\Yars \sim \mu_{0}}[\cdot]$. When using $\EB_1[\cdot\mid\FM_t]$, we assume $\Yars_t$ follows the alternative hypothesis and take the expectation conditioned on the filtration $\FM_t$. When the context is clear, we will omit subscripts.
When the context is unclear, we will explicitly specify which variable the expectation is taken with respect to.
PDF stands for probability density function, and CDF stands for cumulative distribution function.

\subsection{Proof of Theorem~\ref{thm:gumbel-sparse}}
Detectability is essentially determined by whether the Hellinger distance between the joint distribution of $\Yars_{1:n}$ under $H_0$ and $H_1^{\mathrm{mix}}$ approaches 1 as $n$ goes to infinity.
It is important to note that due to the autoregressive generation structure, the joint distribution of $\Yars_{1:n}$ is not $\prod_{t=1}^n \mu_{1, \bP_t}$, but rather $\prod_{t=1}^n \bmut$, where $\bmut = \EB_1[\mu_{1, \bP_t}|\Yars_1,\ldots,\Yars_{t-1}]$ is the conditional version of $\mu_{1, \bP_t}$ given the historical information $\Yars_{1:(t-1)}$. As a result, the Hellinger distance's tensorization cannot apply directly as $\bmut$ still depends on the history $\Yars_{1:(t-1)}$, so we derive a generalization in the following lemma.

\begin{lem}
\label{lem:H2}
Let $\PM$ denote a prior set on which $\bP_t \in \PM$ for all $t \in [n]$.
Let $\rho_i$ denote the joint probability density distribution of $(\Yars_1, \ldots, \Yars_n)$ under the hypothesis $H_i$ for $i \in \{0, 1\}$.
\begin{enumerate}
\item If $n \cdot \sup_{\bP_t \in \PM }H^2(\mu_0, (1-\eps_n) \mu_0 +  \eps_n   \mu_{1,\bP_t}) = o(1)$, then
\[
\mathrm{TV}(\rho_0, \rho_1) \to 0~\text{as}~n \to \infty.
\]
\item Let $\bmut = \EB_1[\mu_{1, \bP_t}|\Yars_1,\ldots,\Yars_{t-1}]$ be the conditional version of $\mu_{1, \bP_t}$ given the history information $\Yars_1, \ldots,\Yars_{t-1}$.
If there exists a positive non-random sequence $c_n$ satisfying $n c_n \to \infty$ such that 
$\min_{t \in [n]}\inf_{\bP_t \in \PM} H^2(\mu_0, (1-\eps_n) \mu_0 + \eps_n  \bmut) \ge c_n$ holds almost surely for each $n \ge 1$, then
\[
\mathrm{TV}(\rho_0, \rho_1)\to 1~\text{as}~n \to \infty.
\]
\end{enumerate}
\end{lem}

We use Lemma \ref{lem:H2} to establish the detectability and identify the optimal detection boundary. With the established Lemma \ref{lem:H2}, we can prove Theorem~\ref{thm:gumbel-sparse}.
\begin{proof}[Proof of Theorem~\ref{thm:gumbel-sparse}]
Recall that $f_{1,\bP_t}(r) = \sum_{\token \in \Voca} r^{1/P_{t,\token}-1}$ is the PDF of $\mu_{1,\bP_t}$.
Note that the density ratio between $\mu_{1,\bP_t}$ and $\mu_0$ is still $f_{1,\bP_t}$.
By Lemma~\ref{lem:H2}, it suffices to show 
\[
n \cdot \sup_{t \in [n]} \sup_{\bP_t}H^2(\mu_0, (1-\eps_n) \mu_0 + \eps_n \mu_{1,\bP_t}) = o(1).
\]
By the definition of Hellinger distance, it is equivalent to
\begin{equation}
\label{eq:help-0}
\inf_{t \in [n]} \inf_{\bP_t} \EB_{0} \sqrt{1-\eps_n + \eps_n f_{1,\bP_t} (Y)} = 1 + o (1/n).
\end{equation}
Using the inequality that $|\sqrt{1+x}-1-x/2| \le c \cdot x^2$ for any $x \ge -1$ and plugging $x = \eps_n(f_{1,\bP_t}(Y)-1)$ into it, we have that
\[
\EB_{0} \sqrt{1-\eps_n + \eps_n f_{1,\bP_t} (Y)} \ge 1- c \cdot \EB_{0} x^2 \ge 1- O( |\Voca|) \cdot  \eps_n^2. 
\]
The last inequality uses the result $ \EB_0(f_{1,\bP_t}(Y)-1))^2 =  O(|\Voca|)$ from Lemma \ref{lem:upper-bound-for-square}.
Due to $p > 0.5$, we have $ \eps_n^2 =n^{-2p}= o(1/n)$ and thus complete the proof.
\end{proof}


\subsubsection{Proof of Lemma~\ref{lem:H2}}

\begin{proof}[Proof of Lemma~\ref{lem:H2}]
It suffices to focus on the the Hellinger distance between $\rho_0$ and $\rho_1$ due to the following inequality:
\[
\frac{1}{2} H^2(\rho_0, \rho_1) \le \mathrm{TV}(\rho_0, \rho_1) \le H(\rho_0, \rho_1) \sqrt{1- \frac{1-H^2(\rho_0, \rho_1)}{4}} \le 1,
\]	
As $\rho_0 \equiv 1$, it follows by definition that
\begin{align}
\label{eq:def-Hellinger}
H^2(\rho_0, \rho_1) =1 - \EB_0 \sqrt{\rho_1(\Yars_1, \ldots, \Yars_n)}.
\end{align}
where the expectation $\EB_0$ means $(\Yars_1, \ldots, \Yars_n) \sim \mu_0^n.$
On the other hand, a crucial fact is that
\begin{equation}
\label{eq:help-rho}
\rho_1(\Yars_1, \ldots, \Yars_{n-1}, \Yars_n) = \rho_1(\Yars_1, \ldots, \Yars_{n-1}) \cdot \left[ (1-\eps_n)  + \eps_n \bfPn(\Yars_n)  \right],
\end{equation}
with 
\[
\bfPt(y) = \EB_1[f_{1, \bP_t}(y)|\GM_{t-1}],
\]
where the conditional expectation is taken with respect to $\bP_t$.
Here $\GM_{n} = \sigma(\{\Yars_t\}_{t=1}^n)$ is the $\sigma$-field generated by all $\Yars_1, \ldots,\Yars_n$.
We will prove this equation at the end of the proof.
We define the measure introduced by the PDF $\bfPt$:
for any measurable set $A$,
\begin{equation}
\label{eq:bar-mu-measure}
\bmut(A) = \EB_1[\mu_{1, \bP_t}(A)|\GM_{t-1}] = \int_A \bfPt(y) \rd y.
\end{equation}

We are now ready to prove this lemma.
\begin{enumerate}
\item By conditional Jensen's inequality, it follows that
\begin{align*}
\EB_{\Yars_n \sim \mu_0} \sqrt{(1-\eps_n)  + \eps_n \bfPn(\Yars_n)}
&\ge  \EB_{\Yars_n \sim \mu_0} \EB_1 \left[ 
\sqrt{(1-\eps_n)+\eps_n f_{1,\bP_n}(\Yars_n)}\bigg|\GM_{n-1}
\right]\\
&\ge \inf_{\bP_n \in \PM} \EB_{\Yars_n \sim \mu_0} \sqrt{(1-\eps_n)  + \eps_n f_{1,\bP_n}(\Yars_n)} \\
&=1-\sup_{\bP_n \in \PM} H^2(\mu_0, (1-\eps_n) \mu_0 + \eps_n \mu_{1,\bP_n}).
\end{align*}
By the last inequality, \eqref{eq:help-rho}, and \eqref{eq:def-Hellinger}, it follows that
\[
H^2(\rho_0, \rho_1) \le 1- \prod_{t=1}^n\left(  1 - \sup_{\bP_t \in \PM} {H^2(\mu_0,(1-\eps_n)\mu_0 + \eps_n \mu_{1,\bP_t})}\right).
\]
We prove the first part by using the inequality that $\re^{- 2\log2 \cdot x} \le 1-x$ for any $x \in [0, 1/2]$. 
\item By the condition, it follows that
\begin{align*}
\EB_{\Yars_n \sim \mu_0} \sqrt{(1-\eps_n)  + \eps_n \bfPn(\Yars_n)}
&=1- H^2(\mu_0, (1-\eps_n) \mu_0 + \eps_n \bmun) \le 1- c_n.
\end{align*}
By the last inequality, \eqref{eq:help-rho}, and \eqref{eq:def-Hellinger}, it follows that
\[
H^2(\rho_0, \rho_1) \ge 1- \left(  1 - c_n\right)^n \ge 1- \re^{-nc_n} \to 1.
\]
\end{enumerate}
\end{proof}

\begin{proof}[Proof of \eqref{eq:help-rho}]
To prove equation \eqref{eq:help-rho}, we require an important lemma to understand the dependence of the pivotal statistic $\Yars_t$ on all text-generating randomness. This result is presented in Lemma \ref{lem:expectation-of-Y}. We will reference Lemma \ref{lem:expectation-of-Y} multiple times in the appendix.

\begin{lem}
\label{lem:expectation-of-Y}
Let Assumption \ref{asmp:main} hold with the filtration $\FM_t$ defined therein. It follows that for any integrable function $h$,
\begin{align}
\label{eq:expect}
\EB_1[h(\Yars_t)\mid\FM_{t-1}]=\EB_1[h(\Yars_t)|\bP_t]
&= \eps_n \cdot \EB_{1, \bP_t} h(Y) + (1-\eps_n) \cdot \EB_{0} h(Y) .
\end{align}
\end{lem}
\begin{proof}[Proof of Lemma \ref{lem:expectation-of-Y}]
Condition (b) in Assumption~\ref{asmp:main} specifies the joint distribution of $(\token_t, \zeta_t) \mid \FM_{t-1}$. With probability $\eps_n$, we have $\token_t = \SM(\bP_t, \zeta_t)$; with the remaining probability $1 - \eps_n$, $\token_t$ is drawn from $\bP_t$, independently of both $\FM_{t-1}$ and $\zeta_t$.
By the definition of pivotal statistics, it follows that $\Yars_t \mid \FM_{t-1} \sim (1 - \eps_n)\mu_0 + \eps_n \mu_{1, \bP_t}$. Since this conditional distribution depends only on $\bP_t$, we have $\EB_1[h(\Yars_t) \mid \FM_{t-1}] = \EB_1[h(\Yars_t) \mid \bP_t]$.
Note that we use Condition (a) in Assumption~\ref{asmp:main} in the above argument, which ensures that, in the $1 - \eps_n$ case where $\token_t$ is not watermarked, $\token_t$ and $\zeta_t$ are statistically independent conditional on $\FM_{t-1}$.
\end{proof}

 Recall that $\GM_t =  \sigma(\{\Yars_j\}_{j=1}^t)$ be the $\sigma$-field generated by all pivotal statistics before and including iteration $t$.
Hence, $\GM_{n-1} \subset \FM_{n-1}$ due to $\Yars_t = \Yars(\token_t, \xi_t)$ and $\FM_{t} = \sigma( 
\{\token_j, \xi_j, \bP_{j+1} \}_{j=1}^{t})$.
For a given measurable $A_n$,
\[
\EB_1[\1_{\Yars_n \in A_n}\mid\GM_{n-1}]
=\EB_1[\EB_1[\1_{\Yars_n \in A_n}\mid\FM_{n-1}]|\GM_{n-1}]
=\EB_1[(1-\eps_n)\mu_0(A_n)+\eps_n\mu_{1,\bP_t}(A_n)|\GM_{n-1}],
\]
where the last equation uses Lemma \ref{lem:expectation-of-Y}.
We emphasize that $\EB_1[\mu_{1,\bP_t}(A_n)|\GM_{n-1}]$ should be regarded as a function of $\Yars_1, \ldots, \Yars_{n-1}$ due to the measurability of $\GM_{n-1}$.
It has a closed expression in \eqref{eq:bar-mu-measure} due to Fubini's theorem.

Let $A_1,\ldots, A_{n-1}$ be any measurable sets.
On the one hand, it follows that
\begin{equation}
\label{eq:expresion1}
\PB_1(\Yars_t \in A_t,\forall t \in [n]) = \int_{\prod_{t=1}^nA_t} \rho_1(\Yars_1, \ldots, \Yars_n) \rd \Yars_1 \ldots \rd \Yars_n.
\end{equation}
On the other hand,
\begin{align}
\PB_1(\Yars_t \in A_t,\forall t \in [n])
&=\EB_1 \prod_{t=1}^n \1_{\Yars_t \in A_t}
=\EB_1 \left[\EB_1\left[ \prod_{t=1}^n \1_{\Yars_t \in A_t}|\GM_{n-1}\right] \right] \nonumber\\
&=\EB_1\left[\prod_{t=1}^{n-1} \1_{\Yars_t \in A_t} \EB_1\left[\1_{\Yars_n \in A_n}\mid\GM_{n-1}\right]\right]\nonumber\\
&=\EB_1\left[\prod_{t=1}^{n-1} \1_{\Yars_t \in A_t} 
\left( (1-\eps_n) \mu_0(A_n) + \EB_1[\mu_{1,\bP_t}(A_n)|\GM_{n-1}] \right)
\right]\nonumber\\
&= \int_{\prod_{t=1}^{n-1}A_t} \rho_1(\Yars_1, \ldots, \Yars_{n-1})
\left[ (1-\eps_n) \mu_0(A_n) + \EB_1[\mu_{1,\bP_t}(A_n)|\GM_{n-1}] \right]
\rd \Yars_1 \ldots \rd \Yars_{n-1}  \nonumber\\
&=\int_{\prod_{t=1}^{n-1}A_t} \rho_1(\Yars_1, \ldots, \Yars_{n-1}) \rd \Yars_1 \ldots \rd \Yars_{n-1} \cdot \int_{A_n}[(1-\eps_n)+\eps_n \bfPn(\Yars_n)] \rd \Yars_n.
\label{eq:expersion2}
\end{align}
By the arbitrariness of $A_1, \ldots, A_n$ in \eqref{eq:expresion1} and \eqref{eq:expersion2}, we complete the proof of \eqref{eq:help-rho}.
\end{proof}

\subsection{Proof of Theorem~\ref{thm:gumbel-dense}}
\label{proof:gumbel-dense}

\renewcommand{\thetheorem}{\ref{thm:gumbel-dense}}

\begin{theorem}[Restated version of Theorem \ref{thm:gumbel-dense}]
Under Assumption \ref{asmp:main}, let $0 < p \le \frac{1}{2}$ and $\Delta_n  \asymp  n^{-q}$ with $0 \le q \le 1$.
\begin{enumerate}
\item If $q + 2p > 1$ and  $\bP_{1:n} \subset \PM_{\Delta_n}^c$, $H_0$ and $H_1^{\mathrm{mix}}$ merge asymptotically.
Hence, for any test, the sum of Type I and Type II errors tends to 1 as  $n\to\infty$.
\item If $q + 2p < 1$ and $\bP_{1:n} \subset \PM_{\Delta_n}$, $H_0$ and $H_1^{\mathrm{mix}}$ separate asymptotically.
Furthermore, for the likelihood-ratio test that rejects $H_0$ if the log-likelihood ratio is positive, the sum of Type I and Type II errors tends to 0 as $n\to\infty$.
\end{enumerate}\end{theorem}

\begin{proof}[Proof of Theorem~\ref{thm:gumbel-dense}]
We prove this theorem with the help of Lemma \ref{lem:H2}. 
To make Lemma \ref{lem:H2} easier to apply, we first simplify the considered Hellinger distance to more explicit terms.
\begin{lem} 
\label{lem:square-d}
Fix $\bP_t$. It follows that
\begin{enumerate}
\item Let $f_{1,\bP_t}(r) = \sum_{\token \in \Voca} r^{1/P_{t,\token}-1}$ be the PDF of $\mu_{1,\bP_t}$. Then,
\[
H^2(\mu_0, 1-\eps_n + \eps_n \mu_{1,\bP_t})  =\Theta(1) \cdot \eps_n^2 \cdot \EB_{0}(f_{1,\bP_t}(Y)-1)^2.
\]
\item Let $\GM_{n} = \sigma(\{\Yars_t\}_{t=1}^n)$ is the $\sigma$-field generated by all $\Yars_1, \ldots,\Yars_n$.
We define the conditioned PDF and probability measure: $\bfPt(y) = \EB_1[f_{1, \bP_t}(y)|\GM_{t-1}]$ and $\bmut = \EB_1[\mu_{1, \bP_t}|\GM_{t-1}]$.
It follows that
\begin{gather*}
H^2(\mu_0, 1-\eps_n + \eps_n \bmut)  =\Theta(1) \cdot \eps_n^2 \cdot \EB_{0}(\bfPt(Y)-1)^2.
\end{gather*}
Here $\Theta(1)$ in the above denotes a universal positive constant.
\end{enumerate}
\end{lem}

Now, we are ready to prove the main theorem.
\begin{enumerate}
\item To prove the first point, we could use a similar argument in the proof of Theorem~\ref{thm:gumbel-sparse}.
By Lemma \ref{lem:H2}, it suffices to show
\begin{equation}
\label{eq:upper-bound-H-dense}
\sum_{t=1}^n \sup_{\bP_t \in \PM_{\Delta_n}^c}H^2(\mu_0, (1-\eps_n) \mu_0 + \eps_n \mu_{1,\bP_t}) \to 0.
\end{equation}
\begin{lem}
\label{lem:upper-bound-for-square}
Fix $\bP_t$.
Let $f_{1,\bP_t}(r) = \sum_{\token \in \Voca} r^{1/P_{t,\token}-1}$ be the PDF of $\mu_{1,\bP_t}$. Then,
\[
\EB_{0}(f_{1,\bP_t}(Y)-1)^2 = O(|\Voca|) \cdot (1-P_{t,\max}),
\]
where $O(1)$ denotes a universal positive constant.
\end{lem}
By Lemma \ref{lem:H2} and \ref{lem:upper-bound-for-square}, it suffices to ensure that 
\[
\eps_n^2 \sum_{t=1}^n (1-P_{t,\max}) \le n \eps_n^2 \Delta_n = O( n^{1-2p-q}) \to 0
\]
which follows directly from the condition that $2p+q >1$.

\item To prove the second point, by Lemma \ref{lem:H2}, we only have to show that there exists a positive non-random sequence $c_n$ satisfying $n c_n \to \infty$ and 
\[
H^2(\mu_0, (1-\eps_n) \mu_0 + \eps_n  \bmut) \ge c_n~~\text{almost surely for}~~\forall n \ge 1.
\]
\begin{lem}
\label{lem:lower-bound-for-square}
If $\bP_{t} \in \PM_{\Delta_n}$ holds almost surely and $n$ is sufficiently large so that $\Delta_n < 0.5$, then $\EB_{0}(\bfPt(Y)-1)^2 \ge \Theta(\Delta_n)$ almost surely where $\Theta(1)$ denotes a universal positive constant.
\end{lem}
By Lemma \ref{lem:square-d} and \ref{lem:lower-bound-for-square}, it suffices to show that
\[
n\eps_n^2 \Delta_n = \Theta(n^{1-2p-q}) \to \infty,
\]
which follows directly from the condition that $2p+q < 1$.

It remains to show that the likelihood ratio test is effective if $2p+q < 1$.
Since the proofs are similar, we show only that under the null hypothesis. Lemma~\ref{lem:loglikelihood} implies that under $H_0$, the log-likelihood ratio $L_n$ goes to $-\infty$ with probability one.
Hence, the likelihood ratio test that rejects $H_0$ if $L_n \ge 0$ has a vanishing Type I error.

\begin{lem}
\label{lem:loglikelihood}
Let $\ell_t(y) = \log(1-\eps_n + \eps_n \bfPt(y))$ and $L_n = \sum_{t=1}^n \ell_t(\Yars_t)$. If $\bP_{1:n} \subset \PM_{\Delta_n}$ and $2p+q < 1$,
\[
\EB_0 L_n \to - \infty \quad \text{and} \quad \frac{\Var_0(L_n)}{[\EB_0 L_n]^2} \to 0.
\]
\end{lem}
\end{enumerate}

\end{proof}

\subsubsection{Missing Proofs of Lemmas}
At the end of this section, we provide the missing proofs for the lemmas mentioned above.

\begin{proof}[Proof of Lemma~\ref{lem:square-d}]
It suffices to focus on $\EB_{0} \sqrt{1-\eps_n + \eps_n f_{1,\bP_t} (Y)}$ due to the relation that
\[
 H^2(\mu_0, 1-\eps_n + \eps_n \mu_{1,\bP_t}) = 
1-\EB_{0} \sqrt{1-\eps_n + \eps_n f_{1,\bP_t} (Y)}.
\]

Note that as long as $n$ is sufficiently large, we would have $\sup_{y \in [0, 1]}\eps_n |f_{1,\bP_t} (y)-1| \le |\Voca| \cdot \eps_n \le 1$.
Using the inequality that $\sqrt{1+x} \le 1+\frac{x}{2}-\frac{x^2}{18}$ for any $-1 \le x \le 3$, we have
\begin{align*}
\EB_{0} \sqrt{1-\eps_n + \eps_n f_{1,\bP_t} (Y)}  
&\le \EB_{0} \left[
1 + \frac{\eps_n}{2}(f_{1,\bP_t}(Y)-1) -\frac{ \eps_n^2}{18} (f_{1,\bP_t}(Y)-1)^2
\right]\\
&= 1 - \frac{ \eps_n^2}{18} \cdot  \EB_{0}(f_{1,\bP_t}(Y)-1)^2.
\end{align*}
On the other hand, by the inequality that $\sqrt{1+x} \ge 1+\frac{x}{2}-\frac{x^2}{2}$ for any $x \ge -1$, we have 
\begin{align*}
\EB_{0} \sqrt{1-\eps_n + \eps_n f_{1,\bP_t} (Y)} 
&\ge \EB_{0} \left[
1 + \frac{\eps_n}{2}(f_{1,\bP_t}(Y)-1) - \frac{ \eps_n^2}{2}(f_{1,\bP_t}(Y)-1)^2
\right]\\
&= 1 - \frac{ \eps_n^2}{2} \cdot \EB_{0}(f_{1,\bP_t}(Y)-1)^2.
\end{align*}
We then prove the first part by combining these two inequalities on $\EB_{0} \sqrt{1-\eps_n + \eps_n f_{1,\bP_t} (Y)}$.

By an almost identical argument, we can prove the second part (so we omit it).
\end{proof}

\begin{proof}[Proof of Lemma~\ref{lem:upper-bound-for-square}]
We note that
\begin{align}
\label{eq:evalute-integral}
    \EB_{0}(f_{1,\bP_t}(Y)-1)^2
    &= \EB_{0} f_{1,\bP_t}^2(Y) - 1
    = \int_0^1  \left( \sum_{\token \in \Voca} r^{1/P_{t,\token}-1} \right)^2  \rd r-1 \nonumber \\
    &= \sum_{\token \in \Voca} \sum_{j \in \Voca}  \frac{1}{1/P_{t,\token}+1/P_{t,j}-1}-1 \nonumber \\
    &=\sum_{\token \in \Voca} \sum_{j \in \Voca}  P_{t,\token}P_{t,j} \frac{(1-P_{t,\token})(1-P_{t,j})}{1-(1-P_{t,\token})(1-P_{t,j})}.
\end{align}
Note that $1-(1-P_{t,\token})(1-P_{t,j}) \ge P_{t,\token} \vee P_{t,j}$. It then follows that
\[
\EB_{0}(f_{1,\bP_t}(Y)-1)^2 \le \sum_{1 \le \token,j \le |\Voca|} (P_{t,\token} \wedge P_{t,j}) \cdot (1-P_{t,\token})(1-P_{t,j}).
\]

In the following, without loss of generality, we assume $P_{t,1} \ge P_{t,2} \ge \ldots \ge P_{t,|\Voca|}$.
We start by analyzing the target quantity
\[
\sum_{\token \in \Voca}  \sum_{j=1}^{|\Voca|}
(P_{t,\token}	\wedge P_{t,j}) \cdot(1-P_{t,\token})(1-P_{t,j}).  \]
It follows that that quantity can be equivalently written as 
\begin{align*}
&\sum_{\token \in \Voca}
\left[	\sum_{j=1}^{\token-1}(P_{t,\token} \wedge P_{t,j})\cdot (1-P_{t,\token})(1-P_{t,j})  + \sum_{j=\token}^{|\Voca|}(P_{t,\token} \wedge P_{t,j})\cdot (1-P_{t,\token})(1-P_{t,j})  \right] \\
&\overset{(a)}{=}\sum_{\token \in \Voca}\left[	\sum_{j=1}^{\token-1}(P_{t,\token} \wedge P_{t,j})\cdot (1-P_{t,\token})(1-P_{t,j})  + \sum_{j=1}^{\token}(P_{t,\token} \wedge P_{t,j})\cdot (1-P_{t,\token})(1-P_{t,j})  \right] \\
&\overset{(b)}{=} \sum_{\token \in \Voca}\left[ P_{t, \token}(1-P_{t,\token})^2 + 	2\sum_{j=1}^{\token-1}P_{t,\token} \cdot (1-P_{t,\token})(1-P_{t,j})   \right] \\
&=\sum_{\token \in \Voca} P_{t, \token}(1-P_{t,\token})\left[ 1- P_{t, \token} + 	2\sum_{j=1}^{\token-1}(1-P_{t,j})   \right] \\
&\overset{(c)}{\le} 2|\Voca|  \sum_{\token \in \Voca} P_{t, \token}(1-P_{t,\token}) \\
&\le 4 |\Voca| \cdot (1-P_{t,1}),
	\end{align*}
where $(a)$ uses $\sum_{\token =1}^{|\Voca|} \sum_{j=\token}^{|\Voca|} = \sum_{j =1}^{|\Voca|} \sum_{\token=1}^{j}$ first and then switches the notation of $\token$ and $j$, $(b)$ simplifies the expression by using the fact that $P_{t, j} \ge P_{t, \token}$ if $j \le \token$, and $(c)$ uses the relation that 
$1- P_{t, \token} + 	2\sum_{j=1}^{\token-1}(1-P_{t,j}) = 2\token-1 - (P_{t, \token} + 2\sum_{j=1}^{\token-1} P_{t,j}) \le 2\token-1 \le 2 |\Voca|$ for any $\token \in \Voca$.
\end{proof}

\begin{proof}[Proof of Lemma \ref{lem:lower-bound-for-square}]
We denote the distribution of $\bP_t$ conditional on $\GM_{t-1}$ by $\rho$ for notation simplicity.
Then our target quantity can be written as
\begin{align*}
\EB_{0}(\bfPt(Y)-1)^2 
&= \EB_{Y \sim \mu_0}(\EB_{\bP_t \sim \rho}f_{1, \bP_t}(Y)-1)^2\\
&= \EB_{Y \sim \mu_0}(\EB_{\bP_t \sim \rho}f_{1, \bP_t}(Y))^2 - 1\\
&=\EB_{Y \sim \mu_0} \EB_{\bP_1 \sim \rho}\EB_{\bP_2 \sim \rho}f_{1, \bP_1}(Y)f_{1, \bP_2}(Y)-1\\
&=\EB_{\bP_1 \sim \rho}\EB_{\bP_2 \sim \rho} \left[\EB_{Y \sim \mu_0} f_{1, \bP_1}(Y)f_{1, \bP_2}(Y)-1 \right]
\end{align*}
Note that the following inequality holds (which one can prove by a similar analysis in \eqref{eq:evalute-integral}):
\begin{align*}
\EB_{Y \sim \mu_0} f_{1, \bP_1}(Y)f_{1, \bP_2}(Y)-1
&=
\sum_{\token \in \Voca} \sum_{j \in \Voca}  P_{1,\token}P_{2,j} \cdot \frac{(1-P_{1,\token})(1-P_{2,j})}{1-(1-P_{1,\token})(1-P_{2,j})}\\
&\overset{(a)}{\ge} \sum_{\token \in \Voca} \sum_{j \in \Voca}  P_{1,\token}P_{2,j} \cdot \frac{(1-P_{1,\token})(1-P_{2,j})}{P_{1,\token} + P_{2,j}}\\
&\overset{(b)}{\ge} \frac{1}{2} \sum_{\token \in \Voca} \sum_{j \in \Voca}  \left( P_{1,\token} \wedge P_{2,j} \right) \cdot (1-P_{1,\token})(1-P_{2,j})\\
&\overset{(c)}{\ge} \frac{1}{2} \Delta_n(1-\Delta_n)^2 =  \Theta(\Delta_n),
\end{align*}
where $(a)$ uses $1-(1-P_{1,\token})(1-P_{2,j}) \le P_{1,\token} + P_{2,j}$, $(b)$ uses $\frac{2P_{1,\token} P_{2,j}}{P_{1,\token} + P_{2,j}} \ge P_{1,\token} \wedge P_{2,j}$, and $(c)$ follows from Lemma \ref{lem:lower-bound-for-cross} by requiring $\Delta_n < 0.5$.

\begin{lem}
\label{lem:lower-bound-for-cross}
If $\bP_1, \bP_2 \in \PM_{\Delta}$ where $\Delta < 0.5$, then it follows that
\[
\sum_{\token \in \Voca} \sum_{j \in \Voca}  \left( P_{1,\token} \wedge P_{2,j} \right) \cdot (1-P_{1,\token})(1-P_{2,j}) \ge  \Delta(1-\Delta)^2 +3\Delta^2 (1-\Delta).
\]
\end{lem}
\end{proof}

We provide the proof of Lemma \ref{lem:lower-bound-for-cross} at the end of the proof.
\begin{proof}[Proof of Lemma \ref{lem:lower-bound-for-cross}]
For simplicity, we define $L(\bP_1, \bP_2) = \sum_{\token \in \Voca} \sum_{j \in \Voca}  \left( P_{1,\token} \wedge P_{2,j} \right) \cdot (1-P_{1,\token})(1-P_{2,j})$.
Let $\Pi$ denote any permutation on $\Voca$.
For any fixed $\bP_2$, we define an auxiliary function $h_{\bP_2}(\bP_1): \bP_1 \mapsto L(\bP_1, \bP_2)$.
It is easy to see that the function $h_{\bP_2}(\bP_1)$ is concave in $\bP_1$ for any fixed $\bP_2$.
By a classic result in convex analysis~\citep{HiriartJeLe96}, the infimum of any concave function over a compact convex set in a Euclidean space is necessarily attained at an extreme point of the convex set. 
We note that by Lemma 3.4 in \citep{li2024statistical}, the set of extreme points of $\FPM$ is given by $\left\{\pi(\bP_{\Delta}^{\star}): \pi \in \Pi \right\}$,
where $\bP_{\Delta}^{\star} = (1-\Delta, \Delta, 0, \ldots, 0)$ (due to the fact that $\Delta < 0.5$).
On the other hand, we find that for any permutation $\pi \in \Pi$, $h_{\bP_2}(\cdot)$ is $\pi$-invariant in the sense that $h_{\bP_2}(\bP_1) = h_{\bP_2}(\pi(\bP_1))$ for any $\bP_1$.
Here $\pi(\bP_1)$ denotes the permuted distribution $(P_{\pi(1)}, \ldots, P_{\pi(n)})$.
As a result, it follows that
\[
h_{\bP_2}(\bP_1) \ge \inf_{\bP_1 \in \FPM} h_{\bP_2}(\bP_1) = \inf_{\pi \in \Pi} h_{\bP_2}(\pi(\bP_{\Delta}^{\star})) =h_{\bP_2}(\bP_{\Delta}^{\star}).
\]
Note that $L(\bP_1, \bP_2)=L(\bP_2, \bP_1)$. Repeating the above argument would yield that
\[
L(\bP_1, \bP_2) 
\ge h_{\bP_2}(\bP_{\Delta}^{\star}) 
= L(\bP_2, \bP_{\Delta}^{\star}) \ge 
h_{\bP_{\Delta}^{\star}}(\bP_{\Delta}^{\star}) =
L(\bP_{\Delta}^{\star}, \bP_{\Delta}^{\star}) = \Delta(1-\Delta)^2 +3\Delta^2 (1-\Delta).
\]
\end{proof}

\begin{proof}[Proof of Lemma~\ref{lem:loglikelihood}]
Assume that $n$ is sufficiently large so that $|\Voca|\eps_n  \le 1$.
As a result, it follows that for any $t \in [n]$, $\eps_n \bfPt (\Yars_t) \le 1$.
Using the inequality that $\log(1+x) \le x - x^2/4$ for $x \in (-1, 1]$, we have that
\begin{align}
\label{eq:help-likelihood}
\EB_0 \ell_t(\Yars_t) 
\le \EB_0  \eps_n (\bfPt(\Yars_t)-1) - \frac{\eps_n^2}{4}  \EB_0(\bfPt(\Yars_t)-1)^2
= - \frac{\eps_n^2}{4}  \EB_0(\bfPt(\Yars_t)-1)^2.
\end{align}
By Lemma~\ref{lem:lower-bound-for-square}, when $P_{t,\max} \le 1-\Delta_n$, we have $ \EB_0(\bfPt(\Yars_t)-1)^2 \ge \Theta(1) \cdot \Delta_n$.
Hence, $\EB_0 L_n = \sum_{t=1}^n \EB_0 \ell_t(\Yars_t) \le - \Theta(1) \cdot n \eps_n^2 \Delta_n = - \Theta(1) \cdot n^{1-2p-q} \to - \infty$ if $2p+q < 1$.

To prove the other statement, it suffices to show $\Var_0(L_n) \le C \cdot \EB_0 L_n$ for some constant $C > 0$.
To that end, we will show $\EB_0 \ell_t(\Yars_t)^2 \le C \cdot | \EB_0 \ell_t(\Yars_t)|$.
Using the inequality that $\log^2(1+x) \le 2\cdot x^2$ for $x \in [-0.5, 1]$, as long as $\eps_n \le 0.5$, we have that
\begin{align*}
\EB_0 \ell_t(\Yars_t)^2
\le 2  \eps_n^2 \cdot \EB_0  (\bfPt(\Yars_t)-1)^2 
\le 8 \cdot |\EB_0 \ell_t(\Yars_t)|,
\end{align*}
where the last inequality uses~\eqref{eq:help-likelihood} which implies $\eps_n^2 \cdot \EB_0  (\bfPt(\Yars_t)-1)^2 \le 4 \cdot |\EB_0 \ell_t(\Yars_t)|$.
\end{proof}

\subsection{Higher Criticism and Its Optimal Adaptivity}
\label{appen:HC}
We first introduce Higher Criticism and then show that it achieves optimal adaptivity introduced in Section \ref{sec:optimality}.

\paragraph{Higher Criticism.}
In literature, Higher Criticism (HC), as a non-parametric procedure, has shown success in the sparse detection problem \citep{donoho2004higher,tony2011optimal,cai2014optimal}.
It is an instance of the above \Algo~test when $s=2$.
We introduce it in more detail in the following.
Given the observed pivotal statistics $\{\Yars_t\}_{t=1}^n$, HC contains three steps:
\begin{enumerate}
\item For each $t \in [n]$, we obtain a p-value by $\rp_t := \PB_{0}( Y \ge \Yars_t|\Yars_t) = 1- \Yars_t$.
\item Sort the p-values to $\rp_{(1)} < \rp_{(2)}<\ldots < \rp_{(n)}$. We make a convention that $\rp_{(n+1)}=1$.
\item Define the HC statistic as
\begin{equation}
\label{eq:higher}
\HC_n^{+} = \sup_{t: \rp_{(t+1)}\ge c_n^+} \HC_{n, t},
\qquad 
\HC_{n, t} = \sqrt{n} \frac{t/n - \rp_{(t)}}{\sqrt{\rp_{(t)}(1-\rp_{(t)})}}.
\end{equation}
\end{enumerate}
For any given $\delta>0$, we would reject $H_0$ if 
\begin{equation}
\label{eq:higher-test}
\HC_n^{+} \ge \sqrt{2(1+\delta) \log \log n}.
\end{equation}

\begin{rem}
We mention that except for the expression in~\eqref{eq:higher}, there are other variants or generalizations of HC statistics. See \citep{jager2007goodness,arias2017distribution,donoho2015higher} for further discussions.
\end{rem}
One can check that $S_n^+(2)$ is connected to $\HC_n^{+}$ via the following relationship \citep{jager2007goodness}:
\begin{equation}
\tag{\ref{eq:relation-between-GoFT-HC}}
n \cdot S_n^+(2) = \sup\limits_{r \in [ \rp^+, 1) }  \frac{n}{2}\frac{(\FB_n(r) - r)^2}{r(1-r)} \1_{\FB_n(r) \ge r} = \frac{1}{2}(\HC_n^{+})^2.
\end{equation}

\paragraph{Heavy Tail of HC.}

The observations in~\citep{donoho2004higher,donoho2015higher} suggest that the statistic $\HC_{n, t}$ behaves poorly for very small values of $t$, such as 1 or 2, often leading to extremely large outlier values.
To mitigate this issue, the constraint $\rp_{(t+1)}\ge c_n^+$ (with $c_n^{+}$ a small positive constant) is introduced in the definition of $\HC_n^{+}$. 
This constraint effectively prevents calculating these extreme values without altering the statistic's asymptotic properties (see Appendix \ref{sec:HC-boundary} for numerical support).
When $c_n^{+} = 0$, the constraint disappears and we refer to this special case as $\HC_n^{\star}$.

\paragraph{Optimal adaptivity of HC.}
Theorem \ref{thm:adaptivity} will be crucial in proving the optimal adaptivity of \Algo.

\begin{thm}[Optimal adaptivity of HC]
\label{thm:adaptivity}
Assume Assumption \ref{asmp:main} holds and $\bP_{1:n} \subset \PM_{\Delta_n}$ almost surely with $\Delta_n \asymp n^{-q}$.
If $q + 2p < 1$, for the HC test in~\eqref{eq:higher-test} with any $0 \le c_n^+ \le \Delta_n$, the sum of Type I and Type II errors tends to 0 as $n\to\infty$.
\end{thm}

In Theorem~\ref{thm:adaptivity}, we demonstrate that HC remains effective for robustly detecting employing Gumbel-max watermarks.
As a result, it achieves complete optimal adaptivity as the likelihood ratio test but does not necessitate any knowledge of the sampling token distribution $\bP_t$. 

We explain the intuition of why HC works.
Under $H_0$, $\rp_{t}$'s are i.i.d. copies of $\UM(0, 1)$ so that $\HC_{n, t} \approx \NM(0, 1)$.
Using results from empirical processes, we can show that $\PB_{0}(\HC_n^{+} \le \sqrt{2\log\log n}) \to 1$.
Hence, $\HC_n^{+}$ would grow to the infinity very slowly.
In contrast, under $H_1^{\mathrm{mix}}$, as long as each $\mu_{1,\bP_t}$ differs from $\mu_0$ moderately, we will have $\HC_n^{+} = \Omega_{\PB}(n^{\frac{1}{2}-p-\frac{q}{2}})$.
In fact, $\HC_n^{+}$ would grow to the infinity in a rate faster than $\sqrt{\log\log n}$ as long as $\frac{1}{2}-p-\frac{q}{2} > 0$.
As a result, the HC test can eventually separate two hypotheses.

\begin{proof}[Proof of Theorem~\ref{thm:adaptivity}]
For simplicity, we use $\PB_0(\cdot)$ and $\PB_1(\cdot)$ to denote the probability measure under $H_0$ and $H_1^{\mathrm{mix}}$ respectively.
It suffices to show that as $n \to \infty$, 
\begin{subequations}
\begin{align}
\PB_{0}\left( \HC_n^{+} \ge \sqrt{(2+\delta)\log \log n} \right) &\to 0 \label{eq:higher-prob-H0} \\	
~\text{and}~\PB_{1}\left( \HC_n^{+} < \sqrt{(2+\delta)\log \log n} \right) &\to 0. \label{eq:higher-prob-H1}
\end{align}
\end{subequations}
Recall that the $\rp$-value is given by $\rp_t = 1-\Yars_t$.
Let $\BFB_n$ be the empirical CDF of $\rp_t$'s so that $\BFB_n(r) = \frac{1}{n} \sum_{t=1}^n \1_{ \rp_t\le r}$ for any $r \in [0, 1]$.
Let $U$ denote the CDF of $\UM(0, 1)$ so that $\FB_0(r)=r$.
We let $W_n(r)$ be the standardized form of $\BFB_n(r) -{{r}}$,
\begin{equation}
\label{eq:W}
W_n(r) = \sqrt{n} \cdot \frac{\BFB_n(r) -{{r}}}{\sqrt{ {{r}} (1-{{r}}) }}
= \frac{1}{ \sqrt{n}}  \sum_{t=1}^n  \frac{ \1_{ \rp_t \le r}-{{r}} }{\sqrt{ {{r}} (1-{{r}}) }}.  
\end{equation}
For a given $t \in [n]$, note that exactly $t$ p-values are less than or equal to $\rp_{(t)}$ so that $\BFB_n(\rp_{(t)}) = \frac{t}{n}$ which implies that
\[
W_n(\rp_{(t)}) = \sqrt{n} \cdot \frac{t/n-\rp_{(t)}}{\sqrt{\rp_{(t)}(1-\rp_{(t)})}}.
\]
Note that for any $r \in [\rp_{(n)}, 1]$, $W_n(r) = \sqrt{n} \cdot \sqrt{\frac{1-r}{r}} \ge 0$, which implies that $W_n(r)$ is non-negative and $\sup\limits_{r \in [\rp_{(n)}, 1]} W_n(r)= W_n(\rp_{(n)})$.
Furthermore, it is easy to see that for any $r \in [\rp_{(t)}, \rp_{(t+1)}]$, $W_n(r)$ is decreasing in $r$ so that
$\sup\limits_{r \in [\rp_{(t)}, \rp_{(t+1)}]} W_n(r)= W_n(\rp_{(t)})$.
Varying the value of $r$ in \eqref{eq:W} and letting $\rp^+ = \sup \{\rp_{(t)} : \rp_{(t)} \le c_n^+\}$ be the smallest $\rp$-value that is no smaller than $c_n^+$, we have that
\[
\HC_n^{+} =  \sup_{t :\rp_{t} \ge c_n^+} W_n(\rp_{t-1}) = \sup_{r\in [\rp^+, 1)} W_n(r). 
\]
We introduce another sequence 
\[
\HC_n^{\star} := \sup_{t \in [n]} W_n(\rp_t) = \sup_{r \in [\rp_{(1)}, 1)} W_n(r).
\]
Under $H_0$, $\HC_n^{\star}$ equals in distribution the extreme value of a normalized uniform empirical process (see Theorem 1.1 of \citep{donoho2004higher} for a reason) and thus in probability
\[
\frac{\HC_n^{\star}}{\sqrt{2\log\log n}} \to 1
\quad \text{and} \quad
\PB_{0}\left( \HC_n^{\star} \ge \sqrt{(2+\delta)\log \log n} \right) \to 0.
\]
On the other hand, we have $\HC_n^{\star} \ge \HC_n^{+}$ almost surely.
Hence, for any $\delta > 0$, \eqref{eq:higher-prob-H0} follows directly.

We now consider to prove \eqref{eq:higher-prob-H1}.
Define the $\sigma$-field $\FM_t$ in the same way as in Lemma \ref{lem:expectation-of-Y}.
By Lemma \ref{lem:expectation-of-Y}, it follows that
\[
\EB[\1_{\Yars_t \ge 1- r} \mid\FM_{t-1}] = \EB_{1, \bP_t}\1_{\Yars_t \ge 1- r}  = (1-\eps_n) r + \eps_n[ 1- F_{1,\bP_t}(1-r)].
\]
Using the notation of conditional expectation, we introduce another sequence $\overline{W}_n(r)$ following the spirit of $W_n(r)$:
\[
\overline{W}_n(r)  = \frac{1}{ \sqrt{n}}  \sum_{t=1}^n  \frac{  \EB[\1_{\Yars_t \ge 1- r} \mid\FM_{t-1}]  -{{r}}}{\sqrt{ {{r}} (1-{{r}}) }}
= \frac{ \eps_n}{ \sqrt{n}}  \sum_{t=1}^n \frac{1-F_{1,\bP_t}(1-r)- {{r}}  }{\sqrt{ {{r}} (1-{{r}}) }}.
\]
An observation is that for any fixed $r$, $W_n(r)-\overline{W}_n(r)$ is a martingale difference sequence adapted to the filtration $\FM_n$.

For simplicity, we introduce $x_n = 1-\Delta_n$ and denote $a_n \succsim (\precsim) b_n$ if $\lim\limits_{n \to \infty}\frac{a_n}{b_n} \ge(\le) c$ for an absolute positive constant.
Direct calculations show that 
\begin{align}
\overline{W}_n(\Delta_n) &= \frac{\eps_n}{\sqrt{n} }\sum_{t=1}^n \frac{x_n- F_{1,\bP_t}(x_n)}{\sqrt{x_n(1-x_n)}} \nonumber \\
&\overset{(a)}{\succsim} \sqrt{n}\eps_n \cdot \frac{	(1-x_n)
}{\sqrt{x_n(1-x_n)}} \nonumber  \\
&= \sqrt{n} \eps_n \cdot  \sqrt{\frac{1-x_n}{x_n}}  \nonumber
\\
&= \Theta(1) \cdot  \sqrt{n} \eps_n  \cdot \sqrt{\Delta_n} \nonumber
\\
&\overset{(b)}{=} \Theta(1) \cdot n^{\frac{1}{2}-p-\frac{q}{2}}   \label{eq:lower-bound-conditional-expectation}
\end{align}
where $(a)$ uses the following Lemma \ref{lem:help-ratio} and $(b)$ uses the definition of $\eps_n$ and $\Delta_n$.

\begin{lem}
\label{lem:help-ratio}
Recall that $x_n = 1-\Delta_n$ with $\Delta_n \asymp n^{-q}$ and $F_{1,\bP_t}(r) = \sum_{\token \in \Voca} P_{t,\token} r^{1/P_{t,\token}}$. If $q \in [0, 1)$, then as long as $\bP_t \in \PM_{\Delta_n}$,
\begin{equation}
\label{eq:help-ratio}
1 + \re^{-1} \asymp   x_n^{\frac{\Delta_n}{1-\Delta_n}} + x_n^{\frac{1-\Delta_n}{\Delta_n}} \le \frac{1 -  F_{1,\bP_t}(x_n)}{1-x_n} \le  |\mathrm{supp}(\bP_t)|,
\end{equation}
where $\mathrm{supp}(\bP_t) = \{\token: P_{t, \token} \neq 0\}$ collects non-zero entries in the NTP distribution $\bP_t$.
\end{lem}
\begin{proof}[Proof of Lemma~\ref{lem:help-ratio}]
Note that by the mean value theorem, it follows that
\begin{equation}
\label{eq:help-mean-value-theorem}
\frac{1 -  F_{1,\bP_t}(x_n)}{1-x_n}  =
f_{1,\bP_t}(\theta_n),
\end{equation}
where $f_{1,\bP}$ is the PDF of $F_{1, \bP}$ and $\theta_n \in [x_n, 1]$.

The upper bound follows from the fact that $f_{1,\bP_t}(\theta_n) \le f_{1,\bP_t}(1) \le   |\mathrm{supp}(\bP_t)| $.
We then turn to the lower bound.
By Lemma \ref{lem:worst-P}, for any $\bP_t \in \PM_{\Delta_n}$, it follows that for 
\[
F_{1,\bP_t}(x_n) \le  \sup_{\bP \in \PM_{\Delta_n}}F_{1,\bP}(x_n) = F_{1,\bP_{\Delta_n}^{\star}}(x_n)
\]
where 
\[
\bP_{\Delta}^{\star} =\biggl(\underbrace{1-\Delta, \ldots, 1-\Delta}_{\floor{\frac{1}{1-\Delta}}\ \textnormal{times}}, 1-(1-\Delta)\cdot \left\lfloor\frac{1}{1-\Delta}\right\rfloor, 0, \ldots\biggr).
\]
Hence, once $n$ is sufficiently large so that $\Delta_n < 0.5$, it follows that
\[
\eqref{eq:help-mean-value-theorem} \ge 	  \frac{1 - F_{1,\bP_{\Delta_n}^{\star} }(x_n)}{1-x_n}  \ge   f_{1,\bP_{\Delta_n}^{\star}}(x_n) = x_n^{\frac{\Delta_n}{1-\Delta_n}} + x_n^{\frac{1-\Delta_n}{\Delta_n}} \asymp 1 + \re^{-1}
\]
where the last $\asymp$ is because $\lim\limits_{x \to 0}(1-x)^{\frac{x}{1-x}} = 1$ and $\lim\limits_{x \to 0}(1-x)^{\frac{1-x}{x}} = \re^{-1}$.
\end{proof}
    
    By \eqref{eq:lower-bound-conditional-expectation}, 	as long as $2p+q < 1$, we would have $\overline{W}_n(\Delta_n) \ge C \cdot n^{\frac{1}{2}-p-\frac{q}{2}} \ge \sqrt{(2+\delta)\log\log n}$ for some constant $C > 0$.
    On the other hand, the quadratic variation of the martingale $\{W_t(\Delta_n) - \overline{W}_t(\Delta_n)\}_{t=1}^n $ is given by
    \begin{align}
    \label{eq:quadratic-variation}
     \frac{\sum_{t=1}^n \Var_1[\1_{\Yars_t \ge 1-\Delta_n}\mid\FM_{t-1}]}{nx_n(1-x_n)}
     &= 
    \frac{1}{n}\sum_{t=1}^n \frac{(1-\eps_n) x_n +  \eps_n F_{1,\bP_t}(x_n)}{x_n} \cdot
    \frac{1-(1-\eps_n) x_n - \eps_n F_{1,\bP_t}(x_n)}{1-x_n} \nonumber \\
    &\le |\Voca|.
    \end{align}
    Here we use the inequality $F_{1, \bP}(r) \le r$ for any $r \in [0, 1]$ and the upper bound in Lemma \ref{lem:help-ratio}.
    As a result, by the properties of square integrable martingales,
    \[
    \EB_1 |W_n(\Delta_n) - \overline{W}_n(\Delta_n)|^2 
    \le \EB_1\text{its quadratic variation} \le |\Voca|.
    \]
    Combining these pieces,  it follows from Doob's martingale inequality that 
    \begin{align}
    \PB_1(W_n(\Delta_n) < 0.5C \cdot n^{\frac{1}{2}-p-\frac{q}{2}} ) &\le\PB_1(|W_n(\Delta_n) -  \overline{W}_n(\Delta_n)| \ge 0.5 \cdot \overline{W}_n(x_n) ) \nonumber \\
    &\le\PB_1(|W_n(\Delta_n) -  \overline{W}_n(\Delta_n)| \ge 0.5C \cdot n^{\frac{1}{2}-p-\frac{q}{2}} )\nonumber \\
    &\le \frac{4}{C^2} \cdot \frac{\EB_1 |W_n(\Delta_n) - \overline{W}_n(\Delta_n)|^2 }{n^{1-2p-q}}\nonumber \\
    &  = O( n^{-(1-2p-q)}). \label{eq:lower-bound-Wnxn}
    \end{align}
    As long as $2p+q < 1$, it follows that as $n \to \infty$,
    \[
    \PB_1(W_n(\Delta_n) <\sqrt{(2+\delta) \log\log n })
    \le   \PB_1(W_n(\Delta_n) < 0.5 C \cdot n^{\frac{1}{2}-p-\frac{q}{2}} )
\to 0.
    \]

We complete the proof by noting that $\Delta_n \ge \rp^+$ (otherwise we would have $ \Delta_n < \rp^+ \le c_n^+$ which is contradictory with the condition that $c_n^+ \le \Delta_n$) and thus 
\begin{align*}
\PB_{1}\left( \HC_n^{+} < \sqrt{(2+\delta)\log \log n} \right) 
&\le \PB_1(W_n(\Delta_n) <\sqrt{(2+\delta) \log\log n }) \to 0.
\end{align*}
\end{proof}
\begin{rem}
	\label{rem:proof-HCstar}
	In the proof of Theorem~\ref{thm:adaptivity}, we also proved that
	\begin{align*}
		\PB_{0}\left( \HC_n^{\star} \ge \sqrt{(2+\delta)\log \log n} \right) \to 0 
		\quad\text{and}\quad\PB_{1}\left( \HC_n^{\star} < \sqrt{(2+\delta)\log \log n} \right) \to 0
	\end{align*}
	Hence, the additional constraint $\rp_{(t)} \ge c_n^+$ in the definition of $\HC_n^{+}$ does not affect the asymptotic behaviors of $\HC_n^{\star}$.
\end{rem}

\subsection{Rationale for Truncation}
\label{appen:truncation}
In this subsection, we expand on the discussion in Remark \ref{rem:drawbacks-of-untruncation}. Our definition of $S_n^+$ introduces two truncations relative to the previous $S_n$. As indicated by the relationship in \eqref{eq:relation-between-GoFT-HC}, the untruncated $S_n(s)$ has a heavy-tail issue, as noted in Remark \ref{rem:drawbacks-of-untruncation} and formally stated in Lemma \ref{lem:heavy-tail-GoF}. Consequently, we impose a constraint in \Algo~\eqref{eq:S+} (for a similar rationale, see Appendix \ref{appen:HC}), distinguishing it from the original approach in \citep{jager2007goodness}.
In numerical experiments, excluding these small $\rp$-values substantially reduces the heavy-tail effect, as demonstrated by comparing Figure \ref{fig:other-s-hist} with Figure \ref{fig:other-s-hist-nomask}.

\begin{lem}
\label{lem:heavy-tail-GoF}
For any $n \ge 2$, $\PB_0(n \cdot S_n(2) \ge z) \ge \frac{1}{2z}$ for sufficiently large $z > 0$.
\end{lem}
\begin{proof}[Proof of Lemma \ref{lem:heavy-tail-GoF}]
It follows that
\begin{align*}
\PB_0(n \cdot S_n(2) \ge z) 
&=\PB_0(n \cdot K_2(1/n, \rp_{(1)}) \ge z) \\
&\overset{(a)}{=} \PB_0(\HC_{n,1}^2 \ge 2z ) \\
&\overset{(b)}{\asymp} \PB_0\left( \left| \frac{1}{\sqrt{E}}-\sqrt{E} \right| \ge \sqrt{2z} \right) \\
&\overset{(c)}{\asymp} \frac{1}{2z},
\end{align*}
where $(a)$ uses the definition of $\HC_{n,1}$ from \eqref{eq:higher} and $(b)$ follows from the result that $\HC_{n,1} \overset{d}{\to} \frac{1}{\sqrt{E}}-\sqrt{E}$ where $E$ is exponentially distributed with mean 1.
This is due to the fact that $n \rp_{(1)} \overset{d}{\to} E$ if $\rp_1, \ldots, \rp_n$ are i.i.d. from $\UM(0, 1)$.
Finally, $(c)$ uses the fact that for large $z > 0$,
\begin{gather*}
\PB_0\left( \frac{1}{\sqrt{E}}-\sqrt{E} \ge \sqrt{2z} \right) 
= \PB_0\left(  \sqrt{E} \le \sqrt{\frac{z}{2}+1} - \sqrt{\frac{z}{2}}
\right) \asymp \frac{1}{2z},\\
\PB_0\left( \frac{1}{\sqrt{E}}-\sqrt{E} \le -\sqrt{2z} \right) 
= \PB_0\left(  \sqrt{E} \ge \sqrt{\frac{z}{2}+1} + \sqrt{\frac{z}{2}}
\right) \asymp \exp(-2z).
\end{gather*}


\end{proof}

Another truncation involves using $K_s^+$ (see Definition \eqref{eq:Ks+}) instead of $K_s$ in defining our statistic $S_n^+$. This modification not only facilitates the theoretical analysis in Theorem \ref{thm:optimal-rate} but also allows $\HC$ to become an exact special case of $S_n^+$ through the relation \eqref{eq:relation-between-GoFT-HC}. Generally, truncating $K_s$ to $K_s^+$ has mild impact, as for any small $\delta>0$, $\rp_{(t)} \le \frac{t}{n} +\delta$ holds for all $t \in [n]$ with probability approaching one, as demonstrated in the following lemma.

\begin{lem}
\label{lem:truncation-K}
Assume $\rp_1, \ldots, \rp_n$ are i.i.d. copies from $\UM(0, 1)$ and let $\rp_{(1)}  \le \ldots \le \rp_{(n)}$ denote the ordered values.
It follows that for any $\delta > 0$,
\[
\lim_{n \to \infty}\PB_0(\rp_{(t)} \le t/n + \delta,~\forall t \in [n]) = 1.
\]
\end{lem}
\begin{proof}[Proof of Lemma \ref{lem:truncation-K}]
To prove this lemma, it suffices to show that for any \( \delta > 0 \),
\[
\lim_{n \to \infty} \PB_0\left(\exists\, t \in [n] \text{ such that } \rp_{(t)} \ge \frac{t}{n} + \delta\right) = 0.
\]
Fix $t, n$, and $\delta$.
First, consider the case when \( \frac{t}{n} + \delta \ge 1 \). In this situation, we have \( \PB_0(\rp_{(t)} \ge \frac{t}{n} + \delta) = 0 \) because \( \rp_{(t)} \in [0, 1] \).
Now, consider the case when \( \frac{t}{n} + \delta < 1 \). By Hoeffding's inequality, we have:
\begin{align*}
\PB_0\left(\rp_{(t)} \ge \frac{t}{n} + \delta\right) &= \PB_0\left( \sum_{i=1}^n \1_{\rp_i \ge \frac{t}{n} + \delta} \ge n + 1 - t \right) \\
&\le \exp\left(-2n \left[ 1 - \frac{t - 1}{n} - \left(1 - \frac{t}{n} - \delta\right) \right]^2\right)\le \exp(-2n\delta^2).
\end{align*}
Combining these two cases, we always have \( \PB_0\left(\rp_{(t)} \ge \frac{t}{n} + \delta\right) \le \exp(-2n\delta^2) \) for any $t, n$ and $\delta$.
Using the union bound, it follows that as \( n \to \infty \),
\begin{align*}
\PB_0\left(\exists\, t \in [n] \text{ such that } \rp_{(t)} \ge \frac{t}{n} + \delta\right) &\le \sum_{t=1}^n \PB_0\left(\rp_{(t)} \ge \frac{t}{n} + \delta\right) \le n \exp(-2n\delta^2) \to 0.
\end{align*}
\end{proof}

\subsection{Proof of Theorem \ref{thm:adaptivity-fitness-of-good}}
\label{proof:adaptivity-fitness-of-good}

\renewcommand{\thetheorem}{\ref{thm:adaptivity-fitness-of-good}}
\begin{theorem}[Restated version of Theorem \ref{thm:adaptivity-fitness-of-good}]
\label{thm:adaptivity-fitness-of-good-formal}
Assume Assumption \ref{asmp:main} holds and $\bP_{1:n} \subset \PM_{\Delta_n}$ almost surely with $\Delta_n \asymp n^{-q}$.
If $q + 2p < 1$, for \Algo~with any $0 \le c_n^+ \le \Delta_n$ and $s \in [-1, 2]$, the sum of Type I and Type II errors tends to 0 as $n\to\infty$.
\end{theorem}
\begin{rem}[Proof sketch of Theorem \ref{thm:adaptivity-fitness-of-good}]
We provide the intuition behind Theorem \ref{thm:adaptivity-fitness-of-good} which focuses on the different behaviors of $S_n^+(s)$ under $H_0$ and $H_1^{\mathrm{mix}}$. 
By Theorem 3.1 in \citep{jager2007goodness}, under $H_0$, 
\begin{equation*}
   \lim_{n \to \infty} \PB_0(n \cdot S_n^+(s) \le \log \log n) = 1.
\end{equation*}
This means that $nS_n^+(s)$ increases towards infinity at an exceptionally slow rate under $H_0$. In Appendix \ref{proof:adaptivity-fitness-of-good}, we show that under $H_1^{\mathrm{mix}}$,
\begin{equation*}
   \lim_{n \to \infty} \PB_1(n \cdot S_n^+(s) \ge  n^{\frac{1}{2}-\frac{q}{2}-p}) = 1.
\end{equation*}
This implies that $n S_n^+(s)$ grows faster than $\log\log n$ as long as $p + \frac{q}{2} < \frac{1}{2}$. These results imply that by rejecting $H_0$ whenever $n S_n^+(s) \ge (1+\delta) \log \log n$ for any given $\delta > 0$, both Type I and Type II errors converge to zero asymptotically. As a result, \Algo~can asymptotically distinguish between $H_0$ and $H_1^{\mathrm{mix}}$.
\end{rem}

To prove Theorem \ref{thm:adaptivity-fitness-of-good}, we start by introducing two key lemmas. The first lemma shows how $K_s^+(u,v)$ and $K_2^+(u,v)$ are connected numerically for different values of $s$. Our main focus is to estimate the lower bound of $K_s^+(u,v)$ using $K_2^+(u,v)$. However, it's also possible to establish an upper bound similarly, as shown in Lemma 7.2 of \citep{jager2007goodness}, which is left for interested readers.

\begin{lem}
\label{lem:reduce-Ks-to-K2}
\begin{enumerate}
\item For any $s \le 2$ and $s\neq 1$, it follows that
\[
K_s^+(u,v)  \ge K_2^+(u,v) \cdot \left[ 1 - (1-v) \left(1- \left( \frac{v}{u}\right)^{2-s} \right) \right].
\]
\item For $s=1$, it follows that
\[
K_1^+(u,v)  \ge K_2^+(u,v) \cdot \frac{v}{u}.
\]
\end{enumerate}
\end{lem}
\begin{proof}[Proof of Lemma \ref{lem:reduce-Ks-to-K2}]
If $u < v$, $K_s^+(u,v) =0$; thus, all the inequalities follow directly.
We then assume $0<v \le u < 1$. Note that by definition, $K_2(u,v) = \frac{1}{2} \frac{(u-v)^2}{v(1-v)}$ is always non-negative.

\begin{enumerate}
    \item If $s \neq 1$, from the proof of Lemma 7.2 (ii) in \citep{jager2007goodness}, it follows that
    \[
    K_s(u,v) = K_2(u,v) \left[1 + v(1-v) D_s(u^{\star}, v) -1\right],
    \]
    where $u^{\star} \in [v, u]$ is determined by the mean value theorem and $D_s(u^{\star}, v)$ is given by
    \begin{equation}
    \label{eq:def-D}
        D_s(u, v) = \left(\frac{v}{u}\right)^{2-s} \frac{1}{v} + \left(\frac{1-v}{1-u}\right)^{2-s} \frac{1}{1-v}.
    \end{equation}
    Given $v \le u$ and $s \le 2$, it follows that
    \[
         -(1-v)  \left[ 1-\left( \frac{v}{u}\right)^{2-s} \right] \le v(1-v) D_s(u^{\star}, v) -1 \le v\left[\left( \frac{1-v}{1-u}\right)^{2-s} -1\right].
    \]
    We then complete the proof by noting that $K_s^+(u,v) =K_s(u,v)$ if $v \le u$.
    \item If $s=1$, from the proof of Theorem 1.1 in \citep{wellner2003note}, we similarly have
    \[
    K_1(u,v) = K_2(u,v) \left[1 + v(1-v) D_1(u^{\star}, v) -1\right],
    \]
    where $u^{\star} \in [v, u]$ is determined by the mean value theorem and the quantity $D_1(u^{\star}, v)$, according to \eqref{eq:def-D}, is given by
    \[
    D_1(u^{\star}, v) = \frac{1}{u^{\star}(1-u^{\star})}
    \]
    By the inequality $u^{\star} \in [v, u]$, it follows
    \[
\frac{v}{u}-1 \le   v(1-v) D_1(u^{\star}, v) -1 \le \frac{1-v}{1-u} -1.
    \]
    We then complete the proof by noting that $K_s^+(u,v) =K_s(u,v)$ if $v \le u$.
\end{enumerate}
\end{proof}


The second lemma establishes constant probabilistic upper and lower bounds for $\frac{\Delta_n}{\FB_n(\Delta_n)}$.

\begin{lem}
\label{lem:bounded-ratio}
Let Assumption \ref{asmp:main} hold.
If $q < 1$, there exists a universal constant $0 <c<C \le 1$ such that
\[
\lim_{n \to \infty} \PB_1\left(c \le \frac{\Delta_n}{\FB_n(\Delta_n)}  \le C\right) = 1.
\]
\end{lem}
\begin{proof}[Proof of Lemma \ref{lem:bounded-ratio}]
We define the $\sigma$-field $\FM_t$ in the same way as in Lemma \ref{lem:expectation-of-Y}.
Recall that each $\rp$-value is given by $\rp_{t} = 1 - \Yars_t$ and its empirical CDF is equivalently written as 
\[
\FB_n(r) = \frac{1}{n} \sum_{t=1}^n \1_{\rp_{t} \le r}
= \frac{1}{n} \sum_{t=1}^n \1_{\Yars_t \ge 1- r}.
\]
We introduce an auxiliary CDF: 
\[
\BFB_n(r) =  \frac{1}{n} \sum_{t=1}^n \EB_1[\1_{\Yars_t \ge 1- r}\mid\FM_{t-1}]
= \frac{1}{n} \sum_{t=1}^n  \left[ (1-\eps_n) r + \eps_n \left(1 - F_{1, \bP_t}(1-r) \right)  \right],
\]
where the last equation uses Lemma \ref{lem:expectation-of-Y}.
Therefore, it follows that
\begin{align*}
\frac{\Delta_n}{\BFB_n(\Delta_n)}
&=\frac{1}{ 1-\eps_n + \frac{\eps_n}{n \Delta_n }\sum_{t=1}^n   (1-F_{1, \bP_t}(1-\Delta_n)) }\\
&\overset{(a)}{=}\frac{1}{ 1-\eps_n + \frac{\eps_n }{n}\sum_{t=1}^n\frac{1-   F_{1, \bP_t}(x_n) }{1-x_n}}\\
&\overset{(b)}{\in} (c, C),
\end{align*}
where $(a)$ uses the notation $x_n = 1-\Delta_n$ and $(b)$ follows from Lemma \ref{lem:help-ratio} which shows $\frac{1- F_{1, \bP_t}(x_n) }{1-x_n}$ is bounded above and below by positive constants $c$ and $C$ for sufficiently large $n$.

To finish the proof, it suffices to focus on the concentration of $\frac{\FB_n(\Delta_n)}{\BFB_n(\Delta_n)}-1$.
On the one hand, it follows from \eqref{eq:quadratic-variation} that
\begin{align*}
\EB_1 &|\FB_n(\Delta_n)-\BFB_n(\Delta_n)|^2 
= \frac{1}{n^2}\EB_1 \sum_{t=1}^n \Var_1[\1_{\Yars_t \ge 1-\Delta_n}\mid\FM_{t-1}] \le \frac{|\Voca| \cdot \Delta_n}{n},
\end{align*}
where the last inequality uses Lemma \ref{lem:help-ratio}.
On the other hand, it follows almost surely that
\[
\BFB_n(\Delta_n) \ge (1-\eps_n) \Delta_n.
\]
As a result, it follows that
\[
\EB_1\left|\frac{\FB_n(\Delta_n)}{\BFB_n(\Delta_n)}-1 \right|^2 \le  \frac{|\Voca|}{n(1-\eps_n)^2\Delta_n} \to 0,
\]
where we use the fact that $q < 1$ so that $n \Delta_n \to \infty$ as $n \to \infty$.
It implies that $\frac{\BFB_n(\Delta_n)}{\FB_n(\Delta_n)}$ converges to $1$ in probability, as a result of which, we complete the proof by using Slutsky's Theorem.
\end{proof}
Now we have the tools in place to prove Theorem \ref{thm:adaptivity-fitness-of-good}.
\begin{proof}[Proof of Theorem \ref{thm:adaptivity-fitness-of-good}]
We start by analyzing the Type I error.
Under $H_0$, $\Yars_t$'s are i.i.d. $\UM(0, 1)$ so that $\FB_n(r)$ is the empirical CDF of $n$ i.i.d. $\UM(0, 1)$ random variables.
We first note that $S_n^+$ can be equivalently written as 
\begin{equation*}
S_n^+(s) 
:= \begin{cases}
\sup\limits_{r \in [\rp^+, 1)}  K_s^+(\FB_n(r), \FB_0(r)),&~~\text{if}~~ s \ge 1, \\
\sup\limits_{r \in [\rp_{(1)}, \rp_{(n)}] \cap [\rp^+, 1)}  K_s^+(\FB_n(r), \FB_0(r)),&~~\text{if}~~ s < 1.
\end{cases}
\end{equation*}
where $c_n^+$ is the custom threshold for stability concern and $\rp_{(t)}$ is the $t$th smallest $\rp$-value.
The equivalence follows because we have the following equities when $s<1$ by definition:
\begin{align*}
\sup\limits_{r \in (0, \rp_{(1)})}  K_s^+(\FB_n(r), \FB_0(r)) &=0,\\
\sup\limits_{r \in (\rp_{(n)}, 1) }  K_s^+(\FB_n(r), \FB_0(r))&=
K_s^+(1, \rp_{(n)})=K_s^+(\FB_n(\rp_{(n)}), \FB_0(\rp_{(n)})).
\end{align*}
Under $H_0$, \citep{jager2007goodness} analyzes the null distribution of another related statistic $S_n(s)$, which is defined by
\[
S_n(s):= \begin{cases}
\sup\limits_{r \in (0, 1)}  K_s(\FB_n(r), \FB_0(r)),&~~\text{if}~~ s \ge 1, \\
\sup\limits_{r \in [\rp_{(1)}, \rp_{(n)}]}  K_s(\FB_n(r), \FB_0(r)),&~~\text{if}~~ s < 1.
\end{cases}
\]
Theorem 3.1 in \citep{jager2007goodness} shows that for any $s \in [-1, 2]$, as $n \to \infty$, the following weak convergence holds:
\begin{equation}
\label{eq:weak-convergence-H0}
n S_n(s) - b_n \overset{d}{\to} Z,
\end{equation}
where $b_n = \log\log n + \frac{1}{2}\log\log\log n -\frac{1}{2}\log (4\pi)$ and $\PB(Z \le r) = \re^{-4\re^{-r}}$ for any $r \in \RB$.
Note that $S_n(s) \ge S_n^+(s)$ and $\frac{r_n}{\log\log n} \to 1$.
It then follows that for any $\delta > 0$, as $n \to \infty$,
\[
\PB_0(n S_n^+(s) \ge (1+\delta)\log\log n)
\le \PB_0(n S_n(s) \ge (1+\delta)\log\log n) \to 0.
\]

We then focus on the Type II error.
By the condition $\Delta_n \ge \rp^+$, it follows that
\[
n S_n^+(s) = \sup\limits_{r \in [\rp^+, 1]} n  K_s^+(\FB_n(r), \FB_0(r))
\ge n  K_s^+(\FB_n(\Delta_n), \Delta_n).
\]
If $s \in [-1, 2]$ and $s\neq 1$, it follows that
\begin{align*}
n  K_s^+(\FB_n(\Delta_n), \Delta_n) 
&\overset{(a)}{\ge} n  K_2^+(\FB_n(\Delta_n), \Delta_n) \cdot \left[ 1 - (1-\Delta_n) \left( 1- \left(  \frac{\Delta_n}{\FB_n(\Delta_n)}\right)^{2-s} \right) \right]\\
&\overset{(b)}{\ge} n  K_2^+(\FB_n(\Delta_n), \Delta_n) \cdot  \Omega_{\PB}(1)\\
&\overset{(c)}{\ge} \Omega_{\PB}(n^{1-2p-q}),
\end{align*}
where $(a)$ uses Lemma \ref{lem:reduce-Ks-to-K2}, $(b)$ uses Lemma \ref{lem:bounded-ratio} with $\Omega_{\PB}(1)$ denoting a random variable which is bounded below with probability one, and $(c)$ uses the relation that $n  K_2^+(\FB_n(\Delta_n), \Delta_n) = \frac{1}{2} W_n^2(\Delta_n)$ with $W_n$ defined in \eqref{eq:W} and the lower bound for $W_n(\Delta_n)$  in \eqref{eq:lower-bound-Wnxn}.

From the last inequality, we know that with probability one, $n  K_s^+(\FB_n(\Delta_n), \Delta_n) \to \infty$ as long as $q+2p < 1$.
The analysis for the case where $s=1$ is almost the same and thus omitted.
\end{proof}

\subsection{Proof of Theorem \ref{thm:suboptimality}}
\label{proof:suboptimality}
In this section, we will prove Theorem \ref{thm:suboptimality}.
To that end, we need first to figure out the expectation gaps of different score functions.

\subsubsection{Expectation Gap}
\label{sec:expectation-gap}

Let $\Sars_n  = \sum_{t=1}^n \hars(\Yars_t)$ denote the sum of scores.
Aaronson~\citep{scott2023watermarking} argued that if $\token_{1:n}$ is LLM-generated, then
\begin{equation}
	\label{eq:aaronson-inequality}
	\EB_1 \Sars_n \ge n + \left(\frac{\pi^2}{6}-1\right) \sum_{t=1}^n \EB_1\Ent(\bP_t),
\end{equation}
where $\Ent(\bP_t)$ is the Shannon entropy defined by $\Ent(\bP_t):= -\sum_{\token \in \Voca} P_{t,\token}\log{P_{t,\token}}$.
Here the expectation $\EB_1(\cdot)$ is taken with respect to the token distributions $\bP_1, \ldots, \bP_n$. 
In his talk \citep{scott2023watermarking}, he did not furnish a proof for it. 
We prove the lower bound and establish a new upper bound in the following proposition.
We note that while this work \citep{fernandez2023three} attempts to offer a proof for \eqref{eq:aaronson-inequality}, the inequality they demonstrate differs from that specified in \eqref{eq:aaronson-inequality} because they use a different notion of entropy.

\begin{prop}
	\label{prop:S0}
	\[
	n +  \sum_{t=1}^n \EB_1\Ent(\bP_t) \ge 
	\EB_1 \Sars_n 
	\ge n + \left(\frac{\pi^2}{6}-1\right) \sum_{t=1}^n \EB_1\Ent(\bP_t).
	\]
\end{prop}

\begin{proof}[Proof of Proposition~\ref{prop:S0}]
	We first evaluate the following integral
	\begin{align*}
		- \frac{1}{p}\int_0^1  r^{1 / p-1} \log (1-r) \rd r 
		&\overset{(a)}{=}\int_0^1 \frac{1-r^{1 / p}}{1-r} \rd r \\
		& \overset{(b)}{=}\sum_{j=1}^{\infty} \left( \frac{1}{j}-\frac{1}{j+1 / p}\right) \\
		&\overset{(c)}{=}   p + \psi(1/p)-\psi(1)\\
		&\overset{(d)}{=}\psi(1+1/p)-\psi(1)
	\end{align*}
	where $(a)$ uses integration by parts, $(b)$ applies Taylor's expansion to $(1-r)^{-1}$, $(c)$ uses the definition of the digamma function that is $\psi(x):= -\gamma  + \sum_{j=0}^\infty \left( \frac{1}{j+1}-\frac{1}{j+x} \right)$ for any $x > 0$ with $\gamma$ the Euler–Mascheroni constant, and $(d)$ uses the property that $\psi(1+x) = \psi(x) + \frac{1}{x}$ for any $x > 0$.

	Recall that an important property of the digamma function from Theorem 2 $(a)$ in \citep{farhangdoost2014new} is that for any $x\ge 1$,\footnote{
		This inequality deviates slightly from the original version (see Theorem 2 $(a)$ in \citep{farhangdoost2014new}). 
		It can be established using the same methodology they employed, with the alteration being the adjustment of the domain from $u \ge 1$ to $u \ge 2$, where $x$ is defined according to their context.
		In this way, their inequality (20) becomes $2-\frac{1}{\psi'(2)} \le (\psi')^{-1}(1/t)-t < 1/2$ where $t=1/\psi'(u)$ and $u \ge 2$.
	} 
	\[
	\frac{1}{x + \frac{1}{\frac{\pi^2}{6}-1}-1} \le \psi'(1+x)  \le \frac{1}{x + \frac{1}{2}},
	\]
	which implies that for any $x \ge 1$,
	\[
	\left(\frac{\pi^2}{6}- 1 \right)\log x \le \psi(1+x)-\psi(2) \le \log x.
	\]
	By setting $x = 1/P_{t,\token}$ and summing over $\token \in \Voca$, we have
	\begin{align*}
		-\EB_{1,\bP_t} \log (1-U_{t,\token_t}) \rd r -1  &= -  \int_0^1  \sum_{\token \in \Voca} r^{1 / P_{t,\token}-1} \log (1-r) \rd r -1\\
		&= \sum_{\token \in \Voca} P_{t,\token} \left[ \psi\left(1 + \frac{1}{P_{t,\token}} \right)-\psi(2)\right] \\
		&\in \left[  \left(\frac{\pi^2}{6}- 1 \right) \Ent(\bP_t), \Ent(\bP_t)\right],
	\end{align*}
	where $\Ent(\bP_t) = \sum_{\token \in \Voca} P_{t,\token} \log\left( \frac{1}{P_{t,\token}} \right)$ is the Shannon entropy of $\bP_t= (P_{t,1}, \ldots, P_{t, |\Voca|})$.
\end{proof}

Proposition \ref{prop:S0} essentially provides a tight bound for the expectation gap of $\hars$, that is, 
\[ 
\EB_{1,\bP} \hars(Y) - \EB_{0} \hars(Y) = \Theta(1) \cdot \Ent(\bP)
\]
where $\Theta(1)$ denotes a universal constant.
In the following, we provide tight bounds (up to constant factors) of the expectation gaps for other score functions.

\begin{prop}
\label{prop:expectation-gaps}
Let $\Delta = 1- P_{\max}$ be the gap between 1 and the largest probability in $\bP$. Once $\Delta$ is smaller than a universal constant $c \in (0, 1)$, it follows that
\begin{enumerate}
\item For $\hars$, $\EB_{1,\bP} \hars(Y) - \EB_{0} \hars(Y) = \Theta(1) \cdot \Ent(\bP) = {\Theta}\left(\Delta \log \frac{1}{\Delta}\right)$.
\item For $\hlog$, $\EB_{1,\bP} \hlog(Y) - \EB_{0} \hlog(Y) =1 -\sum_{\token \in \Voca} P_{\token}^2 =\Theta(\Delta)$.
\item For $\hind$ with $\delta \in (0, 1)$,  $\EB_{1,\bP} \hind(Y) - \EB_{0} \hind(Y) = \delta - F_{\bP}(\delta) =\Theta(\Delta)$.
\item For $\hoptarso$ with $\Delta_0 \in (0, 1)$,  $\EB_{1,\bP} \hoptarso(Y) - \EB_{0} \hoptarso(Y) =\Theta(\Delta)$.
\end{enumerate}
Here $\Theta(1)$ denotes a positive constant that depends on $c$, the log factor $\log |\Voca|$ if we consider $\hars$, $\delta$ if we consider $\hind$, and $\Delta_0$ if we consider $\hoptarso$.
\end{prop}

In fact, we have the following general result for the expectation gap of increasing score functions.
\begin{lem}
\label{lem:general-expectation-gap}
We say a score function $h$ is parameter-free if for any $y \in [0, 1]$, $h(y)$ does not depend on $\Delta$ and $\eps$.
Let $\Delta = 1- P_{\max}$ be the gap between 1 and the largest probability in $\bP$.
For any parameter-free score function $h: [0, 1]\to \RB$, as long as it is increasing and satisfies $\int_0^1 y \log \frac{1}{y} \rd h(y) > 0$ and $h(1)-h(0) < \infty$, it follows that for sufficiently small $\Delta$,
\[
\EB_{1,\bP} h(Y) - \EB_{0} h(Y) = \Theta(\Delta).
\]
\end{lem}
\begin{rem}
Since $\hars(1) = \infty$, $\hars$ does not satisfy the condition in Lemma \ref{lem:general-expectation-gap}, resulting in an expectation gap of ${\Theta}\left(\Delta \log \frac{1}{\Delta}\right)$ rather than ${\Theta}\left(\Delta\right)$.
\end{rem}
\begin{rem}
If $h$ is non-decreasing, non-constant, and does not have discontinuities at both $0$ and $1$, then we have both (i) $h(1) - h(0) < \infty$ and (ii) $\int_0^1 y \log \frac{1}{y} \, \mathrm{d}h(y) > 0$. We prove (ii) as follows.

As a non-decreasing function, $h$ introduces a measure on $[0, 1]$, which we denote by $\mu_{h}$. Since $0$ and $1$ are not discontinuous points of $h$, $\mu_{h}$ assigns all of its probability mass on $(0, 1)$. Therefore, there must exist an interval $(a, b)$ satisfying $0 < a < b < 1$ and $\mu_h(a, b) = h(b) - h(a-) > 0$. 
Thus, we have:
\[
\int_0^1 y \log \frac{1}{y} \, \mathrm{d}h(y) \ge \min_{y \in [a, b]} y \log \frac{1}{y} \cdot \mu_h(a, b) > 0.
\]
\end{rem}
\begin{proof}[Proof of Lemma \ref{lem:general-expectation-gap}]
Using the integration by parts, we have that
\begin{align*}
\EB_{1,\bP} h(Y) - \EB_{0} h(Y) = \int_0^1 \left[y - F_{1, \bP}(y) \right] \rd h(y)
\end{align*}
where $F_{1, \bP}$ is the conditional CDF of watermarked $\Yars$ given $\bP$.
Let $P_{\max}$ denote the largest probability in $\bP$. We then have that (see the Point (c) in Proposition \ref{prop:expectation-gaps}):
\[
 y- y^{1/P_{\max}} \le y - F_{1, \bP}(y) \le y- P_{\max} y^{1/P_{\max}} \le \Delta + P_{\max} ( y- y^{1/P_{\max}} ).
\]
It then suffices to show 
\[
\int_0^1 (y- y^{1/P_{\max}}) \rd h(y) = \Theta(\Delta).
\]
We introduce an auxiliary function $J: [0, 1] \to \RB$ defined by
\[
J(p) := \int_0^1 (y- y^{1/p}) \rd h(y).
\]
We note that $J(1) = 0$ and $J'(1) = - \int_0^1 y \log \frac{1}{y} \rd h(y)$.
Because the score function $h$ is parameter-free, we conclude that $J'(1)$ depends on neither $\Delta$ nor $\eps$.
Furthermore, by assumption, $J'(1) < 0$.
Hence, it follows from Taylor's expansion that
\[
J(P_{\max})= J(1) + J'(1)(P_{\max}-1) + O(1) \cdot (1-P_{\max})^2
= (-J'(1)) \cdot \Delta + o(\Delta)
\]
As long as $\Delta$ is sufficiently small, we would have $J(P_{\max}) = \Theta(\Delta)$.
\end{proof}

Finally, we provide the proof of Proposition \ref{prop:expectation-gaps} below.
\begin{proof}[Proof of Proposition \ref{prop:expectation-gaps}]
This proposition mainly follows from the following inequalities.
\begin{enumerate}
\item We first note that
\begin{align*}
\Ent(\bP)  &= P_{\max} \log \frac{1}{P_{\max}} + (1-P_{\max}) \log \frac{1}{1-P_{\max}} \\
&\qquad + (1-P_{\max}) \sum_{\token : P_{\token}\neq P_{\max}} \frac{P_\token }{1-P_{\max}}\log \frac{1-P_{\max}}{P_{\token}}.
\end{align*}
It is easy to find the lower bound holds: $\Ent(\bP)  \ge (1-P_{\max}) \log \frac{1}{1-P_{\max}}  = \Delta \log \frac{1}{\Delta}$.
For the upper bound, we note that if $P_{\max} \ge 1-c$ (due to $\Delta \le c$),
\begin{align*}
\Ent(\bP)  
&\le \frac{\Delta}{c} +\Delta \log \frac{1}{\Delta} + \Delta \log (|\Voca|-1) = {\Theta}\left(\Delta \log \frac{1}{\Delta}\right).
\end{align*}
\item It follows that $1-P_{\max} \le 1 -\sum_{\token \in \Voca} P_{\token}^2  \le 1-P_{\max}^2 \le 2 (1-P_{\max})$.
\item Due to $P_{\max} \delta^{1/P_{\max}} \le F_{\bP}(\delta) \le \delta^{1/P_{\max}} $, once we set $g(x) = \delta^x$, the mean value theorem implies that
\[
g(1) - g\left(\frac{1}{P_{\token}}\right) = g'(\theta) \left(1-\frac{1}{P_{\token}}\right) \overset{(*)}{=} \Theta(1-P_{\max})
\]
where $\theta \in [1, 1/P_{\max}]$.
Given $P_{\max}$ is smaller than a constant, say $c$, we have that $-g'(\theta)$ is a positive constant that depends only on $\delta$ and $c$, which implies the above equation $(*)$.
\item We first consider the simplest case where $\Delta_0 \in (0, 0.5)$,
\[
\hoptarso(y) = \log\left(y^{\frac{\Delta_0}{1-\Delta_0}} + y^{\frac{1}{\Delta_0}-1} \right)
=\frac{\Delta_0}{1-\Delta_0} \log y + \log \left(1 + y^{\frac{1-\Delta_0}{\Delta_0}-\frac{\Delta_0}{1-\Delta_0}} \right).
\]
We note that by integration by parts, it follows that
\begin{align*}
\EB_{1,\bP} \hoptarso(Y) - \EB_{0} \hoptarso(Y) 
&= \frac{\Delta_0}{1-\Delta_0} \left[ \EB_{1,\bP} \hlog(Y) - \EB_{0} \hlog(Y)  \right]\\
&\qquad +  \int_0^1 \frac{y - F_{\bP}(y)}{1 + y^{\frac{1-\Delta_0}{\Delta_0}-\frac{\Delta_0}{1-\Delta_0}}} \rd y^{\frac{1-\Delta_0}{\Delta_0}-\frac{\Delta_0}{1-\Delta_0}}\\
&= \frac{\Delta_0}{1-\Delta_0} \cdot \Theta(\Delta) +  \Theta(\Delta).
\end{align*}
where the second equation uses the result in Point (b) and the inequality that $0 \le y - F_{\bP}(y) \le \Theta(\Delta) \cdot \ln\frac{1}{y}$ (which is already proved in Point (c)).

We complete the proof by noting that the above argument can be extended to the general case where \(\Delta_0 \in (0, 1)\).
\end{enumerate}
\end{proof}

\subsubsection{Failure of Existing Sum-based Detection Rules}

\renewcommand{\thetheorem}{\ref{thm:suboptimality}}
\begin{theorem}[Restated version of Theorem \ref{thm:suboptimality}]
\label{thm:suboptimality-formal}
Let Assumption \ref{asmp:main} hold.
Consider the detection rule specified by $h$: $T_h(\Yars_{1:n}) = 1$ if $\sum_{t=1}^n h(\Yars_t) \ge n \cdot \EB_0 h(\Yars) +\Theta(1) \cdot n^{\frac{1}{2}} a_n$, otherwise it equals 0, where $a_n \to \infty$ and $\frac{a_n}{n^{\gamma}} \to 0$ for any $\gamma > 0$.
For any score function that is \rm{(i)} non-decreasing, \rm{(ii)} non-constant, \rm{(iii)} parameter-free, and \rm{(iv)} does not have discontinuities at both $0$ and $1$, the following results hold for $T_h$:
\begin{enumerate}
\item If $q + p < \frac{1}{2}$ and $\bP_{1:n} \subset \PM_{\Delta_n}$, the sum of Type I and Type II errors tends to 0.
\item If $q + p > \frac{1}{2}$ and $\bP_{1:n} \subset \PM_{\Delta_n}^c$, the sum of Type I and Type II errors tends to 1.
\end{enumerate}
\end{theorem}

\begin{proof}[Proof of Theorem \ref{thm:suboptimality}]
Recall that the considered detection rule has the following form:
\begin{equation*}
T_h(\Yars_{1:n}) = 
\begin{cases}
1 & ~\text{if}~\sum_{t=1}^n h(\Yars_t) \ge n \cdot \EB_0 h(\Yars) + C \cdot n^{\frac{1}{2}} a_n,\\
0 &~\text{if}~\sum_{t=1}^n h(\Yars_t) < n \cdot \EB_0 h(\Yars) + C \cdot n^{\frac{1}{2}} a_n,
\end{cases}
\end{equation*}
where $\{a_n\}$ is a positive sequence satisfying $a_n \to \infty$ and $\frac{a_n}{\sqrt{n}} \to 0$.

By the choice of the sequence $\{a_n\}$, the Markov inequality implies that the Type I error converges to zero as $n \to \infty$,
\begin{align*}
\PB_0(T_h(\Yars_{1:n})=1)
&=\PB_0\left(\sum_{t=1}^n  [h(\Yars_t) - \EB_0 h(\Yars)] \ge C \cdot n^{\frac{1}{2}} a_n \right) \le \frac{\Var_0(h(\Yars))}{ C^2 a_n^2} \to 0.
\end{align*}

We then focus on the Type II error. 
By Lemma \ref{lem:expectation-of-Y}, we have that
\[
\EB_{1}[h(\Yars_t)\mid\FM_{t-1}] = (1-\eps_n) \EB_0 h(\Yars_t) + \eps_n \EB_{1, \bP_t}h(\Yars_t).
\]
Therefore, the Type II error can be equivalently written as
\[
\PB_1(T_h(\Yars_{1:n})=0)
= \PB_1 \left(   \frac{X_n}{\sqrt{n}} \le   \frac{\overline{X}_n }{\sqrt{n}}  + C a_n   \right),
\]
where we denote by
\begin{gather*}
X_n =  \sum_{t=1}^n  \left( h(\Yars_t)- \EB_1[h(\Yars_t)\mid\FM_{t-1}] \right)
~~\text{and}~~\overline{X}_n= \sum_{t=1}^n  \left( \EB_0 h(\Yars_t)- \EB_{1} [h(\Yars_t)\mid\FM_{t-1}]\right) .
\end{gather*}
By the construction of an embedded watermark, we know that $X_n$ is a square-integrable martingale under $H_1^{\mathrm{mix}}$.
For the considered score functions, given that the alternative CDF of $h(\Yars_t)$ is continuous in $\bP$, we know its conditional variance is also continuous in $\bP$ such that it is bounded uniformly over $\bP \in \Simplex(\Voca)$.
It implies that there exists some $C > 0$ so that for any $n \ge 1$,
\[
\EB_1 |X_n|^2 \le C \cdot n.
\]
On the other hand, by Lemma \ref{lem:general-expectation-gap}, it follows that
\[
\overline{X}_n = -\Theta\left(\eps_n\sum_{t=1}^n (1-P_{t,\max}) \right).
\]

\begin{enumerate}
\item If $p+q < \frac{1}{2}$ and $\bP_{1:n} \subset \PM_{\Delta_n}$, we then have that
\[
\overline{X}_n \le - \Theta(n \cdot\eps_n\cdot \Delta_n).
\]
By condition $\frac{a_n}{n^{1/2-p-q}} \to 0$, we have $\frac{\overline{X}_n }{\sqrt{n}} \to - \infty$ as $n \to \infty$.
Hence, by Chebyshev's inequality, it follows that as long as $n$ is sufficiently large,
\begin{align*}
\PB_1(T_h(\Yars_{1:n})=0)
\le\PB_1 \left(   \left|\frac{X_n}{\sqrt{n}}\right|  \ge    \left|\frac{\overline{X}_n }{\sqrt{n}} \right| - C a_n   \right)
\le \frac{O(1)}{n^{\frac{1}{2}-p-q}} \to 0.
\end{align*}
\item  If $p+q > \frac{1}{2}$ and $\bP_{1:n} \subset \PM_{\Delta_n}^c$, we then have that
\[
0 \ge \overline{X}_n \ge - \Theta(n \cdot \eps_n \cdot\Delta_n ),
\]
which implies that $\frac{\overline{X}_n }{\sqrt{n}} \to 0$ as $n \to \infty$. 	Hence, by Chebyshev's inequality, it follows that as long as $n$ is sufficiently large,
\begin{align*}
\PB_1(T_h(\Yars_{1:n})=1)
\le\PB_1 \left(   \left|\frac{X_n}{\sqrt{n}}\right|  \ge   C a_n -  \left|\frac{\overline{X}_n }{\sqrt{n}} \right|   \right)
\le \frac{O(1)}{a_n^2} \to 0.
\end{align*}
\end{enumerate}
The case for $\hars$ where $\overline{X}_n = -\Theta\left( \sum_{t=1}^n \eps_n (1-P_{t,\max})  \log \frac{1}{1-P_{t,\max}} \right)$ can be analyzed similarly and thus omitted.
\end{proof}

\subsection{Proof of Theorem \ref{thm:optimal-rate}}
\label{proof:rate}

\begin{rem}[Proof sketch of Theorem \ref{thm:optimal-rate}]
If all the underlying NTP distributions $\bP_t$'s are identical (e.g., by setting $\PM$ as a singleton class), the $\PM$-efficiency reduces to the Hodges-Lehmann asymptotic efficiency. When $s \in (0, 1)$ and $c_n^+ = 0$, the asymptotic behavior of $S_n^+(s)$ closely resembles that of the Kolmogorov–Smirnov test. 
Since the Hodges-Lehmann asymptotic efficiency of the Kolmogorov–Smirnov test is well-established in the literature \citep{nikitin1987hodges}, in our proof,  we address the heterogeneity of the NTP distributions \(\bP_t\)'s by reducing them to the least-favorable NTP distribution, \(\bP_{\Delta}^{\star}\), and apply this established result \citep{nikitin1987hodges} to analyze \(R_{\PM_{\Delta}}(S_n^+(s))\). It is also worth noting that the conditions $s \in (0, 1)$ and $c_n^+ = 0$ are also required by the conventional analysis of Bahadur efficiency for $S_n^+(s)$ (see Theorem 4.4 in \citep{jager2007goodness}).
Whether and how to relax these conditions is still open even in the original line of research.
\end{rem}

\begin{proof}[Proof of Theorem \ref{thm:optimal-rate}]
By the definition of the $\FPM$-efficiency (in Definition \ref{def:efficiency}), we should analyze the rate of the exponential decrease of the Type II error $\PB_1(S_n^+(s) \le \gamma_{n, \alpha})$ where $\gamma_{n, \alpha}$ is the critical value satisfying $\PB_0(S_n^+(s) \ge \gamma_{n, \alpha})=\alpha$.
By Theorem \ref{thm:adaptivity-fitness-of-good}, it follows that $\gamma_{n, \alpha} = O(\frac{\log\log n}{n})$.
It implies that for any small $\epso > 0$, as long as $n$ is sufficiently large, we have
\begin{equation}
\label{eq:rate-help1}
\PB_1(S_n^+(s) \le \gamma_{n, \alpha})
\le \PB_1(S_n^+(s) \le \epso).
\end{equation}

Subsequently, we focus on the probability $\PB_1(S_n^+(s) \le \epso)$. To this end, we need two auxiliary lemmas. The first one, Lemma \ref{lem:simplify-K}, bounds the event $\{S_n^+(s) \le \epso\}$ with a more manageable one.
\begin{lem}
\label{lem:simplify-K}
Let $s \in (0, 1)$ and $c_n^+=0$.
For any $\epso \to 0$, there exists $\deltao \to 0$ so that the following event inclusion holds:
\[
\{ S_n^+(s) \le \epso \} \subseteq \{ \FB_n(r) \le r + \deltao,~\forall r \in [0, 1] \}
\]
where $\FB_n(r) = \frac{1}{n} \sum_{t =1}^n \1_{\rp_t \le r}$ is the empirical CDF of observed $\rp$-values (where $\rp_t = 1-\Yars_t$).
\end{lem}
\begin{proof}[Proof of Lemma \ref{lem:simplify-K}]
When $s \in (0, 1)$, $K_s(u, v)$ is jointly continuous in the closed domain $[0, 1]^2$, so it is uniformly continuous.
By defining an auxiliary function $J$ by $J(x) = \sup_{v \in [0, 1]} K_s(v+x, v)$, we have that $J$ is a continuous function.

On the other hand, $K_s(u, u)=0$ and $\frac{\partial}{\partial u} K_s(u, v) = \frac{1}{1-s} \left[\left(  \frac{1-v}{1-u}\right)^{1-s} - \left( \frac{v}{u}\right)^{1-s} \right] \ge 0$ if $u \ge v$.
It's easy to see that $K_s(\cdot, v)$ is strictly increasing on $[v, 1]$.
As a result, $J$ is an increased function on $[0, 1]$.
As $J(0) = 0$, for any $\epso \to 0$, there exists $\deltao \to 0$ so that $\{x\ge 0: J(x) \le \epso \} = \{x \ge 0: x \le \deltao \}$. In fact, we should have $\deltao=J(\epso)$.

We aim to analyze the event $\{ S_n^+(s) \le \epso \}$ for a given $\epso$.
As $c_n^+=0$, we have that $S_n^+(s)  = \sup\limits_{r \in [0, 1] }  K_s^+(\FB_n(r), r) = \sup\limits_{r \in [0, 1] }  K_s(\FB_n(r), r) \1_{\FB_n(r) \ge r}$.
It then follows that $\{ S_n^+(s) \le \epso \} 
\subseteq \{\FB_n(r)-r \le \deltao,~\forall~r \in [0, 1] \}$.
\end{proof}

The second lemma examines the superiority of $\bP_t$ over $\FPM$ in terms of the CDF value.

\begin{lem}
\label{lem:worst-P}
Let $F_{1,\bP}(r) = \mu_{1,\bP}(\Yars \le r) = \sum_{\token \in \Voca} P_{\token} r^{1/P_{\token}}$.
It then follows that for any $r \in [0, 1]$,
\[
\sup_{\bP \in \FPM} F_{1,\bP}(r) = F_{1, \bP_{\Delta}^\star}(r)
\]
where $\bP_{\Delta}^\star$ is defined in \eqref{eq:worst-case-P}. 
\end{lem}
\begin{proof}[Proof of Lemma \ref{lem:worst-P}]
Note that $F_{1,\bP}(r)$ is convex in $\bP$ in any $r\in[0, 1]$. 
A similar argument of Lemma 3.3 in \citep{li2024statistical} can prove this lemma.
\end{proof}

Recall that $\rp_t = 1 - \Yars_t$.
By Lemma \ref{lem:expectation-of-Y}, we know that the CDF of $\rp_t$ conditioned on $\FM_{t-1}$ depends only on $\bP_t$. We denote the CDF of $\rp_t$ conditioned on $\FM_t$ by $G_{\bP_t}$.
It then follows from Lemma \ref{lem:expectation-of-Y} that for any $r \in [0, 1]$,
\begin{align*}
G_{\bP_t}(r) = \PB_1(\Yars_t \ge 1- r|\bP_t)
=(1-\eps) r + \eps[1- F_{1,\bP_t}(1-r)].
\end{align*}
By Lemma \ref{lem:worst-P}, it follows that $G_{\bP_t}(r)\ge G_{\bP_{\Delta}^{\star}}(r)$ for any $r \in [0,1]$ and $t \ge 1$ (since $\bP_t \in \FPM$ by the definition of $\FPM$-efficiency).
Note that each $G_{\bP}$ is strictly increasing and continuous, then its inverse exists uniquely which we denote by $G_{\bP}^{-1}$.
It follows that for any $r \in [0, 1]$,
\begin{equation}
\label{eq:inverse-G}
G_{\bP_t}^{-1}(r) \le G_{\bP_{\Delta}^{\star}}^{-1}(r).
\end{equation}

Now, we are ready to prove this theorem.
\begin{align*}
\PB_1(S_n^+(s) 
\le \gamma_{n, \alpha})
&\overset{(a)}{\le} \PB_1(S_n^+(s) \le \epso) \\
&\overset{(b)}{\le} \PB_1(\FB_n(r)\le r +\deltao,~\forall~r \in [0, 1]) \\
&= \PB_1\left(\FB_n\left(G_{\bP_{\Delta}^{\star}}^{-1}(r)\right)\le G_{\bP_{\Delta}^{\star}}^{-1}(r) +\deltao,~\forall~r \in [0, 1]\right) \\
&\overset{(c)}{\le} \PB_1\left( \frac{1}{n} \sum_{t=1}^n \1_{\rp_t \le G_{\bP_t}^{-1}(r) } \le G_{\bP_{\Delta}^{\star}}^{-1}(r) +\deltao,~\forall~r \in [0, 1]\right) \\
&\overset{(d)}{=} \PB_1\left( \frac{1}{n} \mathrm{N}(nr) \le G_{\bP_{\Delta}^{\star}}^{-1}(r) +\deltao,~\forall~r \in [0, 1]\right),
\end{align*}
where $(a)$ uses \eqref{eq:rate-help1}, $(b)$ uses Lemma \ref{lem:simplify-K}, $(c)$ uses \eqref{eq:inverse-G}, and $(d)$ uses the fact that $G_{\bP_t}(\rp_t) \overset{i.i.d.}{\sim} \UM(0, 1)$, as a result of which, $\frac{1}{n} \sum_{t=1}^n \1_{\rp_t \le G_{\bP_t}^{-1}(r) } = \frac{1}{n} \sum_{t=1}^n \1_{ G_{\bP_t}(\rp_t) \le r}$ has the same distribution of $\frac{1}{n} \mathrm{N}(nr)$ where $\mathrm{N}(t)$ is a Poisson process with parameter 1 under the condition that $\mathrm{N}(n)=n$.
We emphasize that the right-hand side of $(d)$ does not depend on the NTP distributions $\bP_t$'s.

We will apply Theorem 2 in \citep{nikitin1987hodges} to bound the right-hand side of $(d)$. To do this, we need to verify that \(G_{\bP_{\Delta}^{\star}}^{-1}(r)\) is a convex function of \(r\), which is a prerequisite of that Theorem 2. This follows easily from the fact that \(F_{1, \bP_{\Delta}^{\star}}(r)\) is a convex function of \(r\). Therefore, by Theorem 2 in \citep{nikitin1987hodges}, it follows that
\begin{align*}
\lim_{\epso \to 0} \lim_{n \to \infty}\frac{1}{n}\log \PB_1\left( \frac{1}{n} \mathrm{N}(nr) \le G_{\bP_{\Delta}^{\star}}^{-1}(r) +\deltao,~\forall~r \in [0, 1]\right) 
&\le - \int_0^1 [G_{\bP_{\Delta}^{\star}}^{-1}]'(r) \log [G_{\bP_{\Delta}^{\star}}^{-1}]'(r) \rd r\\
&= -\KL(\mu_0, (1-\eps)\mu_0 + \eps \mu_{1, \bP_{\Delta}^{\star}}),
\end{align*}
where $[G_{\bP_{\Delta}^{\star}}^{-1}]'$ is the derivative of $G_{\bP_{\Delta}^{\star}}^{-1}$. 
The proof technique used therein is mainly adapted from the classic works \citep{borovkov1967boundary,borokov1968asymptotically} that study large deviations in boundary-value problems.
As a result, we prove the upper bound:
\[
\limsup_{n \to \infty}\sup_{\bP_t \in \FPM}\frac{1}{n} \log \PB_1(S_n^+(s) \le \gamma_{n, \alpha}) \le -\KL(\mu_0, (1-\eps)\mu_0 + \eps \mu_{1, \bP_{\Delta}^{\star}}).
\]
For the lower bound, it essentially follows from the worst nature:
\begin{align*}
\liminf_{n \to \infty}\sup_{\bP_t \in \FPM}\frac{1}{n} \log \PB_1(S_n^+(s) \le \gamma_{n, \alpha}) 
&\ge \liminf_{n \to \infty}\sup_{\bP_t \in \{\bP_{\Delta}^{\star} \} }\frac{1}{n} \log \PB_1(S_n^+(s) \le \gamma_{n, \alpha}) \\
&\overset{(*)}{\ge}  -\KL(\mu_0, (1-\eps)\mu_0 + \eps \mu_{1, \bP_{\Delta}^{\star}}).
\end{align*}
We explain the last inequality $(*)$ as follows.
When $\bP_t \equiv \bP_{\Delta}^{\star}$, we are essentially dealing with an i.i.d. case where the alternative CDF of each $\Yars_t$ is $G_{\bP_{\Delta}^{\star}}$.
The inequality $(*)$ represents the lower bound for i.i.d. case efficiency, where our $\FPM$-efficiency reduces to the Hodges-Lehmann efficiency \citep{nikitin1987hodges}.
This inequality $(*)$ follows from the existing lower bound found in \citep{rao1962efficient}.

Combining these two parts, we complete the proof.
\end{proof}

\section{Details of Simulation Studies}
\label{sec:additional-simulation}

\subsection{Additional Histogram of \Algo~Statistics}
\label{sec:additional-result-GOT-nomask}

\begin{figure}[!th]
\centering
\includegraphics[width=1.0\textwidth]{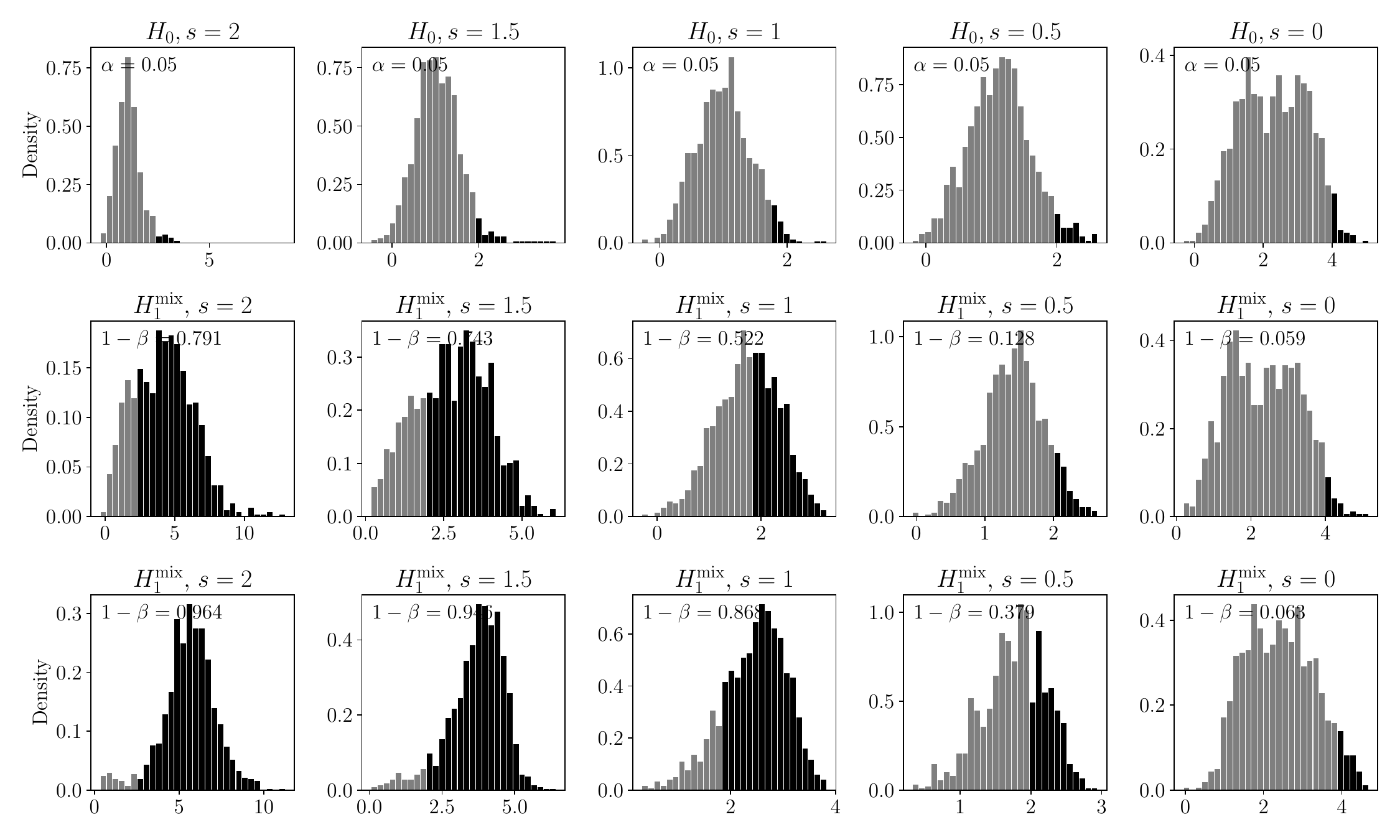}
\caption{Density histograms and powers of $\log(nS_n^+(s))$ for different values of $s$ with $c_n^+ = 0$ and $(p, q)=(0.2, 0.5)$.
The first row presents results under $H_0$, while the second and third rows display results under $H_1^{\mathrm{mix}}$, corresponding to the \textsf{M1} and \textsf{M2} settings, respectively. The dark area indicates the portion of the distribution that rejects $H_0$, which is the Type I error $\alpha$ under the null hypothesis $H_0$ and the power $1 - \beta$ under the alternative hypothesis $H_1^{\mathrm{mix}}$.
 }
\label{fig:other-s-hist-nomask}    
\end{figure}

In Figure \ref{fig:other-s-hist-nomask}, we present the density histograms and powers of $\log(nS_n^+(s))$ for different values of $s$ with $c_n^+ = 0$ and $(p, q)=(0.2, 0.5)$, following the same setting introduced in Section \ref{sec:simulation-main}. We observe that even with $c_n^+$ set to $0$, the patterns identified in Figure \ref{fig:other-s-hist} are still maintained, except that the support under $H_1^{\mathrm{mix}}$ is considerably larger due to the heavy-tailed behavior of $\rp_{(1)}$.

\subsection{Histograms of HC}
\label{sec:exp-hist}

\begin{figure}[!th]
\centering
\includegraphics[width=1.0\textwidth]{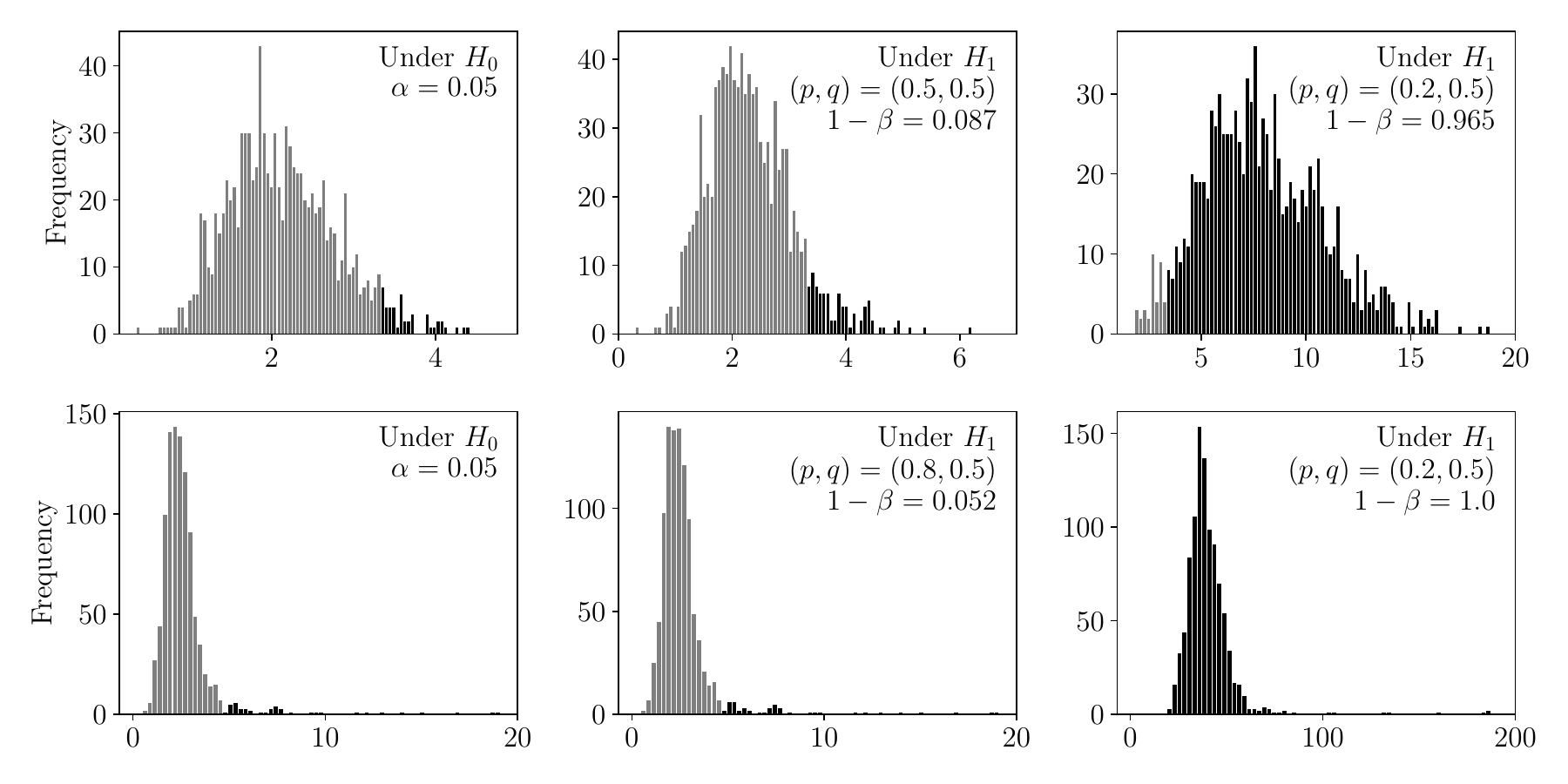}
\caption{Frequency histograms for $\HC_n^{+}$ with $c_n^+ = \frac{1}{n}$ (top) and $\HC_n^{\star}$ (bottom). The dark area indicates the portion of the distribution that rejects $H_0$ which is the the Type I error $\alpha$ under the null hypothesis $H_0$ and the power $1 - \beta$ under the alternative hypothesis $H_1^{\mathrm{mix}}$.
}
\label{fig:robust-hist}
\end{figure}

Though $\HC_n^{+}$ is connected to $S_n^+(s)$ via the square relation in \eqref{eq:relation-between-GoFT-HC}, their distribution might have slightly different shapes and ranges.
We perform an empirical study of the distribution of $\HC_n^{+}$ and $\HC_n^{\star}$ in this subsection.
Recall that $\HC_n^{\star}$ is the special instance of $\HC_n^{+}$ where $c_n^+=0$.
We set $\Delta_n = n^{-q}$ and $\eps_n = n^{-p}$, and use the following procedure to create samples from null and alternative.
\begin{enumerate}
\item Draw $n=10^4$ samples from $\UM(0, 1)$ to represent $H_0$ and then calculate $\HC_n^{+}$ or $\HC_n^{\star}$.
\item Replace $\ceil{n \eps_n}$ of the previous samples by the same number of samples from $F_{1,\bP_t}$ where $\bP_t$ is generated according to the \textsf{M3} method in which the top probability is forced to be $1-\Delta_n$. 
\item Repeat Steps 1 and 2 over $N=10^3$ times and make histograms of the simulated $\HC_n^{+}$ or $\HC_n^{\star}$.
\end{enumerate}

See Figure~\ref{fig:robust-hist} for the distribution of $\HC_n^{+}$ (top) and $\HC_n^{\star}$ (bottom) under interesting $(p, q)$ pairs.
Let's first focus on the top row.
Under $H_0$, the distribution of $\HC_n^{+}$ values is primarily concentrated around 2.
As suggested in~\citep{tony2011optimal}, we pick the $(1-\alpha)$-quantile of the empirical distribution of $\HC_n^{+}$ as the critical value, which is usually much more accurate than $\sqrt{(2+\delta)\log\log n}$.
This rejection region is marked in black in Figure~\ref{fig:robust-hist}, which corresponds to a Type I error rate of $\alpha$ under $H_0$ and a power of $1 - \beta$ under $H_1^{\mathrm{mix}}$.
Recall that our experiment setup confirms that $P_{t, \max}= 1-\Delta_n$ for all $t \in [n]$.
Under $H_1^{\mathrm{mix}}$ with $(p, q) = (0.5, 0.5)$, the distribution nudges slightly right, with a power of $1 - \beta = 0.087$. This slight shift aligns with Theorem~\ref{thm:gumbel-dense}, where the statistical indistinguishability of $H_1^{\mathrm{mix}}$ from $H_0$ results in a similar distribution of $\HC_n^{+}$ if $q + 2p > 1$.
Conversely, if $q + 2p < 1$, the distributions of $\HC_n^{+}$ under $H_1^{\mathrm{mix}}$ and $H_0$ diverge.
As a result, under $H_1^{\mathrm{mix}}$ with $(p, q) = (0.2, 0.5)$, $\HC_n^{+}$ centers around 8 instead and achieves a higher power of $1 - \beta = 0.965$.

The bottom row of Figure~\ref{fig:robust-hist} shows that $\HC_n^{\star}$ tends to exhibit significantly large values particularly when $2p+q<1$. 
This extreme value could potentially lead to the outlier issue.
To illustrate, under $H_0$, the largest value observed for $\HC_n^{+}$ typically hovers around 4, whereas for $\HC_n^{\star}$, it can exceed 15. 
This observation aligns with the theoretical analysis in Section 3 \citep{donoho2004higher} that shows $\HC_n^{\star}$ has ``heavy tails'' under $H_0$.
This heavy-tail issue is more severe under $H_1^{\mathrm{mix}}$: the largest value of $\HC_n^{\star}$ could exceed 180, while the counterpart of $\HC_n^{+}$ is merely around 18.
This disparity raises a concern: it becomes challenging to determine whether a large value of $\HC_n^{\star}$ is attributable to its tendency towards extreme values (due to the heavy-tail issue), or a strong indication of embedded watermarks.
$\HC_n^{+}$ is proposed to mitigate this issue and exhibits less heavy-tailed performance.

\subsection{Detection Boundaries for HC}
\label{sec:HC-boundary}
Generally, the decision boundary for HC methods should align with that of $n S_n^+(2)$ due to the relation in \eqref{eq:relation-between-GoFT-HC}. 
For completeness, we will include numerical illustrations for HC methods after presenting the corresponding results for $n S_n^+(s)$ in the main text.

\begin{figure}[!th]
\centering
\includegraphics[width=1.0\textwidth]{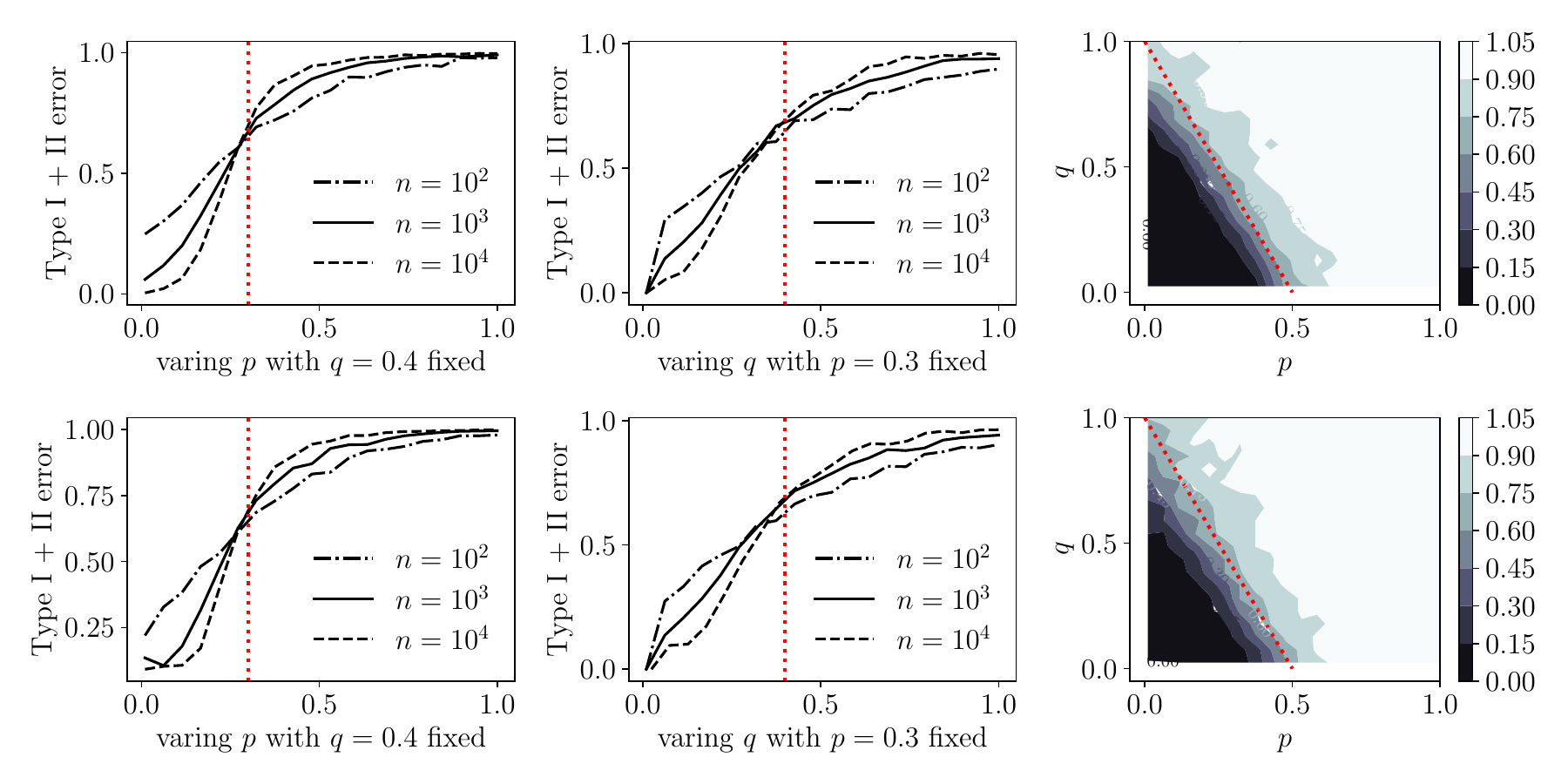}
\caption{The smallest sum of Type I and Type II errors of $\HC_n^{+}$ with $c_n^+ = \frac{1}{n}$ (top) and $\HC_n^{\star}$ (bottom). The red dotted line indicates the theoretical boundary. Each point is obtained by averaging $N=10^3$ independent experiments.
}
\label{fig:empirical-error}
\end{figure}

We aim to verify the correctness of Theorem \ref{thm:adaptivity}.
This theory implies that $\HC_n^{+}$ is of $O_{\PB}(\sqrt{\log\log n})$ under $H_0$ and is of $\Omega_{\PB}(n^{\frac{1}{2}-p-\frac{q}{2}})$ under $H_1^{\mathrm{mix}}$.
Hence, detectability requires increasingly large samples as one approaches the detection boundary $q+2p=1$. 
We investigate the detection boundary by checking the smallest sum of Type I and Type II errors in the following.
Following the approach in \citep{tony2011optimal}, we tune the parameter $\delta$ as the optimal value in the set $\{0, 0.2, 0.4, \ldots, 3.8, 4\}$ that results in the smallest sum of Type I and Type II errors.
We use the same setup introduced in Section \ref{sec:simulation-main} where the vocabulary size is $|\Voca|=10^3$ and each NTP distribution $\bP_t$ is $\bQ_{\Delta_n}$ which is a uniform distribution, after ensuring that the highest probability is $1-\Delta_n$.
Again, we replace pseudorandomness with true randomness for illustration purposes.

The results are displayed in Figure \ref{fig:empirical-error}. All observations noted in Figure \ref{fig:empirical-error-others} remain applicable here as well.
These results not only substantiate the accuracy of our theory but also provide empirical evidence supporting the claim: $\HC_n^{+}$ has the same asymptotic behavior as $\HC_n^{\star}$, making them an identical detection boundary. 
As discussed in Remark~\ref{rem:proof-HCstar}, Theorem~\ref{thm:adaptivity} remains true even if we replace $\HC_n^{+}$ with $\HC_n^{\star}$. 


\section{Details of Language Model Experiments}
\label{sec:LLM}

\subsection{Details of Experiment Setup}
\label{exp:detials}
We employ a context window of size $m=5$, allowing the randomness variable $\xi_t = \AM(s_{(t-m):(t-1)}, \Key)$ to depend on the previous $m$ tokens. We utilize the hash function $\AM$, as used in \citep{zhao2024permute}, for generating this randomness. 
We set different values of $c_n^+$ for the \Algo~test in different experiments. For the experiments in Section \ref{sec:statistical-power}, we use $c_n^+ = \frac{1}{n}$. The specific values of $c_n^+$ for other experiments will be provided in the following implementation details. In all experiments, we apply a watermark to a token only if the current text window is unique within the generation history, aiming to prevent repetitive generation \cite{hu2023unbiased, wu2023dipmark, Dathathri2024}. When no watermark is applied, we use multinomial sampling from the NTP distribution with the temperature set to 0.7.

\paragraph{How to generate prompts.}
The experimental setup we employed is largely based on the methodology described in Appendix D.1 of \citep{kuditipudi2023robust}. 
In our approach, each generation is conditioned on a prompt which is obtained by sampling documents from the news-oriented segment of the C4 dataset \citep{raffel2020exploring}.
We enforce a minimum prompt size of 50 tokens in all experiments and skip over any document that is not long enough.
Note that retokenization may not precisely match the original tokens. Therefore, to guarantee that the verifier consistently receives at least $n$ tokens, we augment its input with special padding tokens, which vary according to the tokenizer of each model. Additionally, to mitigate the need for padding, we initially generate many buffer tokens beyond $n$. 
We set the number of buffer tokens to be 20 in every experiment.
This additional buffer typically makes padding unnecessary.

\paragraph{Computation of critical values.}
Critical values are used to control the Type I error.
For sum-based detection rules, we estimate the critical value by
\[
\hat{\gamma}_{n,\alpha} = n \cdot \EB_0 h (Y) +  \Phi^{-1}(1-\alpha) \cdot \sqrt{n \cdot \Var_0(h(Y))}.
\]
It is easy to see that $\PB_0(\sum_{t=1}^n h(\Yars_t) \ge \hat{\gamma}_{n,\alpha}  ) \to \alpha$ due to the central limit theorem.
To determine the critical value for HC and \Algo, we resort to simulation.
For a given text length $n$ and $s \in [-1 ,2]$, we generate $n$ i.i.d. copies of $\Yars_t$ from $\UM(0, 1)$ and calculate the corresponding statistic $\log(nS_n^+(s))$.
This procedure is replicated $10^4$ times, using the empirical $1-\alpha$ quantile of these $10^4$ samples as an initial estimate. 
To enhance the precision of this estimate, we repeat the process 10 times and average these 10 initial estimates to establish the final critical value.

\paragraph{Details of the random edits.}
In general, we set the text length to $n = 400$ to fully utilize the observed data. However, when the detection problem becomes easier, $n = 400$ may be excessive, leading to near-zero Type II errors for most detection methods. In such cases, we reduce the length for better visualization. In the substitution and insertion experiments, $n$ is set to 400 for temperatures of 0.1 and 0.3, and reduced to 200 for temperatures of 0.7 and 1. For the deletion experiment, we initially generate more than $n$ watermarked tokens to ensure that, even after deleting a fraction, at least $n$ tokens remain. In this case, $n$ is set to 200.
For the \Algo~test, we select the best stability parameter $c_n^+$ from the set $\{0, 10^{-3}, \frac{1}{n}\}$ (with $n$ varying accordingly).

\paragraph{Details of Figures \ref{fig:intro} and \ref{fig:intro-better-result}}
In Figures \ref{fig:intro} and \ref{fig:intro-better-result}, for the paraphrase edit, we use random synonym substitution to edit watermarked texts and evaluate Type II error as a function of text length, with a temperature parameter of 0.3. We select 1000 prompts from the C4 news-like dataset as before and randomly replace $5\%$ of words. For each selected word, synonyms with multiple alternatives are retrieved from WordNet \citep{miller1995wordnet} to ensure meaningful substitutions. The modified text is then reconstructed with each target word replaced by a randomly chosen synonym.

In the adversarial edit, we first compute the pivotal statistics for all tokens, then identify the top 5\% with the largest values, replacing these tokens with uniformly selected alternatives. This type of edit assumes that the human editor has knowledge of the hash function $\AM$ and the secret key $\Key$, which generally results in the removal of more watermark signals.

For the \Algo~test, we set $s=1.5$ and $c_n^+ = \frac{1}{n}$. For $\hoptarso$, we set $\Delta_0 = 0.1$.

\paragraph{Details of the edit tolerance limit.}
We prompt ChatGPT-4o to generate 100 popular poems along with their authors and ask the target model (either OPT-1.3B or Sheared-LLaMA-2.7B) to perform either the poem recitation or poem generation task. Taking the poem \textsf{Adonais by Percy Bysshe Shelley} as an example, the prompt for poem recitation is: \textsf{Please recite the poem: Adonais by Percy Bysshe Shelley}. For poem generation, the prompt is: \textsf{Please write a new poem in the style of this one: Adonais by Percy Bysshe Shelley}.
The temperature is set to 1 for both the poem recitation and generation tasks.

We use binary search (Algorithm \ref{algo:binary}) to determine the edit tolerance limit. 
We set the initial length to $\tn = 400$. 
We pass only the first $n$ tokens in the edited text to our verifier.
For random substitution and insertion, the test length is $n = 200$, while for random deletion, $n = 100$.
All critical values are computed with a Type I error rate of $\alpha=0.01$. For simplicity, the stability parameter $c_n^+$ in the \Algo~test is set to $10^{-3}$.

\begin{algorithm}[t]
\caption{Binary search to compute the edit tolerance limit}
\begin{algorithmic}[1]
\State \textbf{Input:} A watermarked sentence $\ttoken_{1:\tn}$ and a specified edit.
\State \textbf{Initial:} Set $l = 1$, $u = \tn$, and generate a random permutation $\pi_0$ over the vocabulary $\Voca$.
\While{$u - l \ge 2$}
\State Compute the middle point $m = \left\lfloor\frac{u+l}{2}\right\rfloor$.
\State Apply the considered edit type to corrupt the watermarked tokens $\ttoken_{\pi(1)},\ldots, \ttoken_{\pi(m)}$.
\State Denote the resulting edited text by $\token_{1:n_1}^{(m)}$.
\State Pass the first $n$ tokens of $\token_{1:n_1}^{(m)}$ to the detection method.
\If{the detection method rejects the null hypothesis $H_0$}
\State Set $l \gets m$
\Else
\State Set $u \gets m$
\EndIf
\EndWhile
\State \textbf{Output:} The edit tolerance limit is $\frac{m}{\tn} \times 100\%$.
\end{algorithmic}
\label{algo:binary}
\end{algorithm}

\paragraph{Details of the roundtrip translation.}
The roundtrip translation, as described in Appendix D.2 of \citep{kuditipudi2023robust}, involves translating text from English to French and back to English using the OPUS-MT series of translation models. These models are available on Huggingface Hub.\footnote{https://huggingface.co/.} The specific models used for this attack are:
\begin{enumerate}
    \item \textsf{Helsinki-NLP/opus-mt-tc-big-en-fr} for English to French translation,
    \item \textsf{Helsinki-NLP/opus-mt-tc-big-fr-en} for French to English translation.
\end{enumerate}
This method leverages the subtle nuances of translation to detect inconsistencies or vulnerabilities in language models.
Since the text length may change after roundtrip translation, we use the last 200 tokens from each (edited) sentence as the input to our verifier. If a sentence is shorter than 200 tokens, we pad it with zeros at the beginning to ensure a consistent length of 200.
We set $c_n^+ = \frac{1}{n}$ for the \Algo~test.

\subsection{Additional Results on Statistical Power}
\label{appen:extened-results}


Figure \ref{fig:type-I-and-II-errors-2p7} presents results analogous to those shown in Figure \ref{fig:type-I-and-II-errors}, but with the use of a larger model, Sheared-LLaMA-2.7B \citep{xia2023sheared}. All observations noted in Section \ref{sec:statistical-power} remain valid: \Algo~performs exceptionally well at low temperatures and achieve performance comparable to the practical $\hars$ at high temperatures.

\begin{figure}[htbp]
\centering
\includegraphics[width=\textwidth]{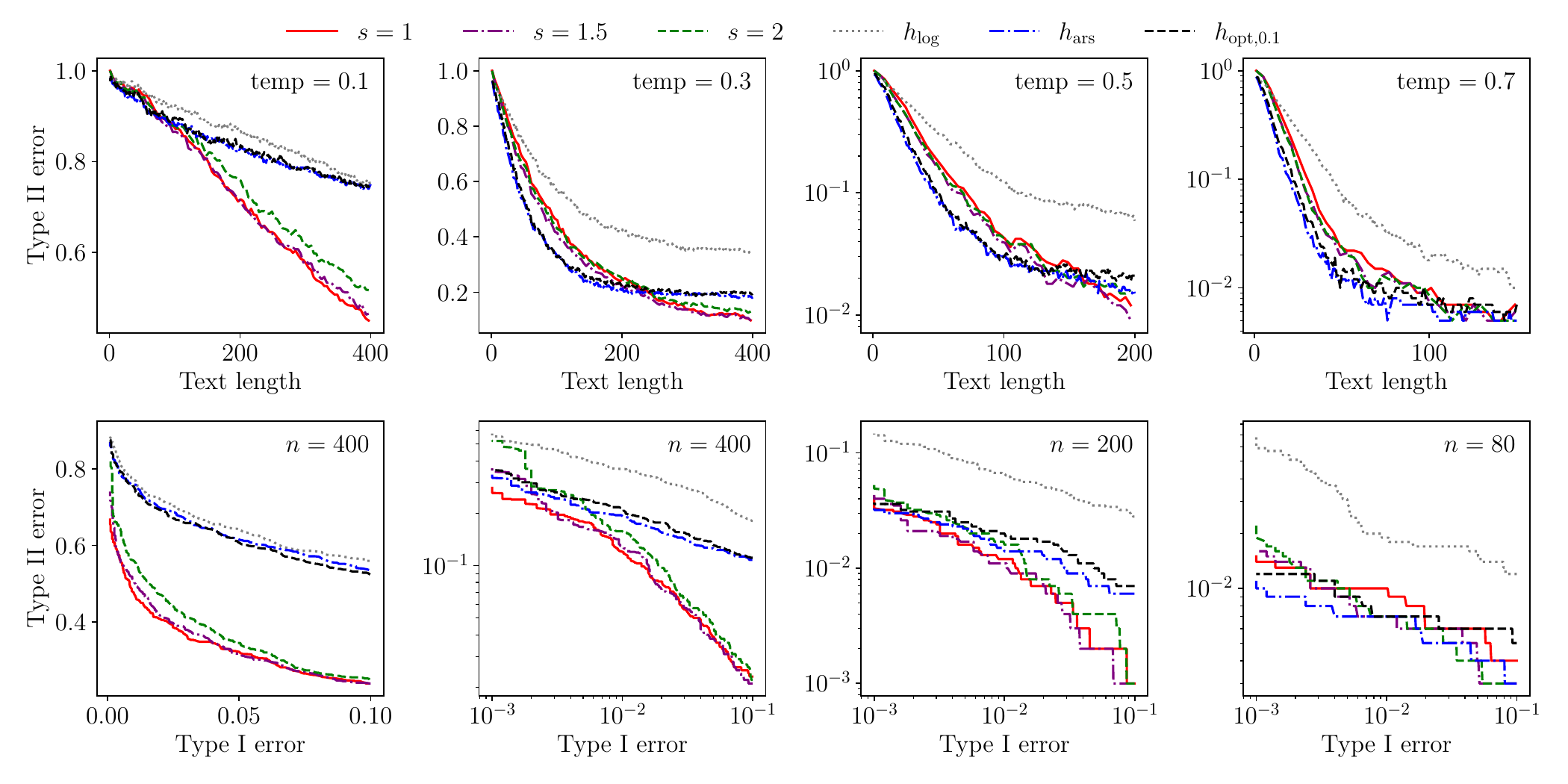}
\caption{Empirical Type II errors (top row) on the C4 dataset across different detection rules applied to the Gumbel-max watermark. Here we use Sheared-LLaMA-2.7B \citep{xia2023sheared}. The bottom row illustrates the trade-off function for a specific text length $n$. The temperatures used, from left to right columns, are 0.1, 0.3, 0.5, and 0.7, respectively.
}
\label{fig:type-I-and-II-errors-2p7}
\end{figure}

\newpage
\subsection{Additional Results on Robustness}
\label{appen:extened-robust-results}
Figure \ref{fig:modification-2p7} presents results analogous to those shown in Figure \ref{fig:modification-nomask}, but with the use of a larger model, Sheared-LLaMA-2.7B \citep{xia2023sheared}. 
Similarly, Figure \ref{fig:translation-2p7} presents results analogous to those shown in Figure \ref{fig:translation-best}, and Figure \ref{tab:largest-fraction-2.7B} presents results analogous to those shown in Figure \ref{tab:largest-fraction-1.3B} with this larger model.
Figure \ref{fig:adversarial-all} provides the complete results of adversarial edits on OPT-1.3B, expanding on Figure \ref{fig:adversarial-best}.
Corresponding results for the Sheared-LLaMA-2.7B model are shown in Figure \ref{fig:adversarial-all-2p7B}.

All observations from Section \ref{sec:robust-evaluation} remain consistent: (i) \Algo~performs exceptionally well at low temperatures and shows comparable performance to the practical $\hars$ at higher temperatures, and (ii) they exhibit the largest edit tolerance limit in most cases among all the detection methods considered.

\begin{figure}[htbp]
\centering
\includegraphics[width=\textwidth]{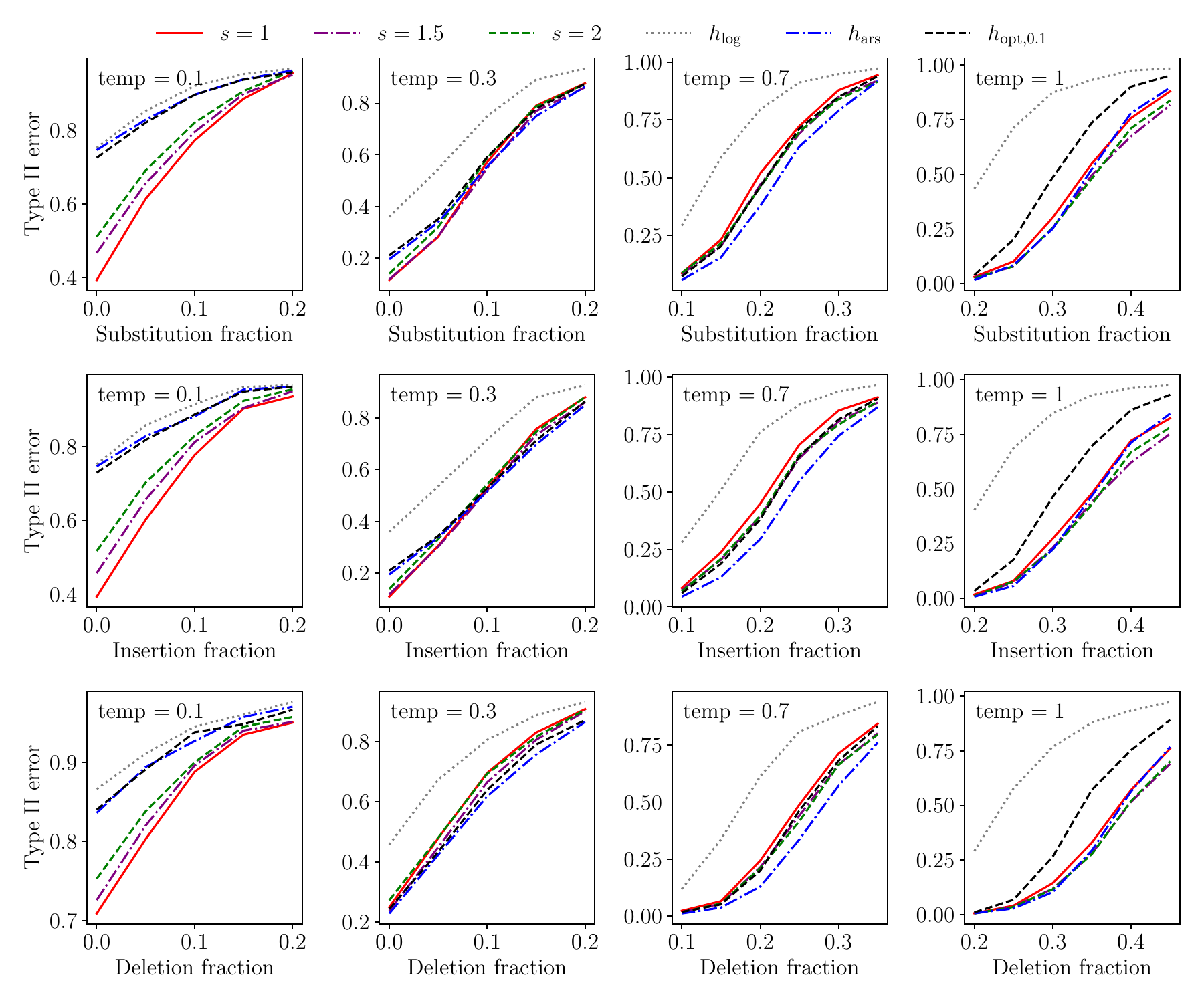}
\caption{
Effect of three random edits on Type II error across different temperatures at $\alpha=0.01$ on the Sheared-LLaMA-2.7B model. The top, middle, and bottom plots correspond to random substitution, insertion, and deletion, respectively.
}
\label{fig:modification-2p7}
\end{figure}

\begin{table}
\centering
\resizebox{\textwidth}{!}{ 
\begin{tabular}{c|c|cccccccc}
\toprule
Task & {Edit types}  & $s=1$ & $s=1.5$& $s=2$  & $\hlog$ & $\hars$ & $h_{\mathrm{opt}, 0.3}$ & $h_{\mathrm{opt}, 0.2}$& $h_{\mathrm{opt}, 0.1}$ \\
\midrule
\multirow{3}{*}{\shortstack{Poem \\Recitation}} 
& Substitution & 36.22&38.49&\textbf{38.76}&22.39&37.08&26.1&28.43&32.03\\
\cmidrule{2-10}
& Insertion &38.49&40.77&\textbf{40.83}&24.93&40.45&28.92&31.79&36.01\\
\cmidrule{2-10}
& Deletion & 35.85&38.66&\textbf{39.43}&22.66&37.73&25.2&28.38&32.34\\
\midrule
\multirow{3}{*}{\shortstack{Poem \\Generation}} 
& Substitution & 36.06&38.15&\textbf{38.81}&22.0&36.88&26.07&28.85&31.94\\
\cmidrule{2-10}
& Insertion & 38.7&40.86&\textbf{41.75}&24.94&40.22&28.35&31.47&35.03\\
\cmidrule{2-10}
& Deletion &39.7&\textbf{42.12}&41.83&21.93&41.08&26.6&30.8&35.23\\
\bottomrule
\end{tabular}}
\caption{The edit tolerance limits $(\%)$ on the Sheared-LLaMA-2.7B model.}
\label{tab:largest-fraction-2.7B}
\end{table}

\begin{figure}[htbp]
\centering
\includegraphics[width=\textwidth]{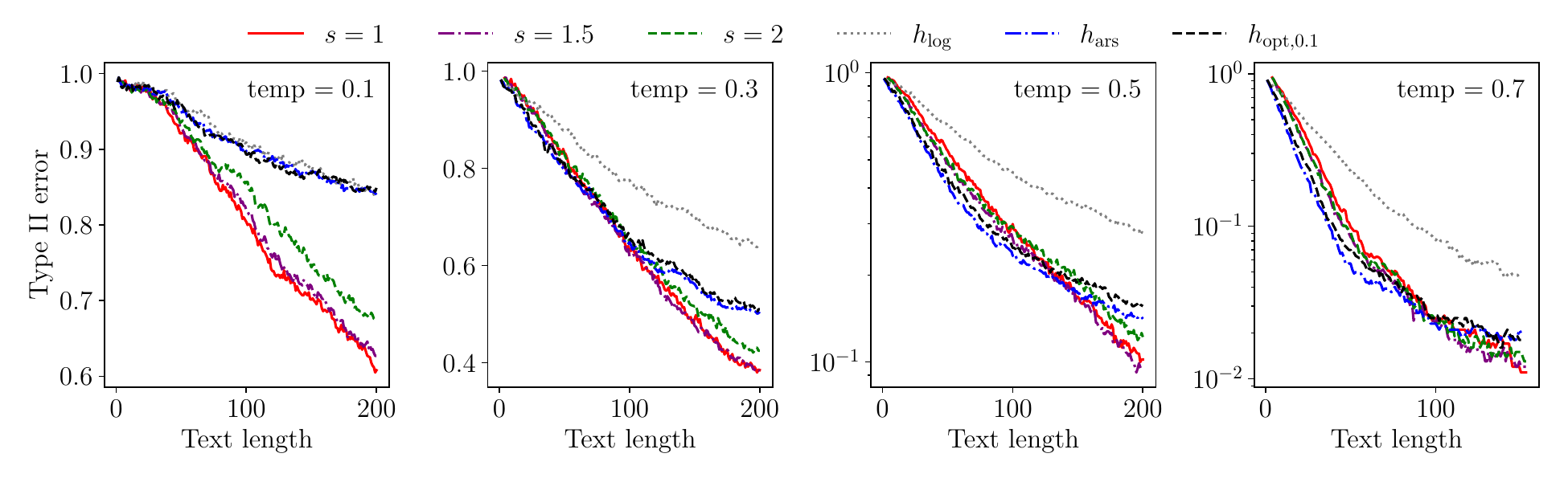}
\caption{Empirical Type II errors on Sheared-LLaMA-2.7B across different lengths of edited texts at different temperature parameters with a fixed Type I error of $\alpha=0.01$.
}
\label{fig:translation-2p7}
\end{figure}

\begin{figure}[ht]
\centering
\includegraphics[width=\textwidth]{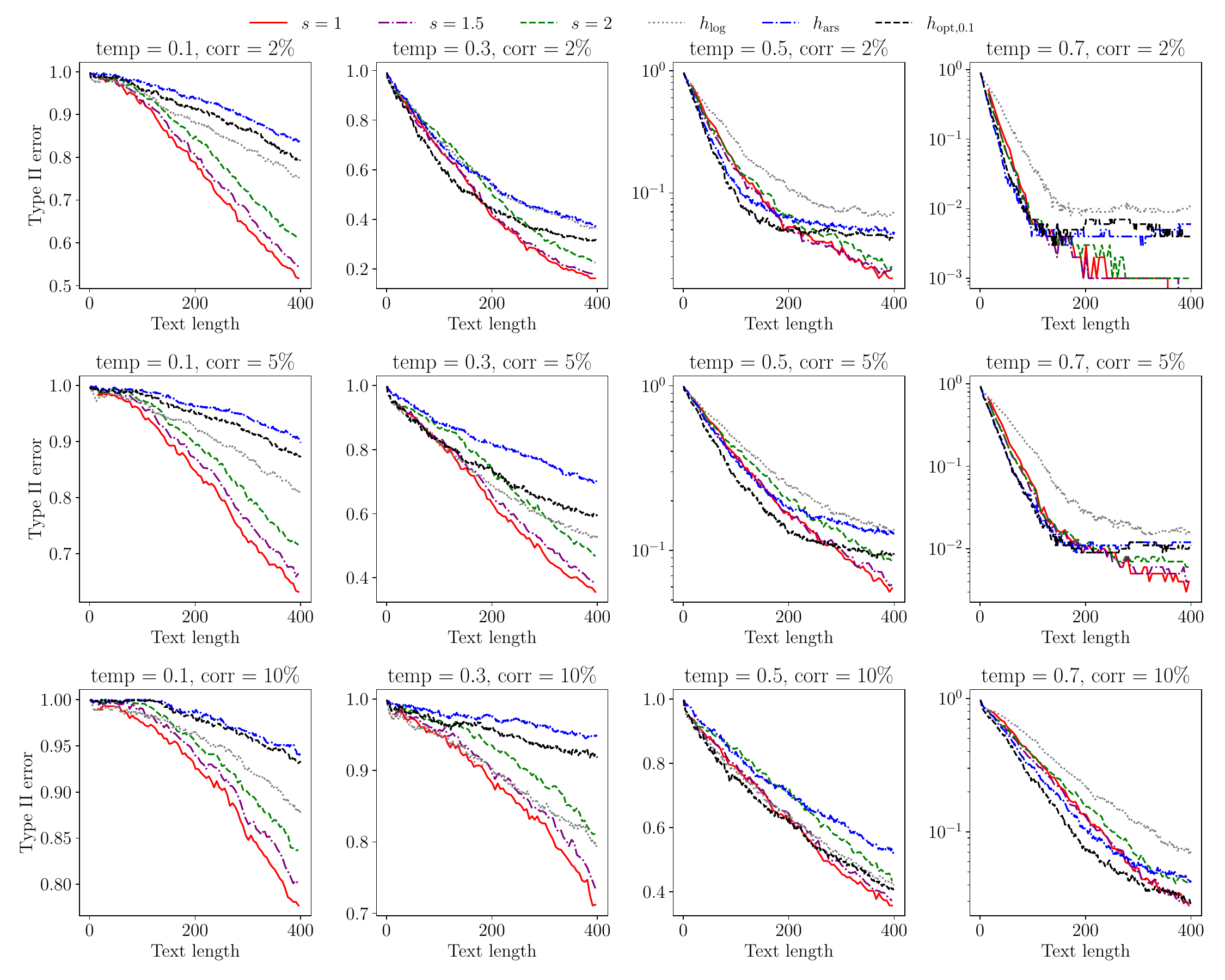}
\caption{
Complete results of Type II error under adversarial edits across various temperatures and edit fractions at $\alpha=0.01$ on the OPT-1.3B model. Columns represent four temperatures: $\{0.1, 0.3, 0.5, 0.7\}$, and rows correspond to four edit fractions: $\{2\%, 5\%, 10\%\}$.
}
\label{fig:adversarial-all}
\end{figure}

\begin{figure}[ht]
\centering
\includegraphics[width=\textwidth]{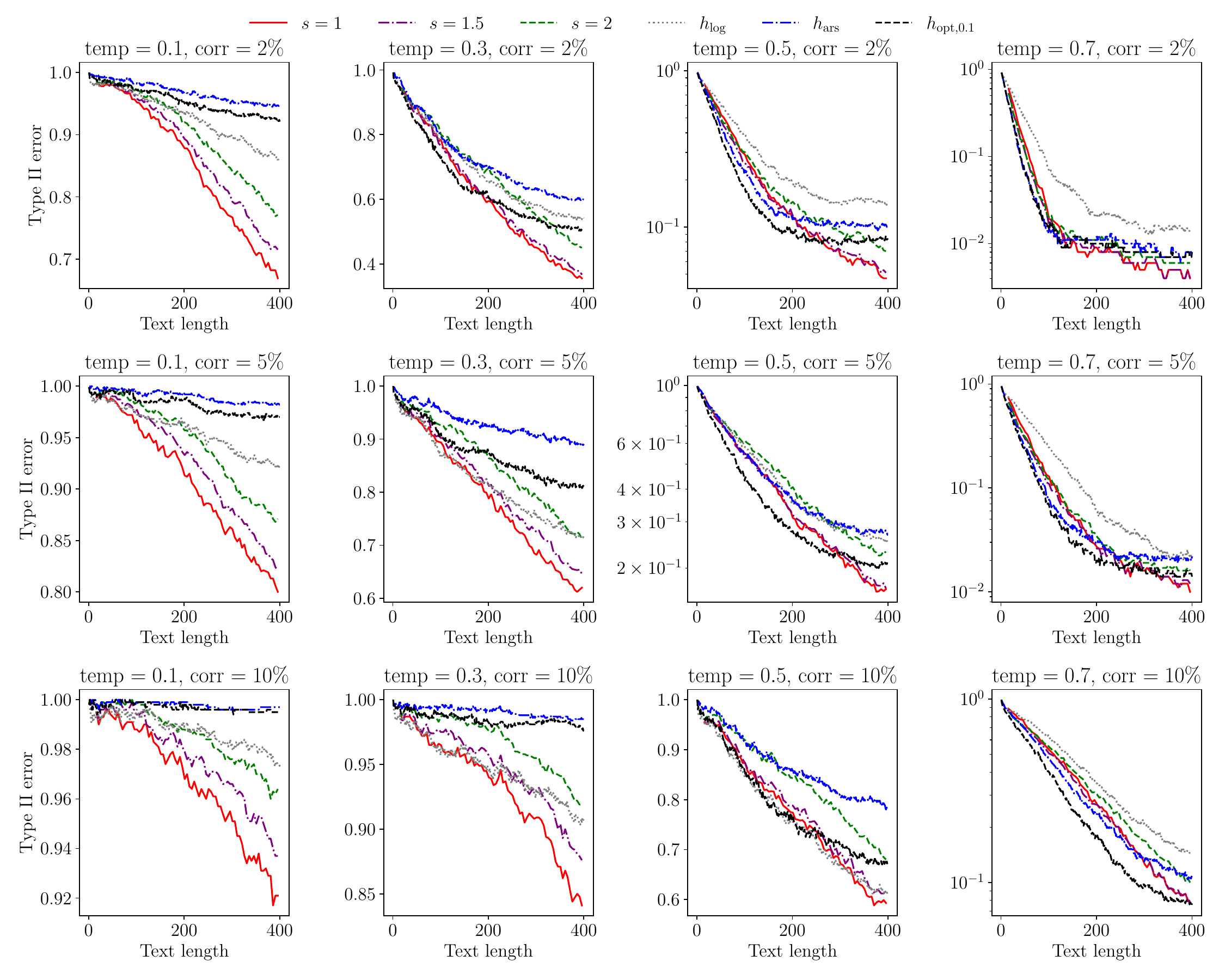}
\caption{
Complete results of Type II error under adversarial edits across various temperatures and edit fractions at $\alpha=0.01$ on the Sheared-LLaMA-2.7B model. Columns represent four temperatures: $\{0.1, 0.3, 0.5, 0.7\}$, and rows correspond to four edit fractions: $\{2\%, 5\%, 10\% \}$.
}
\label{fig:adversarial-all-2p7B}
\end{figure}

\end{appendix}

\end{document}